\documentclass[a4paper,11pt]{article} 
\RequirePackage[top=1.8cm, bottom=1.6cm, left=1.4cm, right=1.4cm]{geometry}
\newcommand{\JEL}[1]{\par\noindent\textbf{JEL Classification:} #1}
\usepackage{amsmath,amssymb,amsthm}
\usepackage{booktabs}
\usepackage{bm}
\usepackage[colorlinks=true,bookmarks=false,citecolor=blue,urlcolor=blue]{hyperref} %pdflatex
\usepackage{natbib}
\usepackage{tikz}
\usepackage{float}
\usepackage{subcaption}
\usepackage{setspace}
\usetikzlibrary{arrows.meta,positioning,calc,shapes.geometric}
\newtheorem{theorem}{Theorem}[section]

\newtheorem{proposition}[theorem]{Proposition}
\newtheorem{lemma}[theorem]{Lemma}

\newtheorem{definition}[theorem]{Definition}

\newtheorem{remark}[theorem]{Remark}
\newtheorem{assumption}[theorem]{Assumption}

\DeclareMathOperator*{\argmax}{arg\,max}

\allowdisplaybreaks
\title{Nash Peer-to-Peer Insurance Bargaining\\
under Price Fairness and Coalitional Stability}
\author{Tim J. Boonen\thanks{Department of Statistics and Actuarial Science, School of Computing and Data Science, The University of Hong Kong, Pokfulam, Hong Kong; Email: \texttt{tjboonen@hku.hk}}\hspace{1em} Wing Fung Chong\thanks{Department of Statistics and Actuarial Science, School of Computing and Data Science, The University of Hong Kong, Pokfulam, Hong Kong; Email: \texttt{chongwf@hku.hk}}\hspace{1em} Kenneth Tsz Hin Ng\thanks{Department of Mathematics, The Ohio State University, Columbus, Ohio 43210,  United States; Email: \texttt{ng.499@osu.edu} }\hspace{1em} Tak Wa Ng\thanks{Department of Statistics and Actuarial Science, School of Computing and Data Science, The University of Hong Kong, Pokfulam, Hong Kong; Email: \texttt{takwang@hku.hk}}}

\date{\today}

\begin{document}
\maketitle
\begin{abstract}
    We study peer-to-peer (P2P) insurance contracting between a risk-averse P2P reinsurer and multiple risk-averse peers in an asymmetric Nash-bargaining framework, where all agents seek to improve expected utility relative to their disagreement points. 
    Consistent with the expected value premium principle, we impose a price-fairness condition requiring each peer's expected contribution to be based on a common loading applied to the peer's expected loss. To justify the bargaining formulation relative to a standard fixed-weight weighted-sum optimization problem, we provide an axiomatic characterization showing that the Nash bargaining solution satisfies properties well-suited to voluntary P2P insurance contracting in small pools. We establish the existence and uniqueness of the optimal contract and derive first-order characterizations for the full-, partial-, and zero-reinsurance regimes.  
    To address subgroup formation, we develop computationally tractable sufficient conditions that rule out viable coalitional deviations, both with and without price fairness. 
    Our numerical study investigates the impact of price fairness and pool size on the optimal contract and agents' welfare. Price fairness reduces dispersion in risk allocations and certainty-equivalent loadings among peers. Regarding pool size, welfare need not increase monotonically, highlighting that risk-pool expansion depends not only on diversification but also on the evolution of bargaining power.
\end{abstract}

\noindent\textbf{Keywords:} 
Peer-to-peer insurance, 
Nash bargaining,
Optimal reinsurance, 
Pareto optimality,
Price fairness, Coalitional stability.

\JEL{C71, C78, D86, G22.}

\section{Introduction}

%\subsection{Context and Motivation}

Recent years have witnessed a resurgence of peer-to-peer (P2P) insurance, which revives the traditional principle of mutual assistance through modern digital platforms and risk-sharing technologies. Rather than relying on the transfer of risk to a conventional centralized insurer (see, e.g., \cite{boonen2024pareto,EscobarChong2026}), P2P insurance allows participants to pool and share losses within a community, potentially improving transparency, aligning incentives, and strengthening solidarity among members. 

This paper develops a practical model for P2P risk-sharing arrangements with relatively small pools, which are commonly observed in practice. For example, Friendsurance\footnote{\href{https://www.friendsurancebusiness.com}{https://www.friendsurancebusiness.com}} in Germany has organized policyholders into groups of around ten, whereas the Dutch Broodfonds\footnote{\href{https://www.broodfonds.nl}{https://www.broodfonds.nl}} model typically organizes groups of 20 to 50 self-employed participants. In such small pools, all members and the P2P reinsurer play a material role in shaping the contract:
% differences in risk exposure, risk preferences, outside insurance opportunities, and negotiating positions may substantially affect the premium decision and risk-sharing design. Moreover, as
participation is voluntary, the allocation of risks and payments is more naturally viewed as the bargaining outcome among peers and the P2P reinsurer rather than as a decision imposed by a social planner. 
% These features motivate our use of a Nash bargaining model, as it captures the outcome of the bargaining process: \takwacomment{a cooperative and efficient interaction to compete for welfare allocation,} in which each agent aims at maximizing the relative welfare gains over the disagreement points. 
These features motivate our use of a Nash bargaining model, as it captures the dual nature of the bargaining process: a cooperative drive for efficiency paired with a competitive allocation of surplus, where agents maximize their respective welfare gains relative to their disagreement points.
This contrasts with weighted-sum Pareto optimization, which does not explicitly incorporate status quo utilities into the objective function. This is particularly relevant in P2P insurance, where participation is voluntary, and agents may opt for outside insurance instead, so that the attractiveness of the outside option influences not only participation decisions but also the distribution of welfare gains and the resulting contract. 

In the Nash P2P insurance bargaining model, 
% \takwacomment{Tim used ``Nash insurance bargaining" so I think the previous version should be fine} \KTHNcomment{Okay. I was thinking Nash bargaining should come together}
we consider a risk-averse P2P reinsurer
and multiple risk-averse peers, where all agents aim to
enhance their expected utility relative to their disagreement points. In particular, the peers' disagreement points are induced by an alternative insurance choice provided by an exogenous centralized insurer, while the P2P reinsurer's disagreement point is the utility obtained from not offering any insurance service. Under such an arrangement, each peer brings their loss exposure to the P2P insurance pool. The participating peers and the P2P reinsurer jointly absorb the aggregate loss, with each party bearing a predetermined proportion of the total loss. In return for assuming part of the aggregate risk and providing the platform, the P2P reinsurer receives premium payments from the peers. In sum, participation in the P2P arrangement allows peers to replace their individual loss exposure with a share of the aggregate loss of the pool, along with a premium payment. 

To ensure fair pricing across heterogeneous peers while limiting cross-subsidization based on expected losses, we introduce a price-fairness condition consistent with the expected value premium principle. This ensures that peers are charged in proportion to their individual mean losses, adjusted by a uniform loading factor, so that no peer is systematically overcharged or undercharged relative to their risk exposure. Notably, this pricing rule allows us to isolate the role of bargaining power in risk allocation. 

\subsection{Contributions}

%Moreover, we enforce a price-fairness condition to align all peers’ expected risk-bearing and side payments as if they pay a uniform loading with the expected value premium principle. [\KTHNcomment{MORE ON WHY WE INTRODUCE PRICE FAIRNESS}]

Our contributions are manifold. 
From the perspective of bargaining theory, we show that the P2P insurance contracting problem fits naturally into the multi-agent asymmetric Nash-bargaining framework in \cite{kalai1977nonsymmetric}. Specifically, we show that the resulting P2P insurance contract is characterized by four properties: (i) Pareto optimality, (ii) strict improvement over the disagreement points, (iii) invariance to positive affine transformations of utility, and (iv) independence of irrelevant alternatives.  Compared with a weighted-sum formulation, the Nash bargaining model explicitly incorporates disagreement points, since Property (ii) reflects agents’ outside insurance options in P2P insurance contracting. In addition, it is invariant to affine transformations of utility, meaning that it does not rely on interpersonally comparable utility units and thus focuses on relative gains over disagreement points.

% \KTHNcomment{[Any mathematical challenge or something in the proof we may slightly elaborate as contributions]} \takwacomment{Basically, the technique is the same, just translating to the current context.} \KTHNcomment{We can elaborate it a bit on the justification of the use of Nash bargaining here, say, by combining.moving the paragraph above here.}

On the methodological side, to solve the optimal P2P insurance contracting problem, we first characterize the admissible set of contracts. This is achieved by introducing lower and upper bounds on the aggregate premium depending on the risk allocation strategy, which correspond to the individual rationality (IR) constraints of the P2P reinsurer and peers, respectively. We show that the admissible set of contracts is convex and compact, and that the subset satisfying all agents' IR constraints strictly is non-empty. This yields a well-posed optimization problem that is amenable to convex analysis. In Theorem \ref{Thm:EUsol}, we establish the existence and uniqueness of the Nash bargaining contract and characterize it using systems of nonlinear first-order conditions under three regimes: full, partial, and zero reinsurance. Economically, these systems balance the agents' weighted marginal welfare effects, thereby characterizing the trade-off between risk-sharing efficiency and the distribution of contractual surplus. 

% \takwacomment{Examining the solution form, we show that, even under exponential utility, bargaining weights have a nontrivial impact on the resulting risk-sharing arrangement. This contrasts with much of the existing literature \takwacomment{(see, e.g., \cite{aase2009nash}),} where bargaining power typically affects only monetary transfers or premiums, while leaving the underlying risk-sharing structure largely unchanged.}
% \KTHNcomment{Seems this is not on the methodological side and we only show it numerically. What about saying Nash bargaining ``captures a form of cooperative interaction with competitive surplus allocation" }
% \takwacomment{I dissolved this passage (commented out) below. ``captures a form of cooperative interaction with competitive surplus allocation" this sounds other way to explain what bargaining is, which is a good line to add at the begining.}
% \KTHNcomment{The associated Karush--Kuhn--Tucker (KKT) conditions further distinguish between interior and boundary solutions, thereby identifying the regimes of partial, full, and zero reinsurance.}
%\KTHNcomment{The system, which is determined by the} \KTHNcomment{Karush–Kuhn–Tucker} (KKT) conditions, \KTHNcomment{illustrate} how the allocation balances the agents’ weighted marginal welfare effects under partial, full, and zero reinsurance. The first-order conditions , while the boundary cases are determined by the corresponding . 

Regarding the stability of such a P2P arrangement against subgroup formation,
% we address the issue of possible subgroup formation. \KTHNcomment{To this end,} 
we develop an ex post stability test that provides a computationally tractable, subgroup-independent condition under which no smaller pool can offer its members and the P2P reinsurer a weakly better contract. 
In Theorem \ref{Thm:core}, under the assumption that the P2P reinsurer must be included in any subgroup, we provide conditions under which no subgroup can improve upon the grand-coalition outcome.

%The result is summarized in 

%\takwacomment{in which the main assumption is to include the P2P reinsurer in each subgroup so that} 
%Then, the theorem 

% which in particular demonstrates that no peer has an incentive to deviate from the P2P arrangement when the alternative insurance option is sufficiently inexpensive. Under this condition, peers in any smaller coalition would prefer the outside insurance option to forming a separate P2P pool.

We further obtain tractable characterizations in two benchmark cases under exponential utility. First, when the price-fairness condition is removed, the optimal risk-sharing rule is determined by agents' risk aversions and the deadweight cost of reinsurance, while bargaining powers and the centralized insurer’s loading affect only the premium and side payments. In the zero-deadweight-loss case $\tau=0$, this rule reduces to the familiar inverse-risk-aversion allocation. Second, under price fairness with homogeneous peers, the peers optimally split the retained risk and aggregate premium equally, reducing the problem to a lower-dimensional characterization in terms of the reinsurance share and aggregate premium. In the homogeneous exponential-utility case, the reinsurance share is independent of bargaining power and the outside insurance loading, whereas the premium reflects these parameters.

% {We further obtain closed-form solutions for the optimal contract in two benchmark cases} under exponential utility. %leading to closed-form characterizations of the optimal contract.
% In the first case, we consider the scenario when the price-fairness condition is lifted. In this setting, the optimal risk allocation is determined by the risk aversion and the deadweight cost of reinsurance, whereas bargaining powers and the centralized insurer’s loading affect only the premium payments.  This is in contrast to the case with price fairness, in which bargaining weights have a nontrivial impact on the resulting risk-sharing arrangement. 
% {In the second case, we consider a pool with homogeneous peers.} As a result, the peers optimally split the retained risk and premium equally, and the reinsurance proportion is again independent of bargaining power and the safety loading of the alternative insurance option.

Our numerical analysis illustrates several economic implications of the model. First, regarding the impact of price fairness, the dispersion of risk-bearing proportions and certainty-equivalent (CE) loadings among peers is reduced. {In response,} the P2P reinsurer may assume more risk, thereby lowering the P2P reinsurer’s CE net profit despite increasing her premium income. In addition, the bargaining powers of the P2P reinsurer and peers have an impact on the risk-bearing proportions, albeit modestly, even under exponential utility. This contrasts with some results in the literature in which bargaining power plays no role in risk sharing; see, e.g., \cite{aase2009nash}. This is consistent with economic intuition, as bargaining power typically reflects negotiation strength and market influence, and is therefore expected to influence risk allocation in practical insurance arrangements. Second, under the homogeneous setting, we examine the effect of the pool size. We find that expanding the pool strengthens internal diversification and reduces peers' risk-bearing proportions. Notably, however, a larger pool does not necessarily improve agents' welfare. When the peers collectively retain sufficient bargaining power, the welfare of the P2P reinsurer and peers may become non-monotonic in pool size.
%, leading to distinct critical pool sizes.
Thus, the desirability of pool expansion depends not only on diversification but also on how bargaining power evolves as additional members enter the pool.

\subsection{Related literature}

% \begin{itemize}
%     \item Decentralized insurance
%     \item Nash bargaining model
% \end{itemize}

Our paper contributes to research on decentralized insurance schemes, including mutual insurance \citep{laux2010financing,li2025mean}, tontines \citep{milevsky2015optimal,ng2025individual}, mutual-aid schemes \citep{abdikerimova2022peer,li2023optimal}, and, most relevant for this paper, P2P insurance contracting \citep{denuit2021risk,denuit2025comonotonicity}, a literature that has developed along several research directions, including the architecture of P2P risk sharing.
%, one direction studies P2P risk sharing as the underlying architecture.
For example, \cite{feng2023peer} devise proportional risk mutualization matrices and \cite{ghossoub2025pareto} study P2P risk sharing under robust distortion risk measures, both with applications in flood risk management. Furthermore, and pertinent to P2P insurance contracting, asymmetric information induces issues such as adverse selection and moral hazard, which are studied in, e.g., \cite{chen2024cost} and \cite{boonen2025contractibility}, respectively. 

Our work considers the strategic interactions between a P2P reinsurer and pool members in small pools using a Nash bargaining model, contributing to the literature on game-theoretic contracting in P2P insurance. In this regard, the closest work to ours is \cite{boonen2026pareto}, which also accounts for an additional reinsurance layer under Pareto and Bowley game frameworks. Our work stands out in that
%In addition to different game models, what separates our work from theirs is that 
the Nash bargaining model incorporates heterogeneous bargaining power and disagreement points among the P2P reinsurer and peers directly in the objectives, reflecting a cooperative interaction with competitive welfare allocations, in contrast to the fully cooperative or monopolistic setting therein.

The present work also connects the literature on decentralized insurance and Nash bargaining models, pioneered in \cite{nash1950bargaining} and extended to the asymmetric case by \cite{kalai1977nonsymmetric}, and studied with other solution concepts in \cite{kalai1975other}. 
In the actuarial context, Nash bargaining models have been used to study centralized insurance contracting. To name a few, \cite{boonen2016pricing} adopt distortion risk measures to study reinsurance contracting. \cite{asimit2021risk} consider the Pareto-optimal insurance design, and its premium decision coincides with the Nash bargaining solution. 
% Regarding the insurance contracting under asymmetric Nash bargaining,
\cite{chi2024optimal} study a stop-loss contract between a risk-neutral insurer and a risk-averse insured, and \cite{boonen2025asymmetric} focus on proportional indemnity with a risk-averse insurer and a risk-averse insured.
%, where both the insurer and the insured are risk-averse.
In the context of decentralized insurance, extending the aforementioned two-agent models to multi-agent models is necessary, but existing contributions remain limited, with the notable exception of \cite{aase2009nash}, which uses multi-agent Nash bargaining models to study the reinsurance market and its relationship to the competitive equilibrium in the absence of the price-fairness condition and a strategic reinsurer. 
%\takwacomment{Despite of different aim of study, their work can be translated to the case without price fairness and reinsurance layer in our context, which separates us our theoretical contribution to theirs.}
%\KTHNcomment{[Also one sentence on the difference between our work? Say no price fairness or reinsurer?] }  \takwacomment{I added some lines for that.}
%Thus, our work bridges a significant divide between studies on P2P insurance and Nash bargaining models, and  <-- kinda repeated in the first sentence of the paragraph
Finally, our study of coalitional stability also links the paper to related topics in game theory, including noncooperative foundations for Nash bargaining, e.g., the Nash core in \cite{okada2010nash} and stable coalitional bargaining in \cite{compte2010coalitional}.

\subsection{Organization}
The remainder of this work is organized as follows. Section \ref{Sec:setting} presents the general setting and formulates the Nash bargaining problem under price fairness. Section \ref{Sec:NB} studies the Nash bargaining solution, characterizes the optimal P2P insurance contract, and derives tractable benchmark cases. Section \ref{Sec:Stability} develops ex post coalitional stability tests under both price fairness and the no-price-fairness benchmark. Section \ref{Sec:Numeric} presents a numerical analysis investigating the impact of price fairness and the effect of pool size under homogeneity. Section \ref{Sec:Conclusion} concludes and discusses future research directions. All proofs are relegated to 
%collected in
the appendix.
%{to enhance readability}.
%to streamline the reading flow. 

\section{Setting and problem formulation}\label{Sec:setting}

\subsection{Peer-to-peer insurance scheme}
We consider a peer-to-peer (P2P) insurance scheme with a P2P reinsurer and $n\geq 2$ peers. The P2P reinsurer, denoted by $R$, is risk-averse, has utility function $U_R$, and is endowed with initial wealth $w_R$. Their collective loss is denoted by $S:=\sum_{i=1}^nX_i$. Throughout the paper, we assume that $S$ is non-degenerate. Otherwise, the problem reduces to a deterministic transfer problem, and the risk-sharing becomes trivial. In addition, it excludes the knife-edge case in which the individual losses are coupled so that their aggregate is almost surely constant.\footnote{In the dependence literature, this corresponds to a jointly mixable configuration of the marginals; see \cite{wang2016joint}.}

Each peer $i\in\mathcal{N}:=\{1,\dots,n\}$ is endowed with initial wealth $w_i$ and faces a loss $X_i$ {defined on a probability space $(\Omega,\mathcal{G},\mathbb{P})$} with $X_i\geq0$ $\mathbb{P}$-a.s.~and $\mathbb{E}[X_i]=\mu_i\in(0,\infty)$. 
%\KTHNcomment{Comment: We are using $i\in\mathcal{N}$ and $i=1,\dots,n$ interchangeably. Maybe better to stick to 1 most of the time. } \takwacomment{Mostly, I use $i\in\mathcal{N}$ to save some space but using $\sum_{i=1}^n$, do we need to replace it by $\sum_{i\mathcal{N}}$ for all?}
Each peer is also a standard expected-utility maximizer with a utility function $U_i$. 
%{\color{red}We additionally assume that differentiation under the expectation sign is valid for all expected utility functions.}
In what follows, we refer to all decision-makers as agents, including peers in the P2P insurance pool and the P2P reinsurer, and denote their collection by $\mathcal{I}:=\{1,2,\dots,n,R\}$.
%Additionally, we collect all peers' indexes as the set $\mathcal{N}:=\{1,2,\dots,n\}$. 

The P2P insurance scheme has two risk-management components: 
\begin{enumerate}
    \item \textbf{Reinsurance}: The group of peers purchases a proportional reinsurance contract written on their collective loss $S$ at a proportion $a_R\in[0,1]$.\footnote{This stems from the group-covered model; see \cite{feng2024unified} for further exposition.} Hence, the retained loss for the group is given by $(1-a_R)S$.  In exchange, the group must pay a reinsurance premium $P$. 
   
    \item \textbf{Risk mutualization}: The retained loss $(1-a_R)S$ and the premium payment $P$ are distributed among peers via a linear sharing rule \citep[e.g.,][]{schumacher2018linear}. In other words, Peer $i$ pays an amount $b_i$ ex ante and bears a proportion $a_i\in[0,1]$ of the collective loss $S$ ex post. Peer $i$'s total contribution is therefore $a_iS+b_i$,
    where $\sum_{i=1}^nb_i=P$ and $a_R+\sum_{i=1}^na_i=1$ (market-clearing condition). 
\end{enumerate}

We assume that the P2P reinsurer faces a linear deadweight loss with a loading factor $\tau\geq 0$ when providing the proportional reinsurance contract. Hence, using the market-clearing condition, the P2P reinsurer's terminal wealth is given by
{\small\begin{equation}\label{Eq:WR}
    W_R(\bm{a},P):=w_R+P-(1+\tau)a_RS=w_R+P-(1+\tau)\left(1-\sum_{i=1}^na_i\right)S.
\end{equation}}%
The premium $P$ also supports the administration of the P2P platform and may include both a subscription fee and a reinsurance premium. Hence, we do not impose the condition that $a_R=0$ implies $P=0$. {On the other hand, the terminal wealth $W_i$ of Peer $i$ is given by $W_i := w_i - a_iS -b_i$ for all $i\in\mathcal{N}$.}

% Throughout the paper, we adopt the following assumption. 
% \begin{assumption}\label{Ass:non_degenerate_S}
%     $S$ is non-degenerate.
% \end{assumption}

% Otherwise, the problem reduces to a deterministic transfer problem, and the risk-sharing becomes trivial. In addition, it excludes the knife-edge case in which the individual losses are coupled so that their aggregate is almost surely constant.\footnote{In the dependence literature, this corresponds to a jointly mixable configuration of the marginals; see \cite{wang2016joint}.}
%for further exposition.}

\subsection{Price fairness}

We impose a price-fairness condition that avoids excessive cross-subsidization, formulated as follows:\footnote{Such a fairness concept can be termed perfect price fairness; see \cite{cohen2022price} for further exposition.}

\begin{definition}[Price fairness]\label{Def:PriceFairness}
    Upon joining the P2P insurance scheme, each peer's expected contribution is determined according to the expected-value premium principle with a common loading $\rho(P,a_R)$, i.e., 
    {\small\begin{equation}\label{Eq:PricingFairness}
    \mathbb{E}[a_iS+b_i]=(1+\rho(P,a_R))\mu_i,\quad\text{for all } i\in\mathcal{N}.
\end{equation}}%
\end{definition}
\noindent
By summing \eqref{Eq:PricingFairness} over all $i\in\mathcal{N}$, and using $a_R+\sum_{i=1}^na_i=1$ and $\sum_{i=1}^nb_i=P$, we obtain $\rho(P,a_R)=P/\mathbb{E}[S]-a_R$.

\begin{remark}[Non-negativity of common loading]
Under the P2P reinsurer's IR constraint, the common loading $\rho(P,a_R)$ is non-negative. Indeed, Jensen's inequality implies that $P\geq(1+\tau)a_R\mathbb{E}[S]$, and therefore $\rho(P,a_R)=\frac{P}{\mathbb{E}[S]}-a_R\geq\tau a_R\geq 0$. Thus, the price-fairness condition is consistent with the interpretation of $\rho(P,a_R)$ as a loading.
\end{remark}

\begin{remark}[On premium principles that define price fairness]\label{Rem:otherPF}
    It is possible to adopt other premium principles to define price fairness. As long as the principle involves a common loading factor, a similar closed-form loading can be obtained. 
    However, using premium principles with two loading factors, e.g., the mean-variance principle, requires either augmenting the P2P reinsurer's decision variables or pre-specifying one of them.
\end{remark}

By the fact that $\rho(P,a_R)=P/\mathbb{E}[S]-a_R$, each peer's total contribution can be written as
{\small\begin{equation*}
a_iS+b_i
=a_i(S-\mathbb{E}[S])+(1+\rho(P,a_R))\mu_i=a_i(S-\mathbb{E}[S])+\left(1-a_R+\frac{P}{\mathbb{E}[S]}\right)\mu_i,
\end{equation*}}%
which is similar to the linear residual sharing rule in \cite{yang2025optimality} with a different fairness notion.
Hence, using the market-clearing and the price-fairness conditions, each peer's terminal wealth becomes
{\small\begin{equation}\label{Eq:Wi}
    W_i(\bm{a},P):=w_i - a_i(S-\mathbb{E}[S]) - \left(\frac{P}{\mathbb{E}[S]} + \sum_{j=1}^n a_j \right)\mu_i.
\end{equation}}%

Summarizing the above, we collect the notions of risk mutualization
among peers, proportional reinsurance, and premium to define the P2P insurance contract below. Note that we do not include $a_R$ in the definition as it can be inferred from $\bm{a}$ by the market-clearing condition.
\begin{definition}[P2P insurance contract]\label{Def:P2Pcontract}
    A pair $(\bm{a},P)$ is called a P2P insurance contract, where $\bm{a}:=(a_1,\dots,a_n)\in[0,1]^n$ with $\sum_{i=1}^n a_i\le 1$ represents the risk mutualization among members and $P\geq0$ is the P2P insurance premium.
\end{definition}
\noindent
Given a P2P insurance contract $(\bm{a},P)$, each Peer $i$'s payment $b_i(\bm{a},P)=\left(\sum_{j=1}^n a_j+\frac{P}{\mathbb{E}[S]}\right)\mu_i-a_i\mathbb{E}[S]$, $i\in\mathcal{N}$, can also be determined; such payment is allowed to be negative, in which case it represents an ex ante transfer or rebate to Peer $i$.

Throughout the paper, unless stated otherwise, we impose the following regularity conditions on expected utility, which ensure boundedness, continuity, and continuous differentiability of the expected-utility functions. These conditions are imposed only up to the first derivative because the subsequent analysis uses first-order differentiability of expected utilities. The $C^2$ assumption, together with strict concavity, ensures that marginal utilities are well-behaved in the first-order and covariance arguments below.

    \begin{assumption}\label{Ass:EU}
    For any $j\in \mathcal{I}$, $U_j(\cdot)$ is twice continuously differentiable, strictly increasing, and strictly concave. In addition, for any $K \in \mathbb{R}$, and any $l,k\in \{0,1\}$, $\mathbb{E}\left[{S}^k\left|U_R^{(l)}\left(w_R -(1+\tau)S +K\right) \right|\right]<+\infty$ and $\mathbb{E}\left[{S}^k\left|U_i^{(l)}\left(w_i-S + K \right)\right| \right]<+\infty$ for any $i\in\mathcal{N}$,
        % \begin{equation*}
        %  \mathbb{E}\left[{S}^k\left|U_R^{(l)}\left(w_R -(1+\tau)S +K\right) \right|\right]<\infty  , \quad     \mathbb{E}\left[{S}^k\left|U_i^{(l)}\left(w_i-S + K \right)\right| \right]<\infty ,\quad\text{for any } i\in\mathcal{N},
        % \end{equation*}
    where $U^{(l)}_R$ and $U^{(l)}_i$ denote the $l$-th derivative of $U_R$ and $U_i$, respectively. 
    \end{assumption}

% \KTHNcomment{\begin{assumption}\label{Ass:EU}
%     All expected utilities appearing in the
% optimization problems are finite and continuous in the P2P insurance contract variables.
% \end{assumption}}

\subsection{Nash bargaining model}

To formulate the Nash bargaining model, we need to define each agent's disagreement point. For the peers, the natural choice would be their status quo utility levels, i.e., the no-insurance case: $\mathbb{E}[U_i(w_i-X_i)]$ for all $i\in\mathcal{N}$. In this paper, we assume that traditional centralized insurance contracts with full indemnity are available to all peers. We suppose that the centralized contract is priced by the expected value principle, with a loading $\theta\geq 0$, and that all peers would choose it without the P2P insurance scheme, assuming it improves welfare relative to the no-insurance case. As such, the disagreement point of Peer $i$ is given by $d_i:=U_i(w_i-(1+\theta)\mu_i)$, where $\theta\geq 0$ is chosen to satisfy $U_i(w_i-(1+\theta)\mu_i)\geq \mathbb{E}[U_i(w_i-X_i)]$ for all $i\in\mathcal{N}$.\footnote{A uniform $\theta\geq 0$ exists by Jensen's inequality and the continuity of $U_i$.} On the other hand, the P2P reinsurer's disagreement point is the status quo utility level $d_R:=U_R(w_R)$.

Let $\delta_R>0$ and $\delta_i>0$ be the {bargaining} power of the P2P reinsurer and Peer $i$, respectively, such that $\sum_{i\in\mathcal{I}}\delta_i=1$. The multilateral asymmetric Nash bargaining {solution} is defined as the {maximizer} of the following optimization problem:  
% \begin{align}\label{Prob:hetero}
% \begin{split}
% \max_{P\geq0,\ \bm{a}\in\Delta^n}&\left(\mathbb{E}[U_R(W_R(\bm{a},P))]-d_R\right)^{\delta_R}\prod_{i=1}^n\left(\mathbb{E}\left[U_i\left(W_i(\bm{a},P)\right)\right]-d_i\right)^{\delta_i}\\
% s.t.\ & \mathbb{E}[U_R(W_R(\bm{a},P))]\geq d_R,\quad\mathbb{E}\left[U_i\left(W_i(\bm{a},P)\right)\right]
% \geq d_i\quad\text{for all }i\in\mathcal{N},
% \end{split}
% \end{align}
{\small\begin{equation}\label{Prob:hetero}
    \max_{P\geq0,\ \bm{a}\in\Delta^n}\prod_{j\in\mathcal{I}}\left(\mathbb{E}\left[U_j\left(W_j(\bm{a},P)\right)\right]-d_j\right)^{\delta_j}\quad\text{s.t.}\quad \mathbb{E}\left[U_j\left(W_j(\bm{a},P)\right)\right]
\geq d_j\quad\text{for all }j\in\mathcal{I},
\end{equation}}%
where $\Delta^n:= \{ \bm{x}\in [0,1]^n : \sum_{i=1}^n x_i\leq 1\}$ denotes the closed unit simplex with slack coordinate $a_R=1-\sum_{i=1}^{n}a_i$, and the $n+1$ inequality constraints are the individual rationality (IR) constraints for all agents. We shall {refer to a contract} $(\bm{a},P)$ such that
% $\mathbb{E}[U_R(W_R(\bm{a},P))] > d_R$ and
$\mathbb{E}\left[U_j\left(W_j(\bm{a},P)\right)\right]> d_j$ for all $j\in\mathcal{I}$ as a strictly feasible contract.

\section{Nash Peer-to-Peer insurance bargaining}\label{Sec:NB}

In this section, we connect Problem \eqref{Prob:hetero} to the asymmetric Nash bargaining theory axiomatized in \cite{kalai1977nonsymmetric}; in particular, the Nash bargaining solution satisfies Pareto optimality, strict feasibility, invariance to positive affine transformations, and independence of irrelevant alternatives. We then characterize the optimal P2P insurance contract of Problem \eqref{Prob:hetero}, and derive tractable benchmark cases.

\subsection{Preliminaries}
\label{sec:prelim}
First, we present preliminary results that facilitate solving Problem \eqref{Prob:hetero} and characterizing the properties of its solution.
% In the sequel, we shall drop the collective insurance proportion $a_R$, since $a_R = 1-\sum_{i=1}^na_i$ by the market-clearing condition\KTHNcomment{, and we reparametrize $W_R(\bm{a},P)$ and $W_i(\bm{a},P)$ as $W_R(\bm{a},P)$ and $W_i(\bm{a},P)$, respectively.} 
% We also define $\Delta := \{ (a_1,\dots,a_n)\in [0,1]^n : \sum_{i=1}^n a_i\leq 1\}$, which is a convex set in $\mathbb{R}^n$. 
To ease the presentation, we define the functions $V_j(\bm{a},P):=\mathbb{E}\left[U_j(W_j\left(\bm{a},P\right) )\right]$, for $\bm{a}\in\Delta^n$ and $P\in\mathbb{R}$,\footnote{For now, we take $P\in\mathbb{R}$. The condition $P\geq 0$ will be imposed later.} for all $j\in\mathcal{I}$.
% \begin{equation*}
% \begin{aligned}
%     V_R(\bm{a},P)&:=\mathbb{E}\left[U_R(W_R\left(\bm{a},P\right) )\right],\\
% V_i(\bm{a},P)&:=\mathbb{E}\left[U_i\left( w_i - a_i(S-\mathbb{E}[S]) - \left(\frac{P}{\mathbb{E}[S]} + \sum_{j=1}^n a_j \right)\mu_i \right)\right],\ i\in\mathcal{N}.
% \end{aligned}
% \end{equation*}
% \begin{equation*}
%     V_j(\bm{a},P):=\mathbb{E}\left[U_j(W_j\left(\bm{a},P\right) )\right],\quad j\in\mathcal{I}.
% \end{equation*}
Using Assumption \ref{Ass:EU}, we can ensure that $V_j$, $j\in\mathcal{I}$, are well-defined and satisfy certain regularity properties.
% , continuous, concave, and continuously differentiable. 
\begin{proposition}
\label{pp:V:continuous}
   Under Assumption \ref{Ass:EU}, the following statements hold:
   \begin{enumerate}
       \item The function $V_j : \Delta^n\times {\mathbb{R}} \to \mathbb{R}$ {is} well-defined for any $j\in\mathcal{I}$, i.e., $|V_j(\bm{a},P)|<\infty$;
       \item $V_j$, $j\in\mathcal{I}$, is continuous and concave on $\Delta^n\times \mathbb{R}$ and belongs to $C^1(\textup{int}(\Delta^n)\times \mathbb{R})$, where $\textup{int}(\Delta^n)$  denotes the interior of $\Delta^n$;
       \item For every $\bm a\in\Delta^n$ and every $i\in\mathcal N$, $\lim_{P\to-\infty} V_i(\bm a,P)>d_i$ and $\lim_{P\to+\infty} V_i(\bm a,P)<d_i$.
   \end{enumerate}
   
   % for any $j\in\mathcal{I},$ the function $V_j : \Delta^n\times {\mathbb{R}} \to \mathbb{R}$ {is} well-defined, i.e., $|V_j(\bm{a},P)|<\infty$. In addition, $V_j$ is continuous and concave on $\Delta^n\times \mathbb{R}$ and belongs to $C^1(\textup{int}(\Delta^n)\times \mathbb{R})$, where $\textup{int}(\Delta^n)$  denotes the interior of $\Delta^n$. 
\end{proposition}
\begin{proof}
See Appendix \ref{App:pp:V:continuous}.
\end{proof}

% \KTHNcomment{
%     \begin{remark}
%         Following the proof of Proposition \ref{pp:V:continuous}, under Assumption \ref{Ass:EU}, one can easily extend the well-definedness and continuity of $V_i$, $i\in \mathcal{N}$, to the domain $\Delta\times \mathbb{R}$. 
%     \end{remark}
% } \takwacomment{But this is included in the proof (See Appendix \ref{App:pp:V:continuous}.), although we say something like it follows similarly.} \KTHNcomment{Our statement in PP \ref{pp:V:continuous} does not include that range; we can also put it in the statement. This is used only to ensure $P_{\max}$ exists beyond $[0,\infty)$. }
% \takwacomment{I amended the domain and added a footnote below to explain why we need this domain.}

For any $\bm{a}\in \Delta^n$, let $P_{\min}(\bm{a})$ denote the premium that solves
$V_R(\bm{a},P_{\min}(\bm{a}))=d_R$. This quantity represents the minimum premium that the P2P reinsurer is willing to accept in order to provide the service to the pool. By the monotonicity of $U_R$, for every $\bm{a}\in\Delta^n$, we have $V_R(\bm{a},0)\le d_R$, whereas $V_R(\bm{a},P)>d_R$ for all sufficiently large $P>0$. Together with the continuity and strict monotonicity of $V_R(\bm{a},\cdot)$, this implies that there exists a unique $P_{\min}(\bm{a})\ge 0$ for every $\bm{a}\in\Delta^n$. Moreover, $P_{\min}(\bm{a})=0$ if and only if $\sum_{i=1}^n a_i=1$, or equivalently $a_R=0$.

    Similarly, we define the maximum price that peers would be willing to pay for the P2P scheme over purchasing the centralized full indemnity. 
%     We impose the following assumption throughout Sections \ref{sec:prelim}-\ref{sec:general_case}.
%     \begin{assumption}\label{P_max_assumption}
%     For every $\bm a\in\Delta^n$ and every $i\in\mathcal N$, $\lim_{P\to-\infty} V_i(\bm a,P)>d_i$ and $\lim_{P\to+\infty} V_i(\bm a,P)<d_i$.
%     \end{assumption}
% \noindent
For each peer $i\in\mathcal{N}$, we let $P_{\max}^i(\bm{a})$ be the root of the equation $V_i(\bm{a},P_{\max}^i(\bm{a}))= d_i$. 
Indeed, for fixed $\bm a\in\Delta^n$, the mapping
$P\mapsto V_i(\bm a,P)$ is continuous and strictly decreasing. Hence, the intermediate
value theorem and
% Assumption \ref{P_max_assumption}
the third statement in Proposition \ref{pp:V:continuous}
give existence, and strict monotonicity gives uniqueness.\footnote{Note that we consider the domain $\Delta^n\times\mathbb{R}$ instead of $\Delta^n\times[0,+\infty)$ in Proposition \ref{pp:V:continuous} to justify this claim.} Note that $P_{\max}^i(\bm{a})$ need not be non-negative for all $\bm{a}\in\Delta^n$. 
    % \KTHNcomment{(Existence and uniqueness of solution.)} \takwacomment{I overlooked this: $P_min$ should follow unconditionally from the intermediate value theorem; $P_max^i$ needs certain conditions to make each peer better off when paying nothing; but this is not an issue; see added explanation below.}
     %   However, those $\bm{a}$ do not create an issue as they are not admissible, as will be shown in Lemma \ref{Lem:ConcaveH}.
    % \begin{equation}\label{Eq:Pmaxi} 
    %   % \mathbb{E}\left[U_i\left( w_i - a_i(S-\mathbb{E}[S]) - \left(\frac{P_{\max}^i(\bm{a})}{\mathbb{E}[S]} + \sum_{j=1}^n a_j \right)\mu_i \right)\right] 
    %   V_i(\bm{a},P_{\max}^i(\bm{a}))= d_i. 
    % \end{equation}
We define $P_{\max}(\bm{a}):= \min_{i\in\mathcal{N}} P_{\max}^i(\bm{a})$, and it represents the maximum price the pool is willing to pay such that every peer is at least indifferent between joining the P2P scheme and purchasing centralized full insurance, i.e., the disagreement points of every peer. Define $H:\Delta^n \to \mathbb{R}$ by $H(\bm{a}) := P_{\max}(\bm{a})-P_{\min}(\bm{a})$. The following lemma establishes its concavity.
\begin{lemma}\label{Lem:ConcaveH}
 Under Assumption \ref{Ass:EU}, 
 % and \ref{P_max_assumption}, 
 $H$ is concave on $\Delta^n$. 
\end{lemma}
\begin{proof}
    See Appendix \ref{App:ConcaveH}.
\end{proof}

We define the admissible set of P2P contracts by $\mathcal{F}:= \left\{ (\bm{a},P) \in \Delta^n\times[0,\infty) : P_{\min}(\bm{a})\leq P \leq P_{\max}(\bm{a}) \right\}$. Thus, a contract is admissible if its aggregate premium is acceptable to both the P2P reinsurer and all peers.
% We define the set of admissible P2P contracts $\mathcal{F}$ which consists of contracts $(\bm{a},P)$ such that $P\geq 0$ and $P_{\max}(\bm{a}) \geq P_{\min}(\bm{a})$. That is, $\mathcal{F}$ is defined by 

The following statement asserts the compactness and convexity of the admissible set $\mathcal{F}$.
\begin{lemma}[Compactness and convexity of feasible set]\label{Lem:ccF}
    Under Assumption \ref{Ass:EU},
    % and \ref{P_max_assumption}, 
    % the admissible contract set
    $\mathcal{F}$ is compact and convex.
\end{lemma}
\begin{proof}
See Appendix \ref{App:ccF}.
\end{proof}

We end this subsection by showing that there exists a strictly feasible P2P insurance contract $(\bm{a},P)\in\mathcal{F}$, i.e., 
%$V_R(\bm{a},P) > d_R$ and
$V_j(\bm{a},P)>d_j$ for all $j\in\mathcal{I}$.
%the $n+1$ inequalities corresponding to all agents' IR constraints in \eqref{Prob:hetero} are strict.
To proceed, we define the set of \textit{strictly feasible contracts} $\mathcal{K}\subset \mathcal{F}$ 
%\KTHNcomment{$\mathcal{F}?$} 
by $\mathcal{K}:= \left\{ (\bm{a},P) \in \Delta^n\times[0,\infty) : P_{\min}(\bm{a})< P < P_{\max}(\bm{a}) \right\}$. By definition of $\mathcal{K}$, all IR constraints are strictly satisfied by any $(\bm{a},P)\in \mathcal{K}$.

\begin{lemma}[Existence of strictly feasible P2P insurance contracts]\label{Lemma:SF}
    Under Assumption \ref{Ass:EU}, 
    % and \ref{P_max_assumption}, 
    there exists a strictly feasible P2P insurance contract,
    %\KTHNcomment{since we are proving $\mathcal{K}\neq\emptyset$ instead of the optimal solution, perhaps we should not relate Problem \ref{Prob:hetero} in the statement?}, \takwacomment{I agree it's not necessary; it's just about strict IRs, I deleted the terms.}
    i.e., $\mathcal{K}\neq \emptyset$, if and only if 
    {\small\begin{equation}\label{Eq:supH}
        \sup_{\bm{a}\in\Delta^n}H(\bm{a})>0. 
    \end{equation}}%
\end{lemma}

\begin{proof}
See Appendix \ref{App:SF}.
\end{proof}

\noindent
Condition \eqref{Eq:supH} is satisfied, for example, if there exists $\bm{\hat{a}}\in\Delta^n$ such that $ V_i(\bm{\hat{a}},P_{\min}(\bm{\hat{a}}))>d_i$ for all $i\in\mathcal{N}$.
% \begin{equation*}
%     V_i(\bm{\hat{a}},P_{\min}(\bm{\hat{a}}))>d_i\quad\text{for all }i\in\mathcal{N}.
% \end{equation*}
Indeed, at $P=P_{\min}(\bm{\hat{a}})$, the P2P reinsurer's utility equals the disagreement point, whereas every peer is strictly better off. Therefore, $P_{\min}(\bm{\hat{a}})<P_{\max}(\bm{\hat{a}})$, and Condition \eqref{Eq:supH} is thus satisfied. 
In Sections \ref{Sec:MANBproperites}-\ref{sec:general_case} and \ref{sec:stability_price_fair}, Condition \eqref{Eq:supH} is imposed so that the set of strictly feasible contracts is non-empty.
% In the remainder of Section \ref{Sec:NB}, unless otherwise stated, 
% we assume Condition \eqref{Eq:supH}, so that the set of strictly feasible contracts is non-empty.

\subsection{Asymmetric Nash bargaining}\label{Sec:MANBproperites}

% {\color{red} What properties can we assert from the solution to the Nash product optimization?}

Define the auxiliary free-disposal utility space $\mathcal{U}$ as
{\small\begin{equation*}
\mathcal{U}:=\{u\in\mathbb{R}^{n+1}: \text{ there exists } (\bm{a},P)\in\mathcal{F}%(\KTHNcomment{\mathcal{F}?} 
\text{ such that } d_j\leq u_j\leq V_j(\bm{a},P) \text{ for all }j\in\mathcal{I}\}.
\end{equation*}}%
Although some elements of $\mathcal{U}$ may not be directly generated by a P2P insurance contract, the maximization of the Nash product over $\mathcal{U}$ is equivalent to the original contract problem \eqref{Prob:hetero}. Indeed, for any $u\in\mathcal{U}$, there exists $(\bm{a},P)\in\mathcal{F}$ such that $u_j\leq V_j(\bm{a},P)$ for all $j\in\mathcal{I}$. Since the Nash product is non-decreasing on $\{\bm{u}:\bm{u}\geq \bm{d}\}$, $\prod_{j\in\mathcal{I}}(u_j-d_j)^{\delta_j}\leq\prod_{j\in\mathcal{I}}(V_j(\bm{a},P)-d_j)^{\delta_j}$. Conversely, $(V_j(\bm{a},P))_{j\in\mathcal{I}}\in\mathcal{U}$ for every $(\bm{a},P)\in\mathcal{F}$. Thus Problem \eqref{Prob:hetero} and the following problem:
{\small\begin{equation}\label{Prob:NB}
\max_{u\in\mathcal{U}}\prod_{j\in\mathcal{I}}(u_j-d_j)^{\delta_j},
\end{equation}}%
have the same optimal value. Under Condition \eqref{Eq:supH}, the optimal value is positive, and any optimizer of \eqref{Prob:NB} must be generated by an actual P2P contract. 
% Note that there might be elements in $\mathcal{U}$ that are not generated by a P2P insurance contract. Nonetheless, the definition guarantees that those elements are dominated by the utility vector induced by some P2P insurance contracts.
% This also implies that finding the optimal utility vector of Problem \eqref{Prob:hetero} is equivalent to solving the following:
% and
Thus, we study the properties of the solution to Problem \eqref{Prob:NB} in the remainder of this subsection, connecting it to the theory of multi-agent bargaining.

Below, we use a pair $(\mathcal{V},\bm{d})$ to represent a bargaining problem, where $\bm{d}=(d_1,d_2,\dots,d_n,d_R)$.\footnote{{Here, we implicitly assumed that the space $\mathcal{V}$ is defined using the disagreement point $\bm{d}$, i.e., $(\mathcal{V},\bm{d}) = (\mathcal{V}(\bm{d}),\bm{d})$. }} 
%(\KTHNcomment{It looks like the definition of $\mathcal{U}$ depends on $\bm{d}$. Are we adopting the notation that $(\mathcal{U}(\bm{d}),\bm{d}) = (\mathcal{U},\bm{d})$, or can the two $d$'s be different? } \takwacomment{There's no differences and I just follows the convention which does not specify $\mathcal{U}$ depending on $\bm{d}$ but making it a pair so that they are linked.}
Let $\mathcal{B}$ denote the class of $(n+1)$-agent bargaining problems
$(\mathcal{V},\bm d)$ such that $\mathcal V\subset\{{\bm u}\in\mathbb{R}^{n+1}:{\bm u}\geq {\bm d}\}$ is
nonempty, compact, convex, and comprehensive from below relative to $\bm d$,
in the sense that if $\bm u\in\mathcal V$ and $\bm d\leq \bm v\leq \bm u$,
then $\bm v\in\mathcal V$, and such that there exists
$\bm u\in\mathcal V$ with $\bm u>\bm d$. 
% A bargaining outcome
% is a map $f:\mathcal B\to\mathbb R^{n+1}$ satisfying
% $f(\mathcal U,\bm d)\in\mathcal U$ for every
% $(\mathcal U,\bm d)\in\mathcal B$.
% a map $f:\mathcal{B}\to\mathbb{R}^{n+1}$, \KTHNcomment{where $f(\mathcal{U},\bm{d})$ is the solution of the optimization problem \eqref{Prob:NB}.} \takwacomment{We can't say this in general, this is only true with those 4 axioms shown in the next theorem.} \KTHNcomment{We need a proper definition or characterization of what $f$ really is. E.g., what a ``bargaining process" or ``optimal utility vector" really mean? What is it optimal with respect to? Or is it just an arbitrary function? I believe in the first statement of the theorem below, we are considering $f$ as the optimal solution to \eqref{Prob:NB}. If we want a ``general definition of $f$, should we say `` Let $f:\mathcal{B}\to \mathbb{R}^{n+1}$ be a bargaining outcome, i.e., $f(\mathcal{U},\bm{d}) \in \mathcal{U}$?  } \takwacomment{The comment that f is the solution to \eqref{Prob:NB} then it satisfies those 4 axioms is true, so I updated the statement in the next theorem. But for the other comments, let me think of them.}
%to characterize the bargaining process. \KTHNcomment{What does it mean and what are some properties of $f$? Just an arbitrary function or does it mean the solution of \eqref{Prob:NB}? } \takwacomment{f is defined to map the bargaining to the out of bargaining, i.e., optimal utility vector. The following theorem is meant to prove its the optimal argument of \eqref{Prob:NB} with those 4 axioms.}
The following lemma asserts the compactness and convexity of the utility space $\mathcal{U}$.

\begin{lemma}[Compactness and convexity of $\mathcal{U}$]\label{Lem:ccU}
    Under Assumption \ref{Ass:EU}, 
    % and \ref{P_max_assumption}, the free-disposal utility set 
    $\mathcal{U}$ is a compact and convex set.
\end{lemma}
\begin{proof}
    See Appendix \ref{App:ccU}.
\end{proof}

%\noindent
%\KTHNcomment{[Why no indent?]}
By construction, $\mathcal{U}$ is nonempty, satisfies $\mathcal U\subset\{{\bm u}\in\mathbb{R}^{n+1}:{\bm u}\geq {\bm d}\}$, and is comprehensive from below relative to ${\bm d}$. Lemma \ref{Lem:ccU} shows that $\mathcal{U}$ is compact and convex. If Condition \eqref{Eq:supH} also holds, then Lemma \ref{Lemma:SF} implies the existence of $\bm u\in\mathcal U$ with $\bm u>\bm d$. Hence, $(\mathcal U,\bm d)\in\mathcal B$.

Armed with the above preparations, we can now prove that the solution to Problem \eqref{Prob:NB} exhibits the properties stated at the beginning of Section \ref{Sec:NB}.

\begin{theorem}[Characterization of Nash bargaining solutions]\label{Thm:properties} 
%{\color{red}TB: I guess we need to assume that $\mathcal K\neq\emptyset$ or equivalently \eqref{Eq:supH}?} \takwacomment{I added it back; otherwise, we need to assume that $\mathcal{U}$ is essential then it's clumsy.}
% Suppose that Condition \eqref{Eq:supH} holds.
For $(\mathcal{V},\bm{d})\in\mathcal{B}$, define
{\small\begin{equation}
f(\mathcal{V},\bm{d}):=\argmax_{u\in\mathcal{V}}\prod_{j\in\mathcal{I}}(u_j-d_j)^{\delta_j}.
\label{eq:general_bargaining_form}
\end{equation}}%
Then, $f(\mathcal{V},\bm{d})$ is well-defined and satisfies the following four properties:
\begin{enumerate}
    \item \textit{Pareto optimality}: If $f(\mathcal{V},\bm{d})=\bm{u}$ and $\Tilde{u}_j\geq u_j$ for all $j\in\mathcal{I}$, then either $\bm{u}=\bm{\Tilde{u}}$ or $\bm{\Tilde{u}}\notin\mathcal{V}$.
    \item \textit{Strict feasibility}: 
    % If $f(\mathcal{U},\bm{d})\in\mathcal{U}$,
 $f(\mathcal{V},\bm{d})>\bm{d}$.\footnote{For any $\bm{a},\bm{b}\in\mathbb{R}^{n+1}$, $\bm{a}>\bm{b}$ means strict inequality in every component.} 
    \item \textit{Invariance to positive affine transformations (IPAT)}: {For any $\bm{\alpha} =(\alpha_i)_{i=1}^{n+1},\bm{\beta} = (\beta_i)_{i=1}^{n+1} \in \mathbb{R}^{n+1}$ with $\alpha_i>0$ for all $i=1,\dots,n+1$,} define 
    {\small\begin{equation}\label{Eq:positiveaffine}
        T:\mathbb{R}^{n+1}\to\mathbb{R}^{n+1},\;T(y_1,\dots,y_{n+1})=(\alpha_1y_1+\beta_1,\dots,\alpha_{n+1}y_{n+1}+\beta_{n+1}). 
    \end{equation}}%
    % $T:\mathbb{R}^{n+1}\mapsto\mathbb{R}^{n+1},(y_1,y_2,\dots,y_n,y_{n+1})\rightarrow(\alpha_1y_1+\beta_1,\alpha_2y_2+\beta_2,\dots,\alpha_ny_n+\beta_n,\alpha_{n+1}y_{n+1}+\beta_{n+1})$, 
%    where $\alpha_i>0$ for all $i=1,\dots,n+1$. 
Then, it holds that $T(f(\mathcal{V},\bm{d}))=f(T(\mathcal{V}),T(\bm{d}))$.
    \item \textit{Independence of irrelevant alternatives (IIA)}: For any $(\mathcal{V}',\bm{d}),(\mathcal{V},\bm{d})\in\mathcal{B}$ with $\mathcal{V}'\subset\mathcal{V}$, if $f(\mathcal{V},\bm{d})\in\mathcal{V}'$, then $f(\mathcal{V}',\bm{d})=f(\mathcal{V},\bm{d})$.
\end{enumerate}

\noindent
Conversely, if {a {bargaining outcome} $f:\mathcal{B}\to \mathbb{R}^{n+1}$ satisfies} these four properties, then there exists a  weight vector
$\bm{\delta}:=(\delta_1,\delta_2,\dots,\delta_n,\delta_R)\in\Delta^{n+1}_+:=\{\bm{x}\in(0,+\infty)^{n+1}:\sum_{i=1}^{n+1}x_i=1\}$ such that $f$ takes the form in \eqref{eq:general_bargaining_form}.
%\Timcomment{$\partial \Delta^{n+1}$ is not really a boundary of the simplex, no?}
%\takwacomment{Indeed, no as $\partial \Delta^{n+1}=\{\bm{x}\in[0,+\infty)^{n+1}:\sum_{i=1}^{n+1}x_i=1\}$ and $x_i$ can be zero. So, I changed it to $\Delta^{n+1}_+$.}
%\Timcomment{Before, $\Delta$ was $n$-dimensional.} 
% \takwacomment{I overlooked the diminsion so I use $\Delta^{n+1}$ and $\Delta^{n+1}_+$ to replace $\Delta$ and $\Delta_+$}
%\takwacomment{I now use $\Delta^n$ and $\partial\Delta^n$ then we don't need to redefine the notation just for different dimension.}
%and a unique bargaining process $f$ such that 
% \begin{equation*}
% f(\mathcal{V},\bm{d})=\argmax_{\bm{u}\in\mathcal{V}}\prod_{j\in\mathcal{I}}(u_j-d_j)^{\delta_j},\quad(\mathcal{V},\bm{d})\in\mathcal{B}.
% \end{equation*}
% where $\Delta^{n+1}_+:=\{(\delta_1,\delta_2,\dots,\delta_n,\delta_R)\in\mathbb{R}^{n+1}_{++}:\sum_{j\in\mathcal{I}}\delta_j=1\}$.
\end{theorem}

\begin{proof}
    See Appendix \ref{App:properties}.
\end{proof}

\noindent
Applying Theorem \ref{Thm:properties} to the bargaining problem $(\mathcal U,\bm d)\in\mathcal B$ generated by the P2P insurance yields the solution to Problem \eqref{Prob:NB}.

\cite{kalai1977nonsymmetric} axiomatizes the two-agent asymmetric Nash bargaining solution using the four axioms stated in Theorem \ref{Thm:properties}. In the present P2P insurance setting, we apply the corresponding multi-agent weighted Nash bargaining characterization to the feasible utility set generated by P2P insurance contracts. Other characterizations of the multi-agent case under different sets of axioms are further studied in, e.g.,  \cite{peters1991characterizing,rachmilevitch2015nash}. For further exposition and a more comprehensive literature review, we refer interested readers to \cite{thomson2022axiomatic}.

The axioms in Theorem \ref{Thm:properties} have natural interpretations in P2P insurance contracting. Pareto optimality ensures that the selected agreement lies on the efficient frontier of the feasible utility set, so that no agent can be made better off without making another agent worse off. Strict feasibility reflects voluntary participation by requiring every agent to strictly improve upon its disagreement point. IPAT implies that the bargaining outcome is unaffected by independent positive affine transformations of agents' utility functions. Equivalently, changing utility scales changes only the numerical representation of the utility vector, not the underlying economic agreement. IIA means that, once the bargaining compromise has been selected, removing unchosen alternatives from the feasible set does not alter the agreement. This is natural in P2P insurance bargaining as the selected contract should not depend on the mere presence of unattractive contracts in the feasible set.
%in the background.
% Inspecting the properties of Nash bargaining, Pareto optimality and IIA are desirable in the context of P2P insurance contracting. Pareto optimality guarantees that the resulting contract lies on the Pareto frontier and thus ensures efficient risk sharing among agents. IIA implies that, economically, if the P2P reinsurer and peers have already chosen the best bargaining compromise, then removing contracts that were not chosen should not change the agreement. This suits the context of P2P insurance bargaining as its outcome should not depend on the mere presence of unattractive contracts in the background.

These properties distinguish Nash bargaining from standard fixed-weight weighted-sum optimization (WSO). While WSO is useful for characterizing Pareto-efficient allocations, it does not directly account for the distribution of gains over disagreement utilities. Thus, even if a WSO strictly improves upon the disagreement point, the magnitude and distribution of utility gains relative to it are not incorporated directly into the objective. Moreover, fixed-weight WSO is not invariant to independent positive affine transformations of individual utility functions, and therefore relies on interpersonally comparable utility units. Finally, the supporting-hyperplane characterization of Pareto optima through WSO generally yields non-negative weights; see, e.g., \cite{aase2002perspectives,ng2025pareto}. Hence, some agents may receive zero weight and become welfare-irrelevant. By contrast, the Nash bargaining characterization in Theorem \ref{Thm:properties} entails strictly positive bargaining weights for all agents, ensuring that every agent’s utility gain enters the bargaining objective. This feature is particularly appropriate for voluntary P2P insurance arrangements in small pools.

\subsection{General case}\label{sec:general_case}

In this subsection, we solve Problem \eqref{Prob:hetero} and thereby characterize the optimal P2P insurance contract. 
Under Condition \eqref{Eq:supH}, Lemma \ref{Lemma:SF} ensures that $\mathcal{K}\neq\emptyset$. Hence, there exists a P2P contract with strictly positive Nash product. By contrast, any P2P contract in $\mathcal{F}\backslash\mathcal{K}$ has at least one binding IR constraint and therefore yields a zero Nash product. Consequently, every maximizer of Problem \eqref{Prob:hetero} must belong to $\mathcal{K}$. On $\mathcal{K}$, all expected-utility surpluses are strictly positive, so maximizing the Nash product is equivalent to maximizing the log-transformed objective:
{\small    \begin{equation}\label{Eq:J}
\mathcal{J}(\bm{a},P)
:=\sum_{j\in\mathcal{I}}\delta_j\ln\left(\mathbb{E}\left[U_j\left(W_j(\bm{a},P)\right)\right]-d_j\right).
\end{equation}}%
The function $\mathcal{J}$ is well-defined precisely on $\mathcal{K}$, equivalently for contracts satisfying $P_{\min}(\bm{a})<P<P_{\max}(\bm{a})$. Since $P_{\min}(\bm{a})\geq 0$, the condition $P>P_{\min}(\bm{a})$ already implies $P>0$.
% which is well-defined if $(\bm{a},P)\in\mathcal{K}$, {and only if $P_{\min}(\bm{a})<P<P_{\max}(\bm{a})$. }
% \KTHNcomment{I suggest either adopting $(\bm{a},a_R,P)$ or $(\bm{a},P)$ all the time, but not using them interchangeably. This is because the partial derivatives will look different, especially $\partial J/\partial a_i$. Although by rearranging the system, we still get the same FOC system. If we only use $(\bm{a},P)$ in 3.1 and 3.2, we shall say so in those 2 subsections. I have no preference which one to use. It is easier to use $(\bm{a},P)$ in previous sections; and if we were to keep Corollary \ref{Cor:ZR}, it seems  better to use $(\bm{a},a_R,P)$.     }

The following lemma shows that, under the optimal P2P insurance contract $(\bm{a}^*,P^*)$, every peer bears a strictly positive share of the residual aggregate fluctuation $S-\mathbb{E}[S]$, unless full reinsurance occurs.

\begin{lemma}\label{Lem:POSai}
    Suppose that Assumption 
    % \ref{Ass:non_degenerate_S},
    \ref{Ass:EU} 
    % and \ref{P_max_assumption} 
    and Condition \eqref{Eq:supH} hold, and let $(\bm{a}^*,P^*)$ be the optimal P2P insurance contract of Problem \eqref{Prob:hetero}. Then, either of the following holds:
    \begin{enumerate}
        \item $a_R^*=1$ and $\bm{a}^*=\bm{0}$;
        \item $a_R^*<1$, and $\bm{a}^*>\bm{0}$ componentwise. In particular, if $\tau=0$, then $a_R^*>0$.
    \end{enumerate}
\end{lemma}

\begin{proof}
    See Appendix \ref{App:POSai}.
\end{proof}

% The following theorem establishes the existence and uniqueness of the optimal P2P insurance contract resulting from the formulated Nash bargaining.

Lemma \ref{Lem:POSai} leaves only three possible reinsurance regimes of full reinsurance, partial reinsurance with all peer shares positive, and zero reinsurance with all peer shares positive. We characterize the optimal contract in each regime in the following theorem. The full- and zero-reinsurance regimes correspond to boundary points of $\Delta^n$; hence, the corresponding first-order conditions are understood as one-sided Karush--Kuhn--Tucker (KKT) conditions, whose directional derivatives exist under Assumption \ref{Ass:EU}.

\begin{theorem}[Existence, uniqueness, and characterization of the solution to Problem \eqref{Prob:hetero}]\label{Thm:EUsol}
    Assume that Assumption 
    % \ref{Ass:non_degenerate_S}, 
    \ref{Ass:EU} 
    % and \ref{P_max_assumption} 
    and Condition \eqref{Eq:supH} hold. Then, Problem \eqref{Prob:hetero} admits a unique solution, and we have the following cases: 
\begin{enumerate}
\item (Partial reinsurance, $a_R^*\in(0,1)$). 
%\takwacomment{This does not exclude zero reinsurance.}\KTHNcomment{What if \eqref{Eq:a0NPD} does not hold, then it means $P^*=0$ cannot be true?} \takwacomment{It's equivalent condiition so it means $a_R^*=1$ cannot be true.} 
If the optimal contract $(\bm{a}^*,P^*)$ satisfies ${0<}\sum_{i=1}^na_i^*<1$, 
% and \KTHNcomment{$a_i^*>0$ (what about now? as we are considering $a_R^*\in(0,1)$ in this case, the latter is implied by the former?)} \takwacomment{$a_i>0$ for all $i$ and $\sum_ia_i<1$ ($0<$ is not needed.) are OK. The logic is like the resulting case is $a_R^*\in(0,1)$.} \KTHNcomment{It seems $\sum a_i \in (0,1)$ is more directly related to the case condition, i.e., $a_R\in(0,1)$, and this would imply $a_i>0$ by the lemma. So I would suggest writing $0<\sum a_i < 1$, instead of $\sum a_i<1$ and $a_i>0$ if they are equivalent.} \takwacomment{I see the point now; invoking the Lemma 3.7 should do the job; we can go for your writing.} for all $i\in\mathcal{N}$, 
then $(\bm{a}^*,P^*)$ is characterized as the unique solution in $\mathcal{K}$ of the system:
{\small\begin{align}\label{Eq:HeteroFOC}
\begin{split}
&\frac{\delta_R\mathbb{E}[S]\mathbb{E}[U_R'(W_R(\bm{a}^*,P^*))]}
{\mathbb{E}[U_R(W_R(\bm{a}^*,P^*))]-d_R}
=\sum_{j=1}^n
\frac{\delta_j\mu_j\mathbb{E}\left[U_j'(W_j(\bm{a}^*,P^*))\right]}
{\mathbb{E}[U_j(W_j(\bm{a}^*,P^*))]-d_j},\\
&
\frac{\delta_i\mathbb{E}\left[U_i'\!\left(W_i(\bm{a}^*,P^*)\right)(S-\mathbb{E}[S])\right]}
{\mathbb{E}[U_i(W_i(\bm{a}^*,P^*))]-d_i} =
\frac{\delta_R\mathbb{E}\left[U_R'(W_R(\bm{a}^*,P^*)) \left((1+\tau)S-\mathbb{E}[S]\right) \right]}
{\mathbb{E}[U_R(W_R(\bm{a}^*,P^*))]-d_R},\quad \text{for all }i\in\mathcal{N}.
\end{split}
\end{align}}%
Conversely, if $(\bm{a}^*,P^*)\in\mathcal{K}$ is the solution of the system \eqref{Eq:HeteroFOC}, which satisfies $\sum_{i=1}^na_i^*<1$ and $a_i^*>0$ for all $i\in\mathcal{N}$,
% \takwacomment{But we cannot invoke Lem 3.7 in the converse, would it be better to stick the original one, $a_i>0$ for all i and $\sum_ia_i<1$?} \KTHNcomment{seems you are right} 
then it is the optimal solution to Problem \eqref{Prob:hetero}.
\item (Full reinsurance, $a_R^*=1$). Assume that $P_{\min}(\bm{0})<P_{\max}(\bm{0})$. Then, {the optimal contract is given by $(\bm{0},P^*(\bm{0}))$ if and only if the following holds:} 
%(\KTHNcomment{I think it should be $\tau$ instead of $1+\tau$)}
%\takwacomment{It should be $1+\tau$ and the main argument is $U_i'(W_i(0,P(0)))$ is deterministic.}
   {\small     \begin{equation}\label{Eq:a0NPD}
            \frac{\delta_R (1+\tau)\mathbb{E}[U_R'(W_R(\bm{0},P^*(\bm{0})))S]}
{\mathbb{E}[U_R(W_R(\bm{0},P^*(\bm{0})))]-d_R}
\leq
\sum_{i=1}^n
\frac{\delta_i\mu_iU_i'(W_i(\bm{0},P^*(\bm{0})))}
{U_i(W_i(\bm{0},P^*(\bm{0})))-d_i},
% +
% \frac{\delta_k\mathbb{E}\!\left[U_k'(W_k(\bm{0},P^*(\bm{0})))(S-\mathbb{E}[S])\right]}
% {\mathbb{E}[U_k(W_k(\bm{0},P^*(\bm 0)))]-d_k}
% , \text{ for all } k\in \mathcal{N}.
    \end{equation}}%
%following holds:
   %\takwacomment{iff} the optimal contract is given by  $(\bm{0},P^*(\bm{0}))$,
   where $P^*(\bm{0})\in(P_{\min}(\bm{0}),P_{\max}(\bm{0}))$ is the unique solution of
  {\small  \begin{equation}
    \label{eq:P*0:equation}
        \frac{\delta_R\mathbb{E}[U_R'(W_R(\bm{0},P^*(\bm{0})))]}
{\mathbb{E}[U_R(W_R(\bm{0},P^*(\bm{0})))]-d_R}
=\sum_{i=1}^n
\frac{\delta_i\mu_i U_i'(W_i(\bm{0},P^*(\bm{0})))}
{\mathbb{E}[S]\left(U_i(W_i(\bm{0},P^*(\bm{0})))-d_i\right)}.
    \end{equation}}%
\item (Zero reinsurance, $a_R^*=0$). Assume that $\tau>0$.\footnote{Zero reinsurance cannot be optimal when $\tau=0$ by Lemma \ref{Lem:POSai}.} The contract $(\bm{a}^*,P^*)\in\mathcal{K}$ with  $\sum_{i=1}^n{a}^*_i=1$ is optimal for Problem \eqref{Prob:hetero} if and only if there exists {$\widehat{\lambda}>0$} such that the following holds: 
% (\KTHNcomment{I replace $\lambda$ with $-\hat{\lambda}$ which seems more natural)}
{\small\begin{align}\label{Eq:ZeroReinCond}
\begin{split}
    \widehat{\lambda} &= \frac{\delta_i\mathbb{E}[U_i'(W_i(\bm{a}^*,P^*))(S-\mathbb{E}[S]) ]}{\mathbb{E}\left[U_i\left(W_i(\bm{a}^*,P^*) \right)\right]-d_i},  \quad\text{for all } i \in \mathcal{N}, \\
    \widehat{\lambda} &\leq    \frac{\delta_R \tau U_R'(w_R+P^*)\mathbb{E}[S]}{U_R(w_R+P^*)-d_R}, \qquad 
          \frac{\delta_RU_R'(w_R+P^*) }{U_R(w_R + P^*)-d_R } - \sum_{j=1}^n\frac{\delta_j\mu_j\mathbb{E}[U_j'(W_j(\bm{a}^*,P^*))] }{\mathbb{E}[S] \left(\mathbb{E}[U_j(W_j(\bm{a}^*,P^*))]-d_j\right)}=0. 
\end{split}  
    \end{align}}%
\end{enumerate}
\end{theorem}

\begin{proof}
See Appendix \ref{App:EUsol}.
\end{proof}

Under partial reinsurance, the first equation in \eqref{Eq:HeteroFOC} balances the P2P reinsurer's weighted marginal gain from receiving an additional unit of premium against the aggregate weighted marginal utility loss borne by the peers who finance that premium. The second set of equations in \eqref{Eq:HeteroFOC} characterizes the optimal allocation of residual aggregate risk. It states that each peer's weighted marginal response to the residual $S-\mathbb{E}[S]$ equals the P2P reinsurer's associated weighted marginal response to a change in the reinsured aggregate loss.
% Under the optimal contract, the last $n$ expressions in \eqref{Eq:HeteroFOC} indicate that the weighted marginal utility responding to the aggregate residual fluctuation (i.e., $S-\mathbb{E}[S]$) is the same among peers. In other words, if one peer had a strictly lower weighted marginal utility than the other, one could shift a piece of the residual fluctuation towards the better improved peer, thereby increasing the Nash product.

Economically, if one peer had a lower weighted marginal cost of bearing an additional share of the residual fluctuation than another peer, then transferring a small amount of residual risk toward that peer would increase the Nash objective. At the optimum, such marginal improvements are exhausted. The boundary conditions in \eqref{Eq:a0NPD} and \eqref{Eq:ZeroReinCond} have analogous interpretations: in the full-reinsurance case, a marginal reduction in reinsurance cannot improve the Nash-product objective, whereas in the zero-reinsurance case, a marginal increase in reinsurance cannot improve it.

% From the premium-allocation perspective, the first equation in
% \eqref{Eq:HeteroFOC} balances the P2P reinsurer's weighted marginal utility gain from an additional unit of premium with the aggregate weighted marginal utility loss of the peers from paying that premium. The second set of equations balances each peer's weighted marginal utility response to the aggregate residual fluctuation with the P2P reinsurer's weighted marginal utility response to changes in the reinsured aggregate risk.

\begin{remark}[Tractability under exponential utility and zero deadweight loss]
% Unlike the homogeneous setting, adopting exponential utility does not yield closed-form risk bearing for agents. This is in stark contrast to a similar Nash bargaining work \cite{aase2009nash}, which studies the expected exponential utility setting with a linear sharing rule.
As shown later in Proposition \ref{Prop:NoFairExpSol}, when the price-fairness condition is removed and agents are equipped with exponential utilities, the zero-deadweight-loss case $\tau=0$ admits a closed-form risk-sharing rule; see \eqref{Eq:tau0aRa}. This rule has the familiar inverse-risk-aversion structure of \cite{aase2009nash}. In contrast, under price fairness, the side payments are pinned down by the common-loading condition, which removes the free transfer parameter that would otherwise yield a closed-form characterization of $\bm{a}^*$. Consequently, in the general heterogeneous case of Theorem \ref{Thm:EUsol}, the optimal risk-sharing vector must be characterized through the nonlinear first-order system.
% It can be proven that without the price-fairness condition, adopting exponential utility and zero deadweight loss, $\tau=0$, leads to a closed-form risk-sharing rule $\bm{a}^*$ (see Proposition \ref{Prop:NoFairExpSol}, and particularly \eqref{Eq:tau0aRa}) that is similar to \cite{aase2009nash},  which can be treated as the Nash bargaining solution in the absence of a reinsurer and the price-fairness condition in our context.
%\KTHNcomment{(how many agents they were considering? we may want to mention the difference with our Proposition \ref{Prop:NoFairExpSol} briefly.)} \takwacomment{They consider $n$-agent case can be treated no-reinsurance and no price fairness case in our context.}
%, which is not the case in Theorem \ref{Thm:EUsol}. 
% The reason is that the price-fairness condition renders all side payments tractable up to the total premium $P$ at the expense of the tractability of $\bm{a}$. 
% {The reason for the lack of a closed-form expression in Theorem \ref{Thm:EUsol} is that the price-fairness condition reduces the flexibility in choosing the side payments, eliminating the free parameter that would allow a closed-form characterization of $\bm{a}$.   }
\end{remark}

\subsection{Tractable cases and benchmarks}

The general characterization in Theorem \ref{Thm:EUsol} is expressed through nonlinear first-order conditions. To obtain more explicit insights and to prepare for the numerical analysis, we now study tractable cases and benchmarks. Section \ref{sec:exp_utility_wo_price_fairness} removes the price-fairness condition under exponential utilities, yielding a benchmark that isolates the role of price fairness. Sections \ref{sec:homo_general_utility} and \ref{sec:homo_exp_utility} return to the price-fairness model and impose homogeneity among peers, first under general utility and then under exponential utility.

\subsubsection{Exponential-utility benchmark without price fairness}\label{sec:exp_utility_wo_price_fairness}

In this subsection, we consider a benchmark Nash-bargaining problem in which the price-fairness condition in Definition \ref{Def:PriceFairness} is removed. Under exponential utilities, this benchmark admits a tractable characterization and will be used in the numerical comparison in Section \ref{sec:numerical_price_fairness}. We set $U_R(x)=-e^{-\gamma_Rx}$ and $U_i(x)=-e^{-\gamma_ix}$ for all $i=1,\ldots,n$, {where $\gamma_i,\gamma_R>0$ represent the risk aversion of Peer $i$ and the P2P reinsurer, respectively.} 
{Also, 
% in the remainder of this subsection, 
we impose the following assumption, which will be used in Sections
\ref{sec:exp_utility_wo_price_fairness},
\ref{sec:stability_wo_price_fairnes}, and
\ref{sec:numerical_price_fairness}.
%\takwacomment{Then, we don't to state it in the following mathematical propositions?}
 \begin{assumption}
    \label{ass:mgf}
      {  Let $\bar{t}:=\sup\{t\geq 0: M_S(t) < \infty\}$, where $M_S(t)$ is the moment generating function (MGF) of $S$. We assume that $\bar{t}>\max\{ \gamma_R(1+\tau) , \overline{\gamma}\}$, where $\overline{\gamma}:=(\sum_{i=1}^n\gamma_i^{-1})^{-1}$.  }
    \end{assumption}
    
    % \begin{assumption}
    % \label{ass:mgf}
    %    There exists $\overline{t}>\gamma_R(1+\tau)\vee \max_{i\in\mathcal{N}}\gamma_i$ such that the moment generating function (MGF) $M_S(t)$ of $S$ is finite for all $t \leq \overline{t}$. \KTHNcomment{In particular, we set $\bar{t}:=\sup\{t\geq 0: M_S(t) < \infty\}$ in the sequel. }
    % \end{assumption}
    
\begin{remark}
\label{rmk:mgf}
Since $S\geq 0$, the moment-generating function $M_S(t)$ is well-defined
for all $t\leq 0$. In addition, for any $k\geq 0$ and $\varepsilon>0$,
there exists $C_{k,\varepsilon}>0$ such that $    S^k e^{tS} \leq C_{k,\varepsilon} e^{(t+\varepsilon)S}.$
% \[
%     S^k e^{tS} \leq C_{k,\varepsilon} e^{(t+\varepsilon)S}.
% \]
Hence, if $M_S$ is well-defined on $(-\infty,\bar{t})$, then $M_S$
is smooth on this interval.

With exponential utility, the differentiability, monotonicity,
and strict concavity required in Assumption \ref{Ass:EU} are automatically
satisfied. Moreover, 
under Assumption \ref{ass:mgf}, the general
integrability requirements in Assumption \ref{Ass:EU} are replaced by the
MGF condition, which is not stronger than that in Assumption \ref{Ass:EU} over the entire contract domain, since it does not require $M_S(\gamma_i)<+\infty$ for all $i\in\mathcal{N}$. Instead, it ensures that $M_S$, $m_S$ and the required derivatives are well-defined at all arguments arising in Proposition \ref{Prop:NoFairExpSol}.
% which is stronger and ensures that the MGF and its
% derivatives are well defined at all arguments used below.
\end{remark}

Denote the vector of peers' payments by $\bm{b}:=(b_1,\dots,b_n)$. After dividing each utility surplus by the positive factor 
$e^{-\gamma_j w_j}$, $j\in\mathcal{I}$, which does not affect the maximizer of the Nash product,
the optimization problem can be expressed as
{\small\begin{align}\label{Prob:NoFairness}
\begin{split}
\max_{P\geq0,\, a_R\in[0,1],\bm{a}\in[0,1]^n,\, \bm{b}\in\mathbb{R}^n}&\left(1-\mathbb{E}\left[e^{\gamma_R((1+\tau)a_RS-P)}\right]\right)^{\delta_R}\prod_{i=1}^n\left(e^{\gamma_i(1+\theta)\mu_i}-\mathbb{E}\left[e^{\gamma_i(a_iS+b_i)}\right]\right)^{\delta_i}\\
s.t.\ & \mathbb{E}\left[e^{\gamma_R((1+\tau)a_RS-P)}\right]\leq 1,\quad e^{\gamma_i(1+\theta)\mu_i}
\geq \mathbb{E}\left[e^{\gamma_i(a_iS+b_i)}\right]\quad\text{for all }i\in\mathcal{N},\\ &a_R+\sum_{i=1}^na_i=1,\quad\sum_{i=1}^nb_i=P.
\end{split}
\end{align}}%
%\takwacomment{$a_R,a_i\geq 0$ are imposed under $\max$.}
Note that under exponential utility, the optimization problem is independent of the agents' initial wealth levels. 
% and its well-posedness is guaranteed by Assumption \ref{Ass:EU}. 
{In the remainder of this subsection, we shall denote by $\widetilde{\mathcal{F}}$ the set of all feasible contracts $(\bm{a},a_R,\bm{b},P)$ that satisfy the constraints in Problem \eqref{Prob:NoFairness}, and by $\widetilde{\mathcal{K}}\subset\widetilde{\mathcal{F}}$ the set of strictly feasible contracts with the additional requirement that $\mathbb{E}[e^{\gamma_R((1+\tau)a_RS-P)}]< 1$ and  $e^{\gamma_i(1+\theta)\mu_i}
> \mathbb{E}\left[e^{\gamma_i(a_iS+b_i)}\right]$ for all $i\in \mathcal{N}$.   }
%\KTHNcomment{(Indeed, we can drop $P,a_R$ as before, so a contract is reduced to the pair $(\bm{a},\bm{b})$)} \takwacomment{Let me check if I can rewrite it, it seems replacing $a_R$ by $1-\sum_ia_i$ and $P$ by $\sum_ib_i$ should work.} \KTHNcomment{Yes it should work but we have to write $C_R(1-\sum a_i)$ which is quite lengthy; so either way is fine. } \takwacomment{We can use $C_R(\bm{a})$ although $m_S(...)$ is still lengthy. }

To prepare for the main result of this section, define
$m_S:(-\infty,\bar t)\to[0,+\infty)$, $C_R:[0,1]\to[0,+\infty)$, and
$C_i:[0,1]\to[0,+\infty]$, $i\in\mathcal{N}$, 
%(\KTHNcomment{or we can write $C_i:[0, \bar{t}/\gamma_i) \to[0,\infty)$ if we don't need to extend the domain)} \takwacomment{(16) considers $\min_{\bm{a}\in\Delta}$ so maybe it's better to extend it; otherwise, if it's finite in between $1$ and $\bar{t}/\gamma_i$, then the subsequent results are not solid.} by
{\small\begin{equation}\label{eq:m:C}
    m_S(t):=\frac{M_S'(t)}{M_S(t)},\qquad
    C_R(a_R):=\frac{\ln M_S(\gamma_R(1+\tau)a_R)}{\gamma_R},\qquad
    C_i(a_i):=\frac{\ln M_S(\gamma_i a_i)}{\gamma_i},\quad i\in\mathcal{N},
\end{equation}}%
with the convention that $C_i(a_i)=+\infty$ when $M_S(\gamma_ia_i)=+\infty$.\footnote{The functions $C_i$, $i\in\mathcal{N}$, might not be finite throughout their stated domain under Assumption \ref{ass:mgf}, which, however, will not affect the subsequent mathematical results.}
By Assumption \ref{ass:mgf} and Remark \ref{rmk:mgf}, $m_S$ and $C_R$ are
well-defined on their stated domains. 

{The following statement asserts the non-emptiness of $\widetilde{\mathcal{K}}$.
\begin{lemma}
    \label{lem:NoFair:strict:feasible}
    % Under the exponential-utility specification, 
    Suppose that Assumption \ref{ass:mgf} holds. Then, $\widetilde{\mathcal{K}}\neq \emptyset$ if 
    the following condition is satisfied: 
   {\small \begin{equation}\label{Eq:IRCond}
        % (1+\theta)\mathbb{E}[S]>\min_{(\bm{a},a_R)\in \Delta^{n+1} }\left\{ C_R(a_R)+\sum_{i=1}^nC_i(a_i)\right\},
        (1+\theta)\mathbb{E}[S]>\min_{\bm{a}\in \Delta^{n} }\left\{ C_R\left(1-\sum_{i=1}^na_i\right)+\sum_{i=1}^nC_i(a_i)\right\}.
    \end{equation}}% %\Timcomment{Don't we want to take a min of a closed set, while $\Delta_+^{n+1}$ is open.} \takwacomment{I changed it back to the closed set $\Delta^{n+1}$}
%    \Timcomment{But dont we have $\sum a_i \leq 1$ instead of $\sum a_i=1$.} \takwacomment{I am not sure what you mean? We have $\sum_{i=1}^na_i\leq1$ and $\sum_{i\in\mathcal{I}}a_i=1$.} \Timcomment{I meant, that $\Delta^{n+1}$ includes the condition that $a_R+\sum_{i=1}^na_i\leq1$, but we actually need $a_R+\sum_{i=1}^na_i=1$?} \takwacomment{I developed another version to avoid the notational problem.}
% where $\partial \Delta^{n+1}:=\{ (a_i)_{i=1}^{n+1}\in [0,1]^{n+1}:\sum_{i=1}^{n+1}a_i = 1\}$. 
% Then, $\widetilde{\mathcal{K}}\neq \emptyset$.  
\end{lemma}
    \begin{proof}
        See   Appendix \ref{App:NoFair:strict:feasible}.
    \end{proof}
}

The strict feasibility enables Problem \eqref{Prob:NoFairness} to be transformed into
{\small\begin{align}\label{Prob:auxNF}
\begin{split}
    &\max_{{(\bm{a},a_R,\bm{b},P)\in \widetilde{\mathcal{K}}} } \delta_R\ln\left(1-e^{\gamma_R(C_R(a_R)-P)}\right)+\sum_{i=1}^n\delta_i\ln\left(1-e^{\gamma_i(C_i(a_i)+b_i-(1+\theta)\mu_i)}\right).
 %   s.t.\ &\sum_{i=1}^nb_i=P,\quad a_R+\sum_{i=1}^na_i=1, \ \KTHNcomment{a_i,a_R\geq 0, \ i\in\mathcal{N}.}
\end{split}
\end{align}}%
%\takwacomment{It seems $\widetilde{\mathcal{K}}$ incoporate all constraints so we don't need to repeat constraints in the second line?}
The following proposition characterizes the unique maximizer of Problem \eqref{Prob:auxNF}
in the regime where the risk-sharing proportions are governed by the system \eqref{Eq:OptNFaRa}. Condition \eqref{eq:Cond:exp:no:fairness} is precisely the condition ensuring that this system admits a solution with $\bm{a}^*\in(0,1)^n$ and $a^*_R\in[0,1)$. If
\eqref{eq:Cond:exp:no:fairness} fails, the optimum lies on the zero-reinsurance boundary and must be characterized by the corresponding boundary KKT conditions.

\begin{proposition}\label{Prop:NoFairExpSol}
  Suppose that Assumption
  % \ref{Ass:non_degenerate_S}, 
  \ref{ass:mgf}, Condition \eqref{Eq:IRCond}, and the following hold: 
  %\takwacomment{They are same.}
   {\small  \begin{equation}
    \label{eq:Cond:exp:no:fairness}
    (1+\tau)\mathbb{E}[S]  \leq m_S(\overline{\gamma}).
    % \left( \frac{1}{\sum_{i=1}^n\gamma_i^{-1}} \right).
    \end{equation}}%
  Then, the optimal contract with strictly satisfied IR constraints in
Problem \eqref{Prob:auxNF} is given by $(\bm{a}^*,a_R^*,\bm{b}^*,P^*) \in \widetilde{\mathcal{K}}$, where 
$\bm{a}^*\in(0,1)^n$ and $a_R^*\in[0,1)$ uniquely solve
  % $(\bm{a}^*,a_R^*)\in\partial \Delta^{n+1}$
  % is the unique solution to 
  the system of equations 
   {\small\begin{equation}\label{Eq:OptNFaRa}
    \left\{\begin{array}{lr}
        (1+\tau)m_S(\gamma_R(1+\tau)a_R^*)=m_S(\gamma_ia_i^*) & \text{for } i=1,\dots,n\\
        a_R^*+\sum_{i=1}^na_i^*=1, & 
        \end{array}\right.
\end{equation}}%
    % which admits a unique solution in $\Delta^{n+1}_+$ if and only if the following condition holds: 
    % \begin{equation}
    % \label{eq:Cond:exp:no:fairness}
    %     (1+\tau)\mathbb{E}[S] \leq m_S\left( \frac{1}{\sum_{i=1}^n\gamma_i^{-1}} \right).
    % \end{equation}}
and $\bm{b}^*$ is given by 
{\small\begin{equation}\label{Eq:Optbi}
    b_i^*=(1+\theta)\mu_i-C_i(a_i^*)-\frac{1}{\gamma_i}\ln\left(1+\frac{q^*-1}{\chi_i}\right),\quad i=1,\dots,n,
\end{equation}}%
where $\chi_i:=\delta_R\gamma_R/\delta_i\gamma_i$ and $q^*>1$ is the unique solution of
{\small\begin{equation}\label{Eq:qstar}
    (1+\theta)\mathbb{E}[S]-\sum_{i=1}^n\left(C_i(a_i^*)+\frac{1}{\gamma_i}\ln\left(1+\frac{q^*-1}{\chi_i}\right)\right)-C_R(a_R^*)-\frac{\ln q^*}{\gamma_R}=0.
\end{equation}}%

In particular, when $\tau=0$, we have
{\small\begin{equation}\label{Eq:tau0aRa}
    a_R^*=\frac{\gamma_R^{-1}}{\gamma_R^{-1}+\overline{\gamma}^{-1}},
    % \frac{\gamma_R^{-1}}{\gamma_R^{-1}+\sum_{i=1}^n\gamma_i^{-1}},
    \qquad a_i^*=\frac{\gamma_i^{-1}}{\gamma_R^{-1}+\overline{\gamma}^{-1}},
    % \frac{\gamma_i^{-1}}{\gamma_R^{-1}+\sum_{i=1}^n\gamma_i^{-1}},
    \quad i=1,\dots,n.
\end{equation}}%

\end{proposition}

\begin{proof}
    See Appendix \ref{App:NoFairExpSol}.
\end{proof}

Under the exponential utility setting, each agent's risk bearing is independent of their bargaining power (i.e., $\delta_R$ and $\delta_i$, $i\in\mathcal{N}$) and the centralized insurer's loading $\theta$ (see \eqref{Eq:OptNFaRa}). In particular, when $\tau=0$, the optimal allocation $\bm{a}^*$ retains the familiar inverse-risk-aversion structure of \cite{aase2009nash}. Unlike the peer-only setting in their study,\footnote{Specifically, \cite{aase2009nash} studies a reinsurance syndicate, which can be interpreted as P2P risk sharing in our setting.} the P2P reinsurer herein participates in risk sharing, thereby scaling down peers' allocations. In addition, the side payments $\bm{b}^*$ are no longer zero-sum, since their aggregate finances the premium paid to the P2P reinsurer.

\subsubsection{Homogeneous peers under price fairness: general utility}\label{sec:homo_general_utility}

In this subsection, we study Problem \eqref{Prob:hetero} under the assumption that all peers are homogeneous in terms of utility functions, bargaining powers, initial wealth, and expected losses.  Specifically, we impose the following assumption throughout Sections \ref{sec:homo_general_utility}--\ref{sec:homo_exp_utility}
and in the numerical analysis of Section \ref{sec:pool_size}.

\begin{assumption}\label{Ass:homo}
    There exist a utility function $U$, constants $\delta \in (0,1/n)$, $w \in \mathbb{R}$, and $\mu >0$ such that, for all $i=1,\dots,n$,  $U_i\equiv U$,   $\delta_i = \delta$, $w_i = w$, and  $\mu_i=\mu$.
\end{assumption}

Under homogeneity,\footnote{Assumption \ref{Ass:homo} does not require the individual losses $X_i$ to be identically distributed. Equality of expected losses is sufficient for the equal-splitting result because, under price fairness, each peer’s post-contract wealth depends on the aggregate loss $S$ and the common mean loss $\mu$.} it is natural to consider an equal-splitting sharing rule among peers, that is,
{\small\begin{equation}\label{Eq:esSharing}
    a_i=\frac{1-a_R}{n},\qquad b_i=\frac{P}{n},\qquad i=1,\dots,n,
\end{equation}}%
for any admissible reinsurance-premium pair $(a_R,P)$ in Problem \eqref{Prob:hetero}. The following proposition justifies restricting attention to the sharing rule \eqref{Eq:esSharing} under Assumption \ref{Ass:homo}. 

\begin{proposition}\label{Prop:OptesSR}
    Under Assumptions
    % \ref{Ass:non_degenerate_S}, 
    \ref{Ass:EU},
    % \ref{P_max_assumption}, 
    \ref{Ass:homo} and Condition \eqref{Eq:supH}, the unique optimal solution $(\bm{a}^*,P^*)$ to Problem \eqref{Prob:hetero}  satisfies 
    $a_i^* = a_j^*$ and $b_i^* = b_j^*=P^*/n$ for all $i,j\in\mathcal{N}$.
    %where $b_i^*=P^*/n$.
\end{proposition}

\begin{proof}
    See Appendix \ref{App:OptesSR}.
\end{proof}

In the remainder of this subsection, we drop the subscripts $i=1,\dots,n$ since all peers are identical under Assumption \ref{Ass:homo} and the sharing rule \eqref{Eq:esSharing}. 

\begin{remark}[Alternative price-fairness principles]
    Assumption \ref{Ass:homo} imposes homogeneity, including equal expected losses. This is because price fairness (see Definition \ref{Def:PriceFairness}) is defined by the expected-value premium principle. If the variance or standard deviation principles are adopted, then equal variances would be required to deduce the equal-splitting sharing rule \eqref{Eq:esSharing}.
\end{remark}

By Assumption \ref{Ass:homo} and Proposition \ref{Prop:OptesSR}, Problem \eqref{Prob:hetero} can be transformed into 
{\small\begin{align}\label{Prob:Homo1}
\begin{split}
    \max_{a_R\in[0,1],\,P\geq0}&(\mathbb{E}[U_R(w_R+P-(1+\tau)a_RS)]-d_R)^{1-n\delta}\left(\mathbb{E}\left[U\left(w-\frac{P+(1-a_R)S}{n}\right)\right]-d\right)^{n\delta}\\
s.t.\ 
& \mathbb{E}[U_R(w_R+P-(1+\tau)a_RS)]\geq d_R,\quad
\mathbb{E}\left[U\left(w-\frac{P+(1-a_R)S}{n}\right)\right]\geq d.
\end{split}
\end{align}}%

\subsubsection{Homogeneous peers under price fairness: exponential utility}\label{sec:homo_exp_utility}

We now further specialize the homogeneous price-fairness model to exponential utility and derive an explicit characterization of the optimal contract $(a_R^*,P^*)$. Specifically, we take
  {\small  \begin{equation}
    \label{eq:exp:homo}
        U_R(x) = -e^{-\gamma_R x} \quad \text{and} \quad U(x) = -e^{-\gamma x}, \ x \in \mathbb{R},
    \end{equation}}%
where $\gamma_R,\gamma>0$. As in Section \ref{sec:prelim}, for $a_R\in[0,1]$, we define the indifferent premiums $P_{\min}(a_R)$ and $P_{\max}(a_R)$, respectively by the solutions of 
   {\small \begin{equation*}
        \mathbb{E}[U_R(w_R+P_{\min}(a_R)-(1+\tau)a_RS)] = d_R  \quad \text{and} \quad \mathbb{E}\left[U\left(w - \frac{P_{\max}(a_R)+(1-a_R)S}{n} \right) \right] = d,
    \end{equation*}}%
Finally, we define $\overline{H}:[0,1]\to \mathbb{R}$ by $\overline{H}(a_R):= P_{\max}(a_R) - P_{\min}(a_R)$. Under exponential utility \eqref{eq:exp:homo}, we have 
  {\small  \begin{equation*}
        \overline{H}(a_R)= (1+\theta)n\mu
        - \frac{n}{\gamma}\ln M_S\left(\frac{\gamma(1-a_R)}{n}\right) - \frac{1}{\gamma_R}\ln M_S\left(\gamma_R(1+\tau)a_R \right). 
    \end{equation*} }%
    % where $M_S$ is the moment generating function of $S$, i.e., $M_S(t):=\mathbb{E}[e^{tS}]$.
    
The following statement derives the closed-form characterization of the contract $(a_R^*,P^*)$ under exponential utility with homogeneous peers, where Assumption \ref{Ass:EU} is replaced by regularity conditions specific to exponential utilities.
%\KTHNcomment{[The statement of the PP still keeps Ass 2.6, to be removed?].} \takwacomment{It should be Ass. \ref{Ass:non_degenerate_S}, a typo.}

\begin{proposition}\label{Prop:ExpHomo}
    Suppose that Assumption
    % \ref{Ass:non_degenerate_S} and 
    \ref{Ass:homo} holds, and that\footnote{The first and third conditions are simply Assumption \ref{ass:mgf} and Condition \eqref{eq:Cond:exp:no:fairness} under Assumption \ref{Ass:homo}, respectively.} 
        \begin{enumerate}
            \item $M_S(t)<\infty$ for $t<\bar{t}$, where $\bar{t}> \max\{\gamma_R(1+\tau), \gamma/n\}$,
            \item $\sup_{a_R\in[0,1]}\overline{H}(a_R)>0$, and \item $(1+\tau)n\mu\leq m_S(\gamma/n)$.
        \end{enumerate} 
    Then, under the exponential utility setting \eqref{eq:exp:homo}, Problem \eqref{Prob:Homo1} admits a unique solution $(a_R^*,P^*)$, 
    where $a_R^* = \argmax_{a_R\in[0,1]} \overline{H}(a_R)\in{[0,1)}$ 
    is the unique solution of 
   {\small \begin{equation}\label{Eq:ExpaREquation}
        m_S\left(\frac{\gamma(1-a_R^*)}{n}\right)=(1+\tau)m_S(\gamma_R(1+\tau)a_R^*),
    \end{equation}}%
   where 
   % $m_S(t):= \frac{M_S'(t)}{M_S(t)}$, and 
   $P^* {\in (P_{\min}(a_R^*),P_{\max}(a_R^*))}$ {uniquely} solves
  {\small  \begin{equation}\label{Eq:HomoPFOC}
\frac{(1-n\delta)\gamma_RM_S(\gamma_R(1+\tau)a_R^*)}{e^{\gamma_RP^*}-M_S(\gamma_R(1+\tau)a_R^*)}=\frac{\delta\gamma M_S\left(\frac{\gamma(1-a_R^*)}{n}\right)}{e^{\gamma\left((1+\theta)\mu-\frac{P^*}{n}\right)}-M_S\left(\frac{\gamma(1-a_R^*)}{n}\right)}.
\end{equation}}%
\end{proposition}

\begin{proof}
    See Appendix \ref{App:ExpHomo}.
\end{proof}

Proposition \ref{Prop:ExpHomo} shows that, under homogeneous peers and exponential utility, the reinsurance share $a^*_R$ is independent of the bargaining weight $\delta$ and the outside loading $\theta$; only the premium $P^*$ depends on these parameters. This contrasts with the general heterogeneous case in Theorem \ref{Thm:EUsol}, where bargaining weights and the outside option can affect the risk-sharing vector itself.

\begin{remark}[Condition 3 of Proposition \ref{Prop:ExpHomo}]
If Condition 3 fails, then $\overline H'(0)<0$. Since $\overline H$ is concave, the maximizer of $\overline H$ over $[0,1]$ is $a^*_R=0$, and \eqref{Eq:ExpaREquation} is replaced by the corresponding boundary condition. Thus, Condition 3 identifies the regime in which $a_R$ is characterized by the first-order equation \eqref{Eq:ExpaREquation}.
\end{remark}

\section{Ex post coalitional stability}\label{Sec:Stability}

%After solving the main problem \eqref{Prob:hetero}, which includes all peers and the P2P reinsurer,
In this section, we investigate subgroup formation, that is, whether a subset of peers, together with the P2P reinsurer, can achieve welfare outcomes that weakly improve upon those obtained in the grand coalition. Because the P2P reinsurer provides the platform service, we assume that it must be included in any subgroup. Henceforth, we only consider subgroup formation among peers, where a subgroup is denoted by
%denoted by 
$\emptyset\neq \mathcal{M}\subsetneq \mathcal{N}$. 
%, where we omit mentioning the P2P reinsurer's index $R$ in the subgroup label, and hereafter we adopt a slight abuse of notation that $\mathcal{N}$ also indicates the grand coalition that includes all peers. \KTHNcomment{why don't we just use $\mathcal{I}$? For ease of definition $\Delta^{|\mathcal{M}|}$, perhaps we can use $\overline{M} := \mathcal{M}\cup \{R\}$, or sth like $\mathcal{M}_R$?}

\subsection{Stability under price fairness}\label{sec:stability_price_fair}
For any subgroup $\emptyset\neq \mathcal{M}\subsetneq\mathcal{N}$, define $S_{\mathcal{M}}:=\sum_{i\in \mathcal{M}}X_i$ 
% \begin{equation*}
% S_{\mathcal{M}}:=\sum_{i\in \mathcal{M}}X_i,\quad \Delta^{|\mathcal{M}|}:=\left\{\bm{a}^{\mathcal{M}}\in[0,1]^{|\mathcal{M}|}:\sum_{i\in \mathcal{M}}{a_i^\mathcal{M}} \leq 1\right\},
% \end{equation*}
and the optimal expected utility levels of all agents in the grand coalition $\mathcal{N}$ associated with the optimal contract $(\bm{a}^*,P^*)$ in Problem \eqref{Prob:hetero}: $\hat{u}_R^{\mathcal{N}}:=\mathbb{E}[U_R(W_R(\bm{a}^*,P^*))]$ and $\hat{u}_i^{\mathcal{N}}:=\mathbb{E}[U_i(W_i(\bm{a}^*,P^*))]$, $i\in\mathcal{N}$. Throughout this subsection, assume that Assumptions
% \ref{Ass:non_degenerate_S}, 
\ref{Ass:EU} and 
% \ref{P_max_assumption},
% together with 
Condition \eqref{Eq:supH} hold. Hence, by Theorem
\ref{Thm:EUsol}, the grand-coalition problem admits a unique optimal
solution $(\bm{a}^*,P^*)$.
In addition, we define the certainty-equivalent (CE) loading $\vartheta_i$, $i\in\mathcal{N}$, as the hypothetical loading that makes each peer indifferent to joining the P2P insurance contract in the grand coalition. It is defined as the solution of the following equation:
{\small\begin{equation}\label{Eq:vartheta}
    \hat{u}_i^{\mathcal{N}}=U_i(w_i-(1+\vartheta_i)\mu_i),\quad i\in \mathcal{N}.
\end{equation}}%
Note that due to the strict welfare improvement, we have 
{\small\begin{equation}\label{Eq:smallvartheta}
    \frac{w_i-U_i^{-1}\left(\hat{u}_i^{\mathcal{N}}\right)}{\mu_i}-1=\vartheta_i<\theta\quad\text{for all }i\in \mathcal{N}.
\end{equation}}%
The loadings $\vartheta_i$ can be interpreted as the {indifference} price of the P2P insurance contract, thereby facilitating comparisons of P2P insurance contracts across subgroups. 
Similarly, we define the P2P reinsurer CE net profit $\Pi$ by
{\small\begin{equation}\label{Eq:Pi}
    \Pi:=U_R^{-1}\left(\hat{u}_R^{\mathcal{N}}\right)-w_R.
\end{equation}}%

Fix a subgroup $\mathcal{M}\subsetneq \mathcal{N}$. 
%and a P2P insurance contract $(\bm{a}^{\mathcal{M}},P^{\mathcal{M}})$ \KTHNcomment{Do we need $\bm{b}^\mathcal{M}$ here first and remove it under the new subgroup fairness condition?} \takwacomment{I added some lines for that.} for $\mathcal{M}$, where $\bm{a}^{\mathcal{M}} = (a^{\mathcal{M}}_j)_{j\in \mathcal{M}}\in\Delta^{|\mathcal{M}|}$ and $P^{\mathcal{M}}>0$. 
Each peer in the subgroup $\mathcal{M}$ contributes $a_i^\mathcal{M}S_{\mathcal{M}}{+b^\mathcal{M}_i}$, where $\bm{a}^\mathcal{M}=(a^\mathcal{M}_i)_{i\in\mathcal{M}}\in\Delta^{|\mathcal{M}|}$ denotes the risk-bearing proportions and $b_i^\mathcal{M}$ is a side payment made by Peer $i\in\mathcal{M}$. {The aggregate premium received by the P2P reinsurer is thus $P^\mathcal{M}=\sum_{i\in\mathcal{M}}b^\mathcal{M}_i$}. 
Similar to contracts for the grand coalition $\mathcal{N}$, we adopt the price-fairness condition in the subgroup $\mathcal{M}$, i.e., $\mathbb{E}[a_i^\mathcal{M}S_{\mathcal{M}} + b_i^\mathcal{M}] = (1+\rho(P^\mathcal{M},a_R^\mathcal{M}))\mu_i$ for all $i\in \mathcal{M}$ (cf. Definition \ref{Def:PriceFairness}). We shall refer to $(\bm{a}^{\mathcal{M}},P^{\mathcal{M}})$ as the P2P insurance contract for the subgroup $\mathcal{M}$, under which the wealth levels of the P2P reinsurer and all peers in the subgroup $\mathcal{M}$ are defined as
{\small\begin{align*}
    &W^{\mathcal{M}}_i\left(\bm{a}^{\mathcal{M}},P^{\mathcal{M}}\right):=w_i-a_i^{\mathcal{M}}(S_{\mathcal{M}}-\mathbb{E}[S_{\mathcal{M}}])-\frac{\mu_i}{\mathbb{E}[S_{\mathcal{M}}]}P^{\mathcal{M}}-\mu_i\sum_{j\in \mathcal{M}}a_j^{\mathcal{M}},\quad i\in \mathcal{M},\\
    &W^{\mathcal{M}}_R\left(\bm{a}^{\mathcal{M}},P^{\mathcal{M}}\right):=w_R+P^{\mathcal{M}}-(1+\tau)\left(1-\sum_{i\in \mathcal{M}}a_i^{\mathcal{M}}\right)S_{\mathcal{M}}.
\end{align*}}%
%\KTHNcomment{I think we need to define a price-fairness condition for the subgroup $\mathcal{M}$ to arrive at $W_i^{\mathcal{M}}$? That is, $\mathbb{E}[a_i^\mathcal{M}S_{\mathcal{M}} + b_i^\mathcal{M}] = (1+\rho^\mathcal{M}(P^\mathcal{M},a_R^\mathcal{M}))\mu_i$, $i\in \mathcal{M}$? Indeed we need to use this to reduce $(a,b,P)$ to $(a,P)$ as before?} \takwacomment{I forgot this point, I added some lines for that.}
Below, we define the notion of subgroup formation in the context of the current work.
\begin{definition}[Subgroup formation]\label{Def:subgroupformation}
  A subgroup formation is {viable} if there exists a subgroup $\emptyset\neq\mathcal{M}\subsetneq\mathcal{N}$ and a corresponding contract $(\bm{a}^{\mathcal{M}},P^{\mathcal{M}})$ such that $ \mathbb{E}[U_R(W^{\mathcal{M}}_R(\bm{a}^{\mathcal{M}},P^{\mathcal{M}}))]\geq\hat{u}_R^{\mathcal{N}}$, and $\mathbb{E}[U_i(W^{\mathcal{M}}_i(\bm{a}^{\mathcal{M}},P^{\mathcal{M}}))]\geq\hat{u}_i^{\mathcal{N}}$ for all $i\in\mathcal{M}$.
    % \begin{align}\label{Eq:subgroupwelfare}
    %     \begin{split}
    %         &\mathbb{E}[U_i(W^{\mathcal{M}}_i(\bm{a}^{\mathcal{M}},P^{\mathcal{M}}))]\geq\hat{u}_i^{\mathcal{N}}
    %     \quad\text{for all }i\in \mathcal{M},\quad\text{and}\quad
    %     \mathbb{E}[U_R(W^{\mathcal{M}}_R(\bm{a}^{\mathcal{M}},P^{\mathcal{M}}))]\geq\hat{u}_R^{\mathcal{N}}.
    %     \end{split}
    % \end{align}
\end{definition}
Since the grand-coalition contract is strictly feasible, $\hat{u}_R^{\mathcal{N}}>d_R$ and $\hat{u}_i^{\mathcal{N}}>d_i$ for all $i\in\mathcal{N}$. Hence, any viable subgroup satisfying the condition in Definition \ref{Def:subgroupformation} also strictly satisfies the original disagreement-point IR constraints for all its members. Note also that in Definition \ref{Def:subgroupformation}, we {do not require strict welfare improvements for the P2P reinsurer and peers in the subgroup}.
%only impose weak welfare improvements at the P2P reinsurer and peers' welfare levels when considering the subgroup formation, 
In other words, it does not exclude the case in which all agents in the subgroup only match the welfare under the grand coalition $\mathcal{N}$. In practice, a smaller group may reduce administrative costs and is thus preferred, provided that welfare does not deteriorate. 

For any subgroup $\mathcal{M}\subsetneq \mathcal{N}$, fix $\bm{a}^{\mathcal{M}}\in\Delta^{|\mathcal{M}|}$ and let $\hat{P}_{\min}^\mathcal{M}(\bm{a}^{\mathcal{M}})$ be the lowest acceptable premium that the P2P reinsurer would charge so that she is willing to deviate from the grand coalition {$\mathcal{N}$}. Mathematically, $\hat{P}_{\min}^\mathcal{M}(\bm{a}^{\mathcal{M}})$ is defined by the solution of the equation 
{\small\begin{equation}\label{Eq:PminM}
    \mathbb{E}\left[U_R(W_R^{\mathcal{M}}(\bm{a}^{\mathcal{M}},\hat{P}_{\min}^{\mathcal{M}}(\bm{a}^{\mathcal{M}})))\right]=\hat{u}_R^{\mathcal{N}}.
\end{equation}}%
Since $\hat u_R^{\mathcal N}>U_R(w_R)$ and the expression tends to
$\sup_x U_R(x)$ as the premium tends to infinity, the required solution
exists under the maintained expected-utility assumptions.
Hence, for any admissible $\bm{a}^{\mathcal{M}}\in\Delta^{|\mathcal{M}|}$, the corresponding expected utility level for peer $i\in \mathcal{M}$ is $\hat{\Psi}_i^{\mathcal{M}}(\bm{a}^{\mathcal{M}}):=\mathbb{E}[U_i(W^{\mathcal{M}}_i(\bm{a}^{\mathcal{M}},\hat{P}_{\min}^{\mathcal{M}}(\bm{a}^{\mathcal{M}})))]$.
% \begin{equation*}
% \hat{\Psi}_i^{\mathcal{M}}(\bm{a}^{\mathcal{M}}):=\mathbb{E}[U_i(W^{\mathcal{M}}_i(\bm{a}^{\mathcal{M}},\hat{P}_{\min}^{\mathcal{M}}(\bm{a}^{\mathcal{M}})))].
% \end{equation*}
As peers' expected utility levels are strictly decreasing in the premium payment,
%evaluating at $\hat P_{\min}^{\mathcal{M}}(\bm a^{\mathcal{M}})$ yields 
$\hat{\Psi}_i^{\mathcal{M}}(\bm{a}^{\mathcal{M}})$ is the highest welfare for Peer $i$ {under $\bm{a}^\mathcal{M}$ given that the P2P reinsurer is indifferent between $\mathcal{N}$ and $\mathcal{M}$}, and 
%compatible with the P2P reinsurer weakly preferring the subgroup. 
%Thus, 
$\sup_{\bm{a}^{\mathcal{M}}\in\Delta^{|\mathcal{M}|}}\hat{\Psi}_i^{\mathcal{M}}(\bm{a}^{\mathcal{M}})$ can be interpreted as the highest welfare level attainable by Peer $i$ when the subgroup $\mathcal{M}$ is formed and the P2P reinsurer is kept indifferent between $\mathcal{M}$ and the grand coalition $\mathcal{N}$.
% To proceed with the main results in this subsection, we need the following assumption:
% \begin{assumption}\label{Ass:SM}
%     For any $\emptyset\neq \mathcal{M}\subseteq \mathcal{N}$, $S_{\mathcal{M}}:=\sum_{i\in \mathcal{M}}X_i$ is non-degenerate.
% \end{assumption}
%\Timcomment{Do we need this assumption, still?}\takwacomment{ I removed the assumption as we don't need that anymore for the strict side of the bounds.}

Below, we develop a bound for $\sup_{\bm{a}^{\mathcal{M}}\in\Delta^{|\mathcal{M}|}}\hat{\Psi}_i^{\mathcal{M}}(\bm{a}^{\mathcal{M}})$.
\begin{lemma}\label{Lem:BoundPsiM}
% Suppose that Assumptions \ref{Ass:non_degenerate_S}, \ref{Ass:EU}, \ref{pp:V:continuous} \ref{Ass:homo} and Condition \eqref{Eq:supH} hold, and let $(\bm{a}^*,P^*)$ be the unique optimal solution to Problem \eqref{Prob:hetero} from Theorem \ref{Thm:EUsol}, Then,

    % Under the maintained assumptions in Section \ref{sec:stability_price_fair}, 
    Assume that Assumption 
    % \ref{Ass:non_degenerate_S}, 
    \ref{Ass:EU} 
    % and \ref{P_max_assumption}, 
    and Condition \eqref{Eq:supH} hold. Then, 
    for any $\emptyset\neq\mathcal{M}\subsetneq \mathcal{N}$,
    % and any $i\in \mathcal{M}$, 
    % \begin{equation}\label{Eq:BoundPsiM}
    %     \sup_{\bm{a}^{\mathcal{M}}\in\Delta^{|\mathcal{M}|}}\hat{\Psi}_i^{\mathcal{M}}(\bm{a}^{\mathcal{M}})\leq U_i\left(w_i-\left(1+\frac{\Pi}{\mathbb{E}[S]-\min_{j\in \mathcal{N}\backslash\{i\}}\mu_j}\right)\mu_i\right).
    % \end{equation}
     {\small{\begin{equation}\label{Eq:BoundPsiM}
        \sup_{\bm{a}^{\mathcal{M}}\in\Delta^{|\mathcal{M}|}}\hat{\Psi}_i^{\mathcal{M}}\left(\bm{a}^{\mathcal{M}}\right)  \leq U_i\left(w_i-\left(1+\frac{\Pi}{\mathbb{E}[S_\mathcal{M}]}\right)\mu_i\right)\quad\text{for all }i\in\mathcal{M}.
    \end{equation}}}%
\end{lemma}

\begin{proof}
    See Appendix \ref{App:BoundPsiM}.
\end{proof}

The term {$\Pi/\mathbb{E}[S_\mathcal{M}]$} 
% \takwacomment{Here the main goal is indeed the subgroup-independent bound, your bound is correct, but it is subgroup-dependent and thus suffers from the curse of dimensionality.} \KTHNcomment{This bound is slightly more relaxed, and we don't need the subgroup-independent bound until the next statement. Indeed, we can slightly ``generalize" Theorem 3.13 to saying that $\mathcal{M}$ is not feasible if $\vartheta_i < \Pi/\mathbb{E}[S_{\mathcal{M}}]$; in particular, no subgroup is feasible under \eqref{Eq:thetaCond2}. But is that useful to describe? Maybe not.. } \takwacomment{It's possible to use the curse of dimension in $\vartheta_i < \Pi/\mathbb{E}[S_{\mathcal{M}}]$ to motivate \eqref{Eq:varthetaUpperBound}, maybe your way is better with this, as the subgroup-independent bound sounds like a jump and readers might not know why we need that.}
%\Pi/(\mathbb{E}[S]-\min_{j\in \mathcal{N}\backslash\{i\}}\mu_j)$ 
appearing on the right-hand side of \eqref{Eq:BoundPsiM} can be interpreted as a lower bound on the indifference loading for peers in any subgroup $\mathcal{M}\subseteq \mathcal{N}$. If this lower bound is so high that $\vartheta_i < \Pi/\mathbb{E}[S_{\mathcal{M}}]$, meaning a potential overpricing of the contract for peers in $\mathcal{M}$, then peers have no incentive to depart from the grand coalition $\mathcal{N}$. 
To rule out viable subgroup formations, one must check the condition $\vartheta_i < \Pi/\mathbb{E}[S_{\mathcal{M}}]$ for all $2^n-2$ nonempty proper subgroups, which is computationally expensive for large $n$.
The following theorem provides a subgroup-independent sufficient condition for ruling out viable subgroup formation.
% reinforces this idea by comparing the lower bound with the indifference loadings $\vartheta_i$ defined in \eqref{Eq:vartheta},
% and thus provides a sufficient condition to prevent subgroup formation.

\begin{theorem}\label{Thm:core}
    Suppose that Assumption 
    % \ref{Ass:non_degenerate_S}, 
    \ref{Ass:EU} 
    % and \ref{P_max_assumption}, 
    and Condition \eqref{Eq:supH} hold. 
    %In addition, 
    Also, assume that for each peer $i\in \mathcal{N}$, 
    {\small\begin{equation}\label{Eq:varthetaUpperBound}
     \vartheta_i < \frac{\Pi}{\mathbb{E}[S]-\min_{j\in \mathcal{N}\backslash\{i\}}\mu_j}.
     %\KTHNcomment{ \vartheta_i <  \min_{ i\in \mathcal{M},\,  \mathcal{M}\subsetneq \mathcal{N}  }  \frac{\Pi}{\mathbb{E}[S_\mathcal{M}]} = \frac{\Pi}{\mathbb{E}[S] - \mu_i},}
    \end{equation}}%
    Then, there is no viable subgroup formation. In particular, the conclusion holds if the centralized insurer's loading $\theta$ satisfies 
     {\small \begin{equation}\label{Eq:thetaCond2}
        \theta\leq \frac{\Pi}{\mathbb{E}[S]-\min_{j\in \mathcal{N}}\mu_j}. 
    \end{equation}}%
\end{theorem}

\begin{proof}
    See Appendix \ref{App:core}.
\end{proof}

By interpreting the term on the right-hand side of \eqref{Eq:varthetaUpperBound} as a lower bound to the indifference loading for the subgroup $\mathcal{M}$, Theorem \ref{Thm:core} implies that there is no viable subgroup when this bound is larger than the indifference loadings for each peer in the group $\mathcal{N}$. In other words, peers cannot obtain a better price offer (i.e., a lower CE loading) by forming a subgroup, and therefore they will remain in the grand coalition $\mathcal{N}$. In particular, Condition \eqref{Eq:thetaCond2} implies that any subgroup can do no better than the centralized insurer’s offer and, therefore, no better than the grand coalition $\mathcal{N}$.

Conditions \eqref{Eq:varthetaUpperBound} and \eqref{Eq:thetaCond2} can also be viewed as requiring the P2P reinsurer's CE net profit $\Pi$ in the grand coalition to be sufficiently large. 
%value of the P2P reinsurer's welfare in the grand coalition,
This CE net profit is the P2P reinsurer's opportunity cost of leaving the grand coalition to serve a smaller subgroup. As such, if the optimal utility of the P2P reinsurer under the grand coalition $\mathcal{N}$ is sufficiently large, no subgroups can be formed to pay such a high price to incentivize the P2P reinsurer to provide service to a smaller group.

% \subsection{No price fairness case under exponential utility setting}

\subsection{Stability without price fairness}\label{sec:stability_wo_price_fairnes}
We study the stability of the solution to Problem \eqref{Prob:NoFairness} in this subsection. Throughout Section \ref{sec:stability_wo_price_fairnes}, we impose Assumption
% \ref{Ass:non_degenerate_S} and 
\ref{ass:mgf} and Conditions \eqref{Eq:IRCond} and \eqref{eq:Cond:exp:no:fairness} so that Proposition \ref{Prop:NoFairExpSol} supplies the unique optimal contract $(\bm{a}^*,a_R^*,\bm{b}^*,P^*)$.
As in Section \ref{sec:stability_price_fair}, we assume that the P2P reinsurer must be included in every viable subgroup. Let $\Tilde{u}_R^{\mathcal{N}}$ and $\Tilde{u}_i^{\mathcal{N}}$ be optimal expected utility levels of the P2P reinsurer and Peer $i$ under the optimal contract $(\bm{a}^*,a_R^*,\bm{b}^*,P^*)$,
% depicted in Proposition \ref{Prop:NoFairExpSol}, 
respectively. Following the same spirit of \eqref{Eq:vartheta} and \eqref{Eq:Pi}, we associate them with the CE artificial loading $\Tilde{\vartheta}_i$ and the CE net profit $\Tilde{\Pi}$, defined as $\Tilde{\Pi}:=U_R^{-1}\left(\Tilde{u}_R^{\mathcal{N}}\right)-w_R$, and $\Tilde{\vartheta}_i=\frac{w_i-U_i^{-1}\left(\Tilde{u}_i^{\mathcal{N}}\right)}{\mu_i}-1$, for any $i\in\mathcal{N}$. The wealth levels of each agent in the subgroup $\mathcal{M}$ {under the risk-bearing-payment strategy $(\bm{a}^{\mathcal{M}},\bm{b}^{\mathcal{M}})\in \Delta^{|\mathcal{M}|}\times\mathbb{R}^{|\mathcal{M}|}$}  
are defined as
{\small\begin{align*}
\Tilde{W}_i^{\mathcal{M}}(\bm{a}^{\mathcal{M}},\bm{b}^{\mathcal{M}})&:=w_i-a_i^{\mathcal{M}}S_{\mathcal{M}}-b_i^{\mathcal{M}}, i\in \mathcal{M},\quad  
\Tilde{W}_R^{\mathcal{M}}(\bm{a}^{\mathcal{M}},\bm{b}^{\mathcal{M}}):=w_R+\sum_{i\in \mathcal{M}}b_i^{\mathcal{M}}-(1+\tau)\left(1-\sum_{i\in \mathcal{M}}a_i^{\mathcal{M}}\right)S_{\mathcal{M}},
\end{align*}}%
where we have omitted the components $a_R^\mathcal{M}$ and $P^\mathcal{M}$, since they are completely determined by $\bm{a}^\mathcal{M}$ and $\bm{b}^\mathcal{M}$ via the market-clearing and aggregate premium conditions, respectively.

Similar to Definition \ref{Def:subgroupformation}, we define the viable subgroup in the context of this section, which also guarantees IR constraints are strictly satisfied for all viable subgroup members. 

\begin{definition}[Subgroup formation without price fairness]\label{Def:subgroupformationNG}
    A subgroup $\emptyset\neq\mathcal{M}\subsetneq \mathcal{N}$ is called viable if there exists $(\bm{a}^{\mathcal{M}},\bm{b}^{\mathcal{M}})\in\Delta^{|\mathcal{M}|}\times\mathbb{R}^{|\mathcal{M}|}$ such that for all $j\in \mathcal{M}\cup\{R\}$, $\mathbb{E}[U_j(\Tilde{W}_j^{\mathcal{M}}(\bm{a}^{\mathcal{M}},\bm{b}^{\mathcal{M}}))]\geq \Tilde{u}_j^{\mathcal{N}}$. 
    % and $\mathbb{E}[U_R(\Tilde{W}_R^{\mathcal{M}}(\bm{a}^{\mathcal{M}},\bm{b}^{\mathcal{M}}))]\geq \Tilde{u}_R^{\mathcal{N}}$. 
%     \begin{equation}\label{Eq:FeaSubgroupNF}
% \mathbb{E}[U_i(\Tilde{W}_i^{\mathcal{M}}(\bm{a}^{\mathcal{M}},\bm{b}^{\mathcal{M}}))]\geq \Tilde{u}_i^{\mathcal{N}}\quad\text{and}\quad\mathbb{E}[U_R(\Tilde{W}_R^{\mathcal{M}}(\bm{a}^{\mathcal{M}},\bm{b}^{\mathcal{M}}))]\geq \Tilde{u}_R^{\mathcal{N}}.
% \end{equation}
\end{definition}

In Theorem \ref{Thm:core}, the price-fairness condition dictates a side-payment level for each peer. 
Here, the unboundedness of the side payments $b_i$ in Problem \eqref{Prob:NoFairness} allows us only to develop a subgroup-independent aggregate condition ensuring that peers in the subgroup $\mathcal{M}$ cannot induce the P2P reinsurer to deviate from the grand coalition. This condition is presented in the following proposition.

\begin{proposition}\label{Prop:core2}
    % Under the stated preconditions 
 %   \KTHNcomment{perhaps just write them explicitly if not too troublesome?}  
    % at the beginning of Section \ref{sec:stability_wo_price_fairnes}, 
Suppose that Assumption
% \ref{Ass:non_degenerate_S} and 
\ref{ass:mgf}, Conditions \eqref{Eq:IRCond} and \eqref{eq:Cond:exp:no:fairness}
hold. If, in addition, $\Tilde{\Pi} > \max_{\emptyset \neq \mathcal{M}\subsetneq \mathcal{N}} \sum_{i\in \mathcal{M}}\Tilde{\vartheta}_i \mu_i$,  then there is no viable subgroup formation. 
In particular, 
%Particularly,
the condition holds if $\Tilde{\Pi} > \theta\mathbb{E}[S]$. 
\end{proposition}

\begin{proof}
    See Appendix \ref{App:core2}.
\end{proof}

\section{Numerical analysis}\label{Sec:Numeric}

%In this section, we conduct a numerical analysis 
%to study the impact of price fairness under the exponential utility setting. Furthermore, we investigate the impact of pool size under the assumption of homogeneity.

In this section, we conduct a numerical analysis under exponential utilities to investigate the impact of price fairness and the effect of pool size under homogeneous peers.

\subsection{Impact of price fairness}\label{sec:numerical_price_fairness}

In this subsection, we consider a three-peer example to investigate the impact of price fairness on the P2P insurance contract and all agents' welfare. Peers' losses are modeled by independent Gamma distributions, i.e., $X_i\sim \mathrm{Gamma}(\kappa_i,s_i)$, $i=1,2,3$, where $\kappa_i$ and $s_i$ represent the shape and scale parameters, respectively. In the sequel, we take $(\kappa_1,\kappa_2,\kappa_3)=(0.69,0.73,1.46)$ and $(s_1,s_2,s_3)=(1.7,1.85,1.09)$. Therefore, we have $(\mu_1,\mu_2,\mu_3)=(1.173,1.3505,1.5914)$ and $\mathbb{E}[S]=4.1149$.
% and \KTHNcomment{(seems the strategy only depends on $S$ and $\mu_i$ but not higher individual moments, so we may skip writing the variance.)} \takwacomment{I mentioned at the end that variance is irrelevant (so I need some proof for that); you can decide when you finish the last part.}
% \begin{equation*}
% (Var(X_1),Var(X_2),Var(X_3))=(2.2870,3.9439,1.7346).
% \end{equation*}

We equip all agents with exponential utilities, where agents' initial wealth does not affect the numerical analysis. The price-fairness case is solved using the characterization in Theorem \ref{Thm:EUsol}, whereas the no-price-fairness benchmark is solved using Proposition \ref{Prop:NoFairExpSol}. 
%In this case, their initial wealth levels would not influence the numerical analysis.
%, thereby making the results more compact.
%For bargaining power, 
We set $\delta_i= {\delta :=} (1-\delta_R)/3$ so that all peers have equal bargaining power. Other baseline parameters are reported in Table \ref{Tab:baseline2}. We note that the parameter choices, including the varying values of $\delta_R$ and $\theta$ considered below, satisfy the conditions that prevent subgroup formation with and without price fairness (see Theorem \ref{Thm:core} and Proposition \ref{Prop:core2}). 

    \begin{table}[H]
\centering
\begin{tabular}{cccccccc}
%\centering
\toprule
$\gamma_1$ & $\gamma_2$ & $\gamma_3$& $\gamma_R$ & $\tau$ & $\theta$ & $\delta_R$ & {$\delta$}  \\ \hline
0.33        & 0.29 & 0.51     & 0.14        & 0.025       & 0.4   & 0.6      & 0.1333          \\ \bottomrule
\end{tabular}
\caption{Baseline setting of agents' preferences.}
\label{Tab:baseline2}
\end{table}

\subsubsection{Comparison of P2P insurance contracts}

\begin{table}[htbp]
\centering
\begin{tabular}{ccccccccc}
%{|c|c|c|c|c|c|c|c|c|}
\toprule
$\delta_R$ & $a_R^*$ & $a_1^*$ & $a_2^*$ & $a_3^*$ & $b_1^*$ & $b_2^*$ & $b_3^*$ & $P^*$ \\ \midrule
0.6 & 0.40373 & 0.19088 & 0.23637 & 0.16901 & 0.65587 & 0.68678 & 1.25999 & 2.60264 \\ 
0.65 & 0.40378 & 0.19116 & 0.23589 & 0.16917 & 0.67367 & 0.71057 & 1.28503 & 2.66926 \\ 
0.7 & 0.40383 & 0.19145 & 0.23528 & 0.16944 & 0.69156 & 0.73508 & 1.30984 & 2.73648 \\ 
0.75 & 0.4039 & 0.19177 & 0.23448 & 0.16984 & 0.70955 & 0.76059 & 1.33435 & 2.8045 \\ 
0.8 & 0.40399 & 0.19212 & 0.23343 & 0.17046 & 0.72769 & 0.78746 & 1.35839 & 2.87354 \\ 
0.85 & 0.40411 & 0.19252 & 0.232 & 0.17137 & 0.74597 & 0.81625 & 1.38162 & 2.94384 \\ 
\bottomrule
\end{tabular}

\caption{Contract results for different values of $\delta_R$ under price fairness. 
%\KTHNcomment{What about plot it in 2 figures, both with $\delta_R$ as the $x$-axis. First figure consists of the curves $a_R,a_i$; and the second consists of $\bm{b}$ and $P$. Which one would be better? But to compare with no price fairness, one may need to combine two figures}\takwacomment{That will become visually indifferent, unless we do some trick on the Y-axis and skip the nonplotting range.} \KTHNcomment{Your call. We can start the y-axis at 0.1x if we want to plot.}
}
\label{tab:delta_r_-_price_fairness}
\end{table}
\begin{table}[htbp]
\centering
\begin{tabular}{ccccccccc}
\toprule
$\theta$ & $a_R^*$ & $a_1^*$ & $a_2^*$ & $a_3^*$ & $b_1^*$ & $b_2^*$ & $b_3^*$ & $P^*$ \\ \midrule
0.38 & 0.40375 & 0.19102 & 0.23614 & 0.16908 & 0.64146 & 0.67179 & 1.24092 & 2.55417 \\ 
0.39 & 0.40374 & 0.19095 & 0.23626 & 0.16904 & 0.64867 & 0.67927 & 1.25047 & 2.5784 \\ 
0.4 & 0.40373 & 0.19088 & 0.23637 & 0.16901 & 0.65587 & 0.68678 & 1.25999 & 2.60264 \\ 
0.41 & 0.40372 & 0.19082 & 0.23648 & 0.16898 & 0.66306 & 0.69431 & 1.26948 & 2.62686 \\ 
0.42 & 0.40372 & 0.19075 & 0.23658 & 0.16896 & 0.67025 & 0.70186 & 1.27895 & 2.65107 \\ 
0.43 & 0.40371 & 0.19068 & 0.23667 & 0.16894 & 0.67744 & 0.70944 & 1.2884 & 2.67528 \\ 
\bottomrule
\end{tabular}

\caption{Contract results for different values of $\theta$ under price fairness.}
\label{tab:theta_-_price_fairness}
\end{table}

\begin{table}[!ht]
    \centering
    
    \begin{subtable}[t]{0.5\textwidth}
    \centering
     \begin{tabular}{ccccc}
\toprule
$\delta_R$ & $b_1^*$ & $b_2^*$ & $b_3^*$ & $P^*$ \\ \midrule
0.6 & 0.51091 & 0.62855 & 1.43273 & 2.57219 \\ 
0.65 & 0.53331 & 0.6511 & 1.45447 & 2.63888 \\ 
0.7 & 0.55588 & 0.6738 & 1.47644 & 2.70612 \\ 
0.75 & 0.57867 & 0.69671 & 1.49871 & 2.77409 \\ 
0.8 & 0.60175 & 0.7199 & 1.52136 & 2.84301 \\ 
0.85 & 0.6252 & 0.74343 & 1.54447 & 2.9131 \\ \bottomrule
\end{tabular}
\caption{Impact of $\delta_R$.}
    \end{subtable}\hfill
    \begin{subtable}[t]{0.495\textwidth}
    \centering
        \begin{tabular}{ccccc}
\toprule
$\theta$ & $b_1^*$ & $b_2^*$ & $b_3^*$ & $P^*$ \\ \midrule
0.38 & 0.4988 & 0.61297 & 1.41191 & 2.52368 \\ 
0.39 & 0.50486 & 0.62076 & 1.42232 & 2.54794 \\ 
0.4 & 0.51091 & 0.62855 & 1.43273 & 2.57219 \\ 
0.41 & 0.51695 & 0.63632 & 1.44315 & 2.59642 \\ 
0.42 & 0.52299 & 0.6441 & 1.45356 & 2.62065 \\ 
0.43 & 0.52903 & 0.65186 & 1.46397 & 2.64486 \\ \bottomrule
\end{tabular}
\caption{Impact of $\theta$.}
    \end{subtable}
    \caption{Contract results without price fairness, where $a_R^*=0.39671$, $a_1^*=0.21662$, $a_2^*=0.2465$, and $a_3^*=0.14017$; these risk-sharing proportions are invariant to changes in $\delta_R$ and $\theta$.}
    \label{tab:NPFcontract}
\end{table}

% \begin{itemize}
%     \item Table \ref{tab:delta_r_-_price_fairness}-\ref{tab:theta_-_price_fairness} show that price fairness does not change agents' risk bearing significantly, even if $\delta_R$ and $\theta$ alternate.
%     \item With PF, $a_R^*$ increases slightly when $\delta_R$ increases and decreases slightly when $\theta$ increases.
%     \item Comparing Table \ref{tab:delta_r_-_price_fairness}-\ref{tab:theta_-_price_fairness} with Table \ref{tab:NPFcontract}, $a_R^*$ is higher than the NPF counterpart.
%     \item Note that according to Proposition \ref{Prop:NoFairExpSol}, without price fairness, agents' risk bearing does not respond to the change of $\delta_R$ and $\theta$, shown in Table \ref{tab:NPFcontract}.
%     \item With and without PF, all peers' side payments increase with $\delta_R$ and $\theta$, so does the P2P insurance premium. 
% \end{itemize}

%%%% contrast btwn price fairness or not 
%We first examine the impact of price fairness on the P2P insurance contracts.

%% general observations 
Tables \ref{tab:delta_r_-_price_fairness} and \ref{tab:theta_-_price_fairness} report the optimal P2P insurance contracts for different values of $\delta_R$ and $\theta$, respectively, under the price-fairness condition \eqref{Eq:PricingFairness}. The results in the absence of the fairness condition are depicted in Table \ref{tab:NPFcontract}.
%The results are shown in Tables \ref{tab:delta_r_-_price_fairness}-\ref{tab:NPFcontract} with different $\delta_R$ and $\theta$. 
% At first glance, agents with lower risk aversion tend to bear larger shares of the aggregate risk. \KTHNcomment{(would this be affected by the differences in their mean as well? which one is more prominent or should we disclose both?)} \takwacomment{The most prominent effect lies in $\gamma_i$ and $\gamma_R$, $\mu_i$ will impact on it but its effect is rather minor and enters via price-fairness condition, when comparing to the no price fairness case.} \KTHNcomment{Can we observe this from the results?} \takwacomment{Or I put it in this way: without price fairness, the risk sharing $a_i$ only depends on the risk aversion parameters and $\tau$ (see (26)), and then the changes will be due to additional involved parameters, so even the effect of $\mu_i$ can be obscure. So, the rather safe choice is not to assert it as we don't have enough evidence for that.} Therefore, price fairness does not significantly alter the risk-sharing pattern. 
From the tables, we see that the risk-sharing pattern is primarily driven by the agents' risk aversion levels, where agents with higher risk aversion parameters (Peers 1 and 3) tend to bear smaller shares of the aggregate risk. {Indeed, in the absence of price fairness, $\bm{a}^*$ depends solely on $\tau$ and the risk aversion parameters; see \eqref{Eq:OptNFaRa}.
In addition, peers that bear a larger share of the aggregate risk do not
necessarily contribute smaller side payments. In particular, Peer 2 bears the largest risk exposure but contributes more than Peer 1. According to \eqref{Eq:PricingFairness}, each peer's side payment under price fairness depends on both the peer's ex ante expected loss and risk-bearing proportion.
Although $a_2^*>a_1^*$ reduces Peer 2's required side payment relative to Peer 1, this effect is outweighed by $\mu_2>\mu_1$, resulting in $b_2^*>b_1^*$.
% Peers that bear a larger share of the aggregate risk contribute smaller side payments, and vice versa; in particular, Peer 2, who bears the largest risk exposure, contributes the least.
} 
%\takwacomment{Larger $a_i$ leads to smaller $b_i$ seems corresponding to price fairness case?}

%% impact of price fariness
{Examining the impact of the price-fairness condition, we see that the dependence of $\bm{a}^*$ on $\delta_R$ and $\theta$, albeit modest, stands in contrast to the case without price fairness, where the risk-sharing allocation is completely independent of these two parameters.}
%Nonetheless, it is observed that the P2P reinsurer and Peer 2, who are less risk-averse, bear more risk under the price fairness contracts; on the other hand, Peer 1 and Peer 3, who are more risk-averse, bear less risk accordingly. \KTHNcomment{(same as above)} 
In addition, the price-fairness condition limits the flexibility to use side payments to adjust premiums and risk allocations. Consequently, the variation in risk allocations among peers is reduced relative to the case without price fairness. Regarding the side payments, we observe that the payments made by Peer 1 and Peer 2
drop when price fairness is removed, whereas Peer 3's side
payment increases significantly.
% and Peer 1's payment changes only slightly. 
The aggregate
premium is lower without price fairness in the baseline case. Equivalently, imposing price fairness raises the total premium and shifts part of the payment burden from the most risk-averse peer, Peer 3, toward the less risk-averse Peers 1 and  2. Thus, price fairness limits cross-subsidization based on expected losses, but heterogeneous risk aversion can still generate welfare redistribution in certainty-equivalent terms.

%%%% effect of \delta_R on a and b
With the presence of price fairness, as $\delta_R$ increases, greater power is placed on the P2P reinsurer’s utility in the bargaining solution, leading to a more favorable contract for the P2P reinsurer. In the present setting, where the P2P reinsurer is less risk-averse than all peers,
% such a contract need not correspond to lower risk exposure. 
% Instead, 
%\takwacomment{The commented sentence is rather weird that less risk aversion usually lead to higher risk exposure.}
she is willing to assume a larger share of the aggregate risk in exchange for higher compensation. This explains the observed increase in $a_R^*$ and in the side payments for all peers as $\delta_R$ increases, with the latter effect occurring regardless of whether the price-fairness condition is imposed. Consequently, the aggregate risk retained by the peer group decreases. Within the peer group, the adjustment is heterogeneous: the risk-bearing share of Peer 2 decreases, whereas those of the other peers (i.e., $a_1^*$ and $a_3^*$) increase.
%: since Peer 2 initially bears the largest share of risk among peers, its allocation decreases the most. \takwacomment{I am not sure the causality implied by this sentence.} \KTHNcomment{any alternative explanations on why $a_2$ mostly bear the drop?} \KTHNcomment{run a quick test to see if thats the case by having a larger $\mu_2$?} 
% \takwacomment{One possible explanation: increasing $\delta_R$ has a similar effect like fairness because each peer has lees power and thus becoming unified in terms of risk bearing and side-payment, so $a_2^*$ decreases to match this unification.}
Although Peer 2 is the least risk-averse peer and bears the largest share of the retained aggregate risk, it also faces the highest CE loading under price fairness; see Fig.~\ref{fig:PF_vartheta_vs_deltaR}. As the P2P reinsurer's bargaining power and the premium increase, the optimal contract modestly reduces Peer 2's exposure and reallocates part of the retained risk to Peers 1 and 3, consistent with the weighted marginal risk-sharing conditions in \eqref{Eq:HeteroFOC}.

% and is therefore the most sensitive to the resulting premium increase. A further reduction in Peer 2’s relatively small surplus carries a comparatively large marginal cost. Consequently, the pairwise marginal risk-sharing conditions (see \eqref{Eq:HeteroFOC}) are restored by reducing $a_2^*$ and reallocating most of this risk to Peers 1 and 3.
}
%, these increases are dominated by the reduction in Peer 2’s allocation, resulting in a lower total risk retained by the peer group.  

%%%% effect of \theta 
{As the loading $\theta$ offered by the alternative centralized insurance increases, we observe that under the price-fairness condition, the more risk-averse agents (Peers 1 and 3) experience a decrease in risk exposure, while Peer 2 exhibits an increase in $a_2^*$. In general, the reinsurance share $a_R^*$ decreases slightly and the side payments for all peers increase regardless of the price-fairness condition. 
%This pattern reflects a heterogeneous reallocation of marginal risk-bearing capacity within the peer group. 
As the outside option becomes less attractive, the bargaining solution adjusts asymmetrically across agents: risk-averse peers (Peers 1 and 3) reduce their exposure due to their high marginal disutility of risk, while the least risk-averse peer (Peer 2) serves as the marginal risk absorber and takes on slightly higher risk. Although Peer 2 absorbs more risk, this increase more than offsets the reductions among the more risk-averse agents, resulting in a slight increase in total peer risk retention. This also explains the increase in premiums, as a higher loading $\theta$ makes the outside option less attractive, allowing the P2P arrangement to charge higher premiums while remaining preferable to the substitute centralized-insurance contract. }

%Regarding peers' side payments and the total premium payment,  Tables \ref{tab:delta_r_-_price_fairness}-\ref{tab:NPFcontract} show that all payments increase with $\delta_R$, showcasing that when the P2P reinsurer possesses a larger market power, it can earn a higher premium income. Also, all peers pay more when the centralized insurance becomes more expensive. This indicates that these two types of insurance contracts are substitutes. 

% Moreover, Peer 2, the least risk-averse peer, pays a significantly larger amount under the P2P insurance contracts with price fairness, while the most risk-averse Peer 3's side payment declines substantially, and Peer 1's side payment also drops mildly.  
% In other words, the adoption of price fairness results in cross-subsidization from less risk-averse peers to more risk-averse peers, given the increase in total premium payments with price fairness.

\subsubsection{Welfare analysis}

\begin{figure}[!ht]
    \centering
    \begin{subfigure}[t]{0.48\linewidth}
        \centering\includegraphics[scale=0.435]{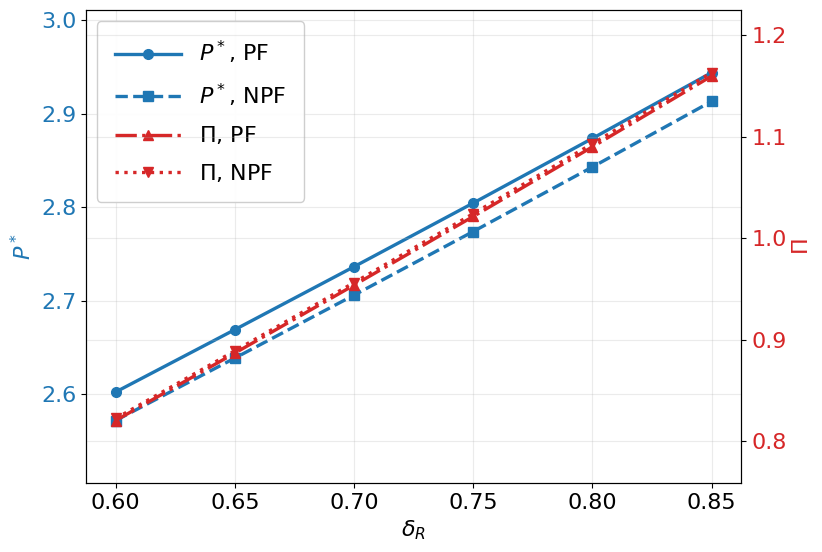}
    \subcaption{Impact of $\delta_R$.}
    \label{fig:P_Pi_vs_deltaR}
    \end{subfigure}
    \begin{subfigure}[t]{0.48\linewidth}
        \centering\includegraphics[scale=0.435]{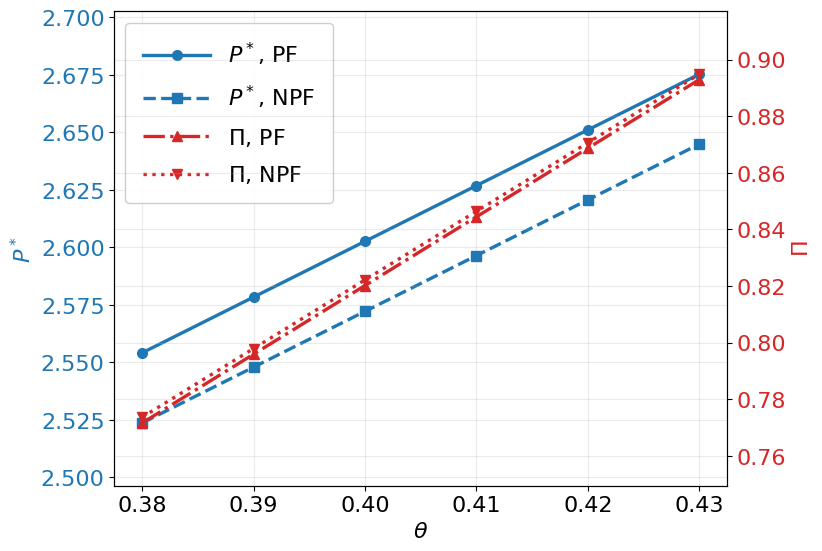}
    \subcaption{Impact of $\theta$.}
    \label{fig:P_Pi_vs_theta}
    \end{subfigure}
    \caption{Optimal premium and P2P reinsurer’s certainty-equivalent net profit with price fairness (PF) and without price fairness (NPF).}
    \label{fig:P_Pi}
\end{figure}

% \begin{itemize}
%     \item As shown in Fig.~\ref{fig:P_Pi}, with and without price fairness, $P^*$ and $\Pi$ increase with $\delta_R$ and $\theta$.
%     \item Notably, price fairness leads to higher premium income for the P2P reinsurer. 
%     \item Nonetheless, the P2P reinsurer's CE net profit is slightly higher in the absence of price fairness. Recall that $a_R^*$ is higher than the NPF counterpart; this means that a higher premium income cannot compensate for the extra risk bearing.
% \end{itemize}

We conduct a welfare analysis by comparing the {CE loadings} $\vartheta_i$ {for peers and the CE net profit $\Pi$ for the P2P reinsurer} (see \eqref{Eq:vartheta} and \eqref{Eq:Pi}) in different P2P insurance contracts with and without the price-fairness condition. Fig.~\ref{fig:P_Pi} plots the optimal premium $P^*$ and $\Pi$ against $\delta_R$ and $\theta$. We observe that, regardless of whether price fairness is imposed, $P^*$ and $\Pi$ increase with $\delta_R$ and $\theta$, where the former observation is consistent with the results shown in Tables \ref{tab:delta_r_-_price_fairness}-\ref{tab:NPFcontract}. This is intuitive, as both higher bargaining power for the P2P reinsurer and a less attractive centralized-insurance alternative strengthen the P2P reinsurer’s bargaining position, resulting in a higher CE profit. Notably, both panels show that imposing price fairness increases the P2P reinsurer’s premium income, despite lowering its CE profit.
%Nonetheless, the P2P reinsurer's CE net profit is slightly higher in the absence of price fairness, 
%informing a welfare loss.
Recall that $a_R^*$ is higher under price fairness (cf. Tables \ref{tab:delta_r_-_price_fairness}-\ref{tab:NPFcontract}), {which indicates that the extra risk taken by the P2P reinsurer outweighs the additional premium income, leading to an overall welfare loss.}

%this means that a higher premium income cannot compensate for the extra risk taking for the P2P reinsurer.

\begin{figure}[!ht]
    \centering
    \begin{subfigure}[t]{0.495\linewidth}
        \centering\includegraphics[scale=0.43]{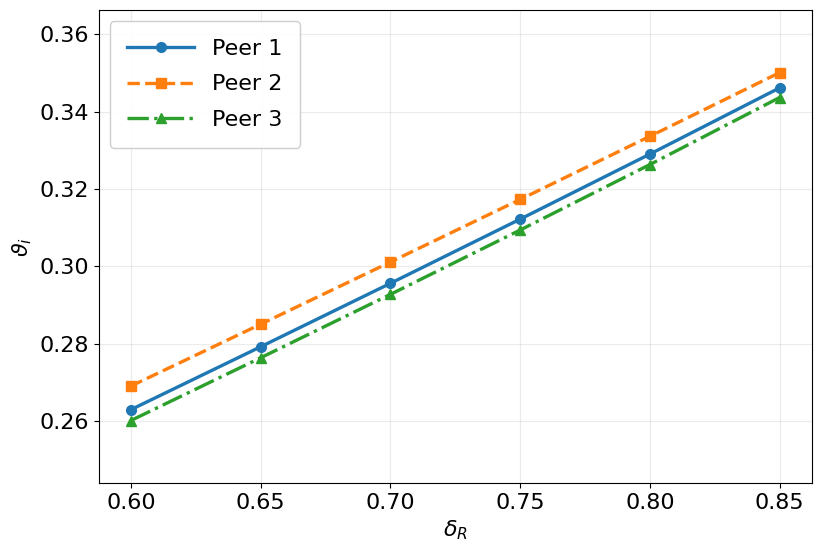}
    \subcaption{Impact of $\delta_R$ with price fairness.}
    \label{fig:PF_vartheta_vs_deltaR}
    \end{subfigure}
    \begin{subfigure}[t]{0.495\linewidth}
        \centering\includegraphics[scale=0.43]{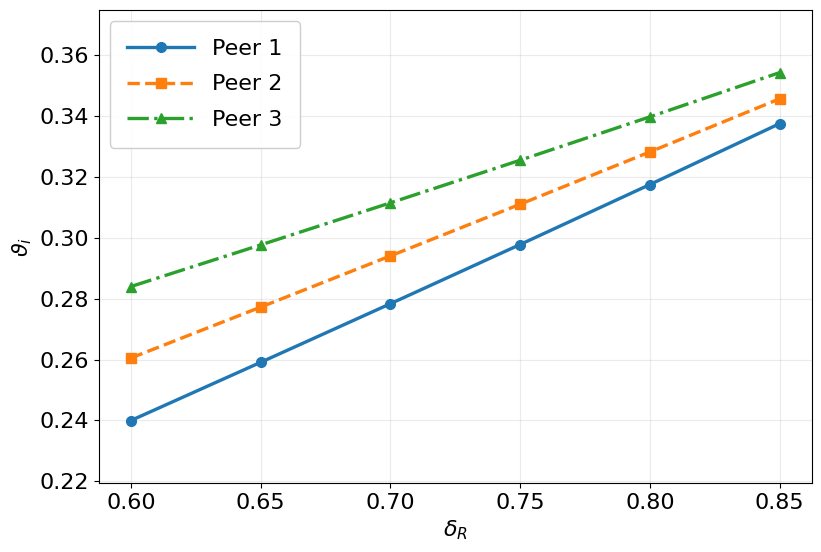}
    \subcaption{Impact of $\delta_R$ without price fairness.}
    \label{fig:NPF_vartheta_vs_deltaR}
    \end{subfigure}
    \begin{subfigure}[t]{0.495\linewidth}
        \centering\includegraphics[scale=0.43]{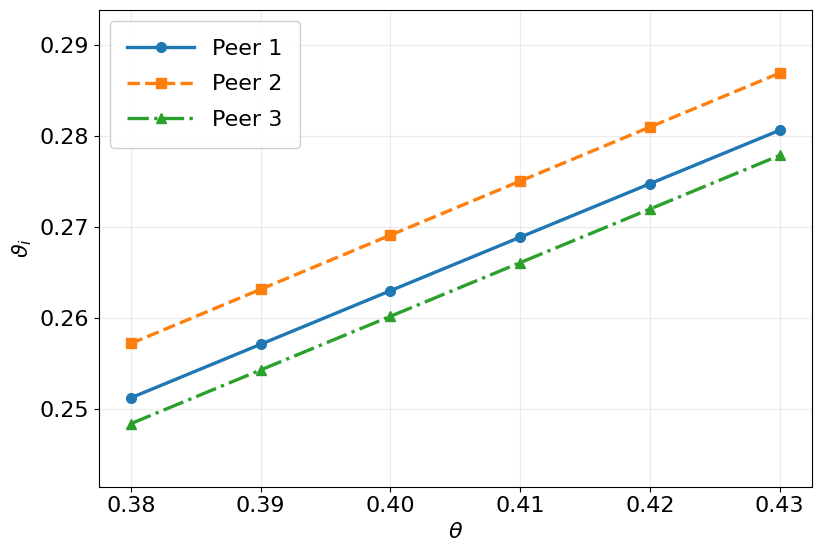}
    \subcaption{Impact of $\theta$ with price fairness.}
    \label{fig:PF_vartheta_vs_theta}
    \end{subfigure}
    \begin{subfigure}[t]{0.495\linewidth}
        \centering\includegraphics[scale=0.43]{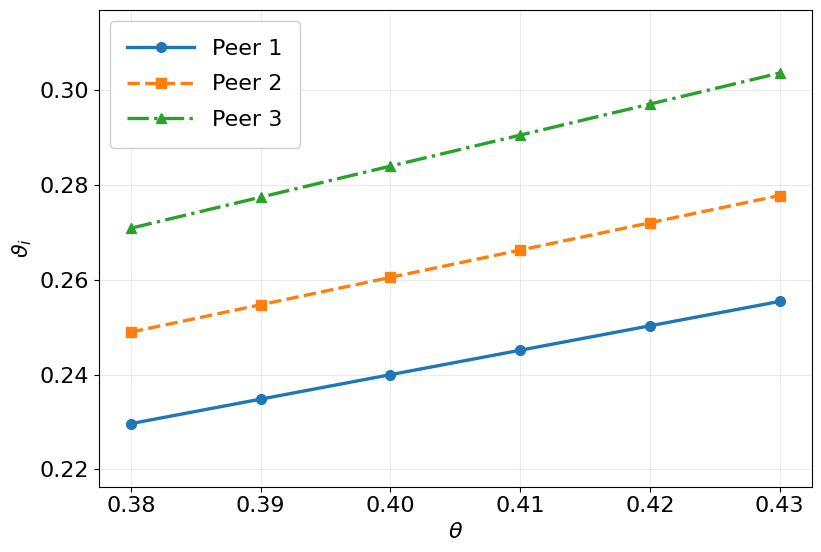}
    \subcaption{Impact of $\theta$ without price fairness.}
    \label{fig:NPF_vartheta_vs_theta}
    \end{subfigure}
    \caption{Certainty-equivalent loadings. 
    %\Timcomment{Use same $y$-axis, as they are similar}
    } 
    %\takwacomment{Done.}
    \label{fig:varthetai}
\end{figure}

% \begin{itemize}
%     \item In Fig.~\ref{fig:varthetai}, all peers' CE loadings increase with $\delta_R$ and $\theta$.
%     \item Juxtaposing Fig.~\ref{fig:PF_vartheta_vs_deltaR} and Fig.~\ref{fig:NPF_vartheta_vs_deltaR}, price fairness narrows the differences of CE loadings, which can also be observed when comparing Fig.~\ref{fig:PF_vartheta_vs_theta} and Fig.~\ref{fig:NPF_vartheta_vs_theta}.
%     \item The same narrowing effect can be seen in Fig.~\ref{fig:NPF_vartheta_vs_deltaR} when $\delta_R$ is large.
%     \item PF changes the order of paid CE loadings from $3>1>2$ to $2>1>3$.
%     \begin{itemize}
%         \item Without PF, the order $3>1>2$ correspond to the ranking of risk aversion. 
%         \item \takwacomment{The main reason of shifting is that it follows the order of $1/\mu_i$ as }
%         \begin{equation*}
%             \vartheta_i^{PF}=\rho+\frac{\ln\mathbb{E}[e^{\gamma_ia_i^*(S-\mathbb{E}[S])}]}{\mu_i\gamma_i}\approx\rho+\frac{\gamma_i(a_i^*)^2}{2\mu_i}Var(S),
%         \end{equation*}
%         \takwacomment{where the approximation follows by the Taylor expansion,
%         and the effect of $1/\mu_i$ dominates. Hence, the peer with more loss severity in terms of mean loss gets the cross subsidization due to the price fairness.}
%     \item When PF is imposed, the loss is taken into account: Peer 2's loss is more severe in terms of mean and variance, compared to Peer 1. {\color{red}Not the mean}
%     \end{itemize}
% \end{itemize}

Fig.~\ref{fig:varthetai} plots all peers' CE loadings against $\delta_R$ and $\theta$. In general, all peers' CE loadings increase with $\delta_R$ and $\theta$, implying welfare deterioration due to the P2P reinsurer's increasing bargaining power or more expensive alternatives. Recall from the observations in Tables \ref{tab:delta_r_-_price_fairness}-\ref{tab:NPFcontract} that the impact of $\delta_R$ and $\theta$ on risk sharing is modest. The drop in peers' welfare is thus mainly due to the higher side payments. 
%Along with the implications in Fig.~\ref{fig:P_Pi}, the P2P reinsurer benefits more from the contracts with larger $\delta_R$ and $\theta$ at the expense of peers' welfare. <-- a bit repititve so i hide it   

Comparing Fig.~\ref{fig:PF_vartheta_vs_deltaR} with Fig.~\ref{fig:NPF_vartheta_vs_deltaR}, and Fig.~\ref{fig:PF_vartheta_vs_theta} with Fig.~\ref{fig:NPF_vartheta_vs_theta}, respectively, we observe that price fairness narrows the differences in CE loadings among peers and lowers the variation of peers' welfare. 
In particular, this effect 
becomes more pronounced as $\delta_R$ increases. 
Moreover, peers' CE loadings converge as $\delta_R$ increases.
As the P2P reinsurer captures a larger share of the contractual surplus, the utility surplus
%distribution of welfare 
among peers becomes more homogeneous, leading to more similar CE loadings.

The last noteworthy observation concerns the ordering of CE loadings. The CE loading $\vartheta_3$ is the smallest when price fairness is imposed, but becomes the largest when price fairness is removed.
% $\vartheta_2$ becomes the largest while
% being the smallest without price fairness, and the reverse happens to
% Peer 3. 
This echoes the cross-subsidization effect observed earlier under the fairness condition, whereby the less risk-averse peers transfer welfare to the more risk-averse peers. To see this mathematically, first define $\rho_i$ as the solution to $\mathbb{E}[a_i^{*,NPF}S+b_i^{*,NPF}]=(1+\rho_i)\mu_i$. Then, by Taylor expansion, we obtain the following CE loadings with and without the price-fairness condition:
% \KTHNcomment{what is $\rho$ here? The $\rho$ in \eqref{Eq:PricingFairness}? } \takwacomment{Yes}
{\small\begin{align}\label{Eq:CEdecomposition}
%\begin{split}
    \vartheta_i^{PF}&= \rho\left(P^{*,PF},a_R^{*,PF}\right)  +\frac{\ln\mathbb{E}\left[e^{\gamma_ia_i^{*,PF}(S-\mathbb{E}[S])}\right]}{\mu_i\gamma_i}\approx \rho(P^{*,PF},a_R^{*,PF})  +\frac{\gamma_i}{2\mu_i}\operatorname{Var}\left(a_i^{*,PF}S\right),\nonumber\\
            \vartheta_i^{NPF}
            %&=\frac{\mathbb{E}[a_i^{*,NPF}X+b_i^{*,NPF}]}{\mu_i}-1+\frac{\ln\mathbb{E}[e^{\gamma_ia_i^{*,NPF}(S-\mathbb{E}[S])}]}{\mu_i\gamma_i}\\
            &= \rho_i+\frac{\ln\mathbb{E}\left[e^{\gamma_ia_i^{*,NPF}(S-\mathbb{E}[S])}\right]}{\mu_i\gamma_i}\approx\rho_i+\frac{\gamma_i}{2\mu_i}\operatorname{Var}\left(a_i^{*,NPF}S\right),        
%\end{split}
        \end{align}}%
where $\rho(P^{*,PF},a_R^{*,PF})$ is the common loading in the price-fairness condition \eqref{Eq:PricingFairness}. %\takwacomment{Table \ref{tab:baseline_ce_decomposition} reports the decomposition of CE loadings based on \eqref{Eq:CEdecomposition}. 
{As reported in Tables \ref{tab:delta_r_-_price_fairness}--\ref{tab:NPFcontract}, the discrepancies in the risk-sharing and reinsurance strategies between the two cases are relatively modest, with the primary distinction arising from the size of the side payments.} This leads to a significant difference between $\rho(P^{*,PF},a_R^{*,PF})$ and $\rho_i$ for $i=1,2,3$; see Table \ref{tab:baseline_ce_decomposition} for the decomposition of the CE loading {under the baseline setting}, which also explains the resulting ordering of the CE loadings across peers.

%Obviously, such a reversal of CE loadings order is caused by $\rho_i$ as there is no significance differences in the figures of $a_i^*$. 

\begin{table}[htbp] \centering 
\begin{tabular}{ccccc} \toprule Peer $i$ & {$\rho\left(P^{*,PF},a_R^{*,PF}\right)$} & $\rho_i$ & $\frac{\gamma_i}{2\mu_i}\operatorname{Var}\left(a_i^{*,PF}S\right)$ & $\frac{\gamma_i}{2\mu_i}\operatorname{Var}\left(a_i^{*,NPF}S\right)$ \\ \midrule 1 & 0.228758 & 0.195458	 & 0.031916 & 0.041102 \\ 2 & 0.228758 & 0.216492	 & 0.037356	 & 0.040626	 \\ 3 & 0.228758 & 0.262727	 & 0.028502	 & 0.019604	 \\ \bottomrule \end{tabular}
\caption{Decomposition of the certainty-equivalent loadings under the baseline setting. 
% \KTHNcomment{The table format is different from the above. I kinda prefer this than the previous ones} } \takwacomment{Take a look (I changed one), I don't know why this format looks nice here but not that nice in the above. I think I will keep the baseline setting table unchanged.
}\label{tab:baseline_ce_decomposition} 
\end{table}

\subsection{Impact of pool size and critical pool sizes}\label{sec:pool_size}

In this subsection, we conduct a numerical analysis to study the impact of the pool size $n$ under the homogeneity setting. To this end, we adopt the exponential-utility setting in Proposition \ref{Prop:ExpHomo} and model each peer's loss using Gamma distributions with a common shock, i.e., $    X_i=Y_0+Y_i\sim \mathrm{Gamma}(\kappa,s), \quad i=1,\dots,n,$
% \begin{equation*}
%     X_i=Y_0+Y_i\sim \mathrm{Gamma}(\kappa,s), \quad i=1,\dots,n,
% \end{equation*}
where $Y_0\sim \mathrm{Gamma}(\varrho\kappa,s)$ {represents the common shock}, and $Y_i\sim \mathrm{Gamma}((1-\varrho)\kappa,s)$ are assumed to be mutually independent. The parameter $\varrho\in[0,1]$ indicates the correlation between any distinct pair of loss variables, i.e., $\operatorname{Corr}(X_i,X_j)=\varrho$ if $i\neq j$ and $i,j\in\{1,\dots,n\}$. {By Proposition \ref{Prop:ExpHomo},} the bargaining power influences only the premium decision, and the optimal contract is independent of the initial wealth.  

%We specify the bargaining power of agents as follows. 
For $n\in\mathbb{N}$, where $\mathbb{N}:=\{1,2,\dots\}$, the bargaining power of each peer is specified by $\delta(n)=\alpha e^{-\beta/n}/n$, where $\alpha\in[0,1]$ indicates the limiting collective power of all peers {as $n\to\infty$} (and thus $1-\alpha$ signifies the P2P reinsurer's asymptotic bargaining power), and $\beta$ controls the speed with which the bargaining power changes due to changes in the pool size. {Fig.~\ref{fig:deltaRdelta_vs_n} plots the bargaining power of the P2P reinsurer and the representative peer against $n$ for $\alpha=0.5$ and $\alpha=0.999$. In both cases, the P2P reinsurer's bargaining power decreases at a decreasing rate and converges to $1-\alpha$, while the peers' collective bargaining power converges to $\alpha$.
} % asymptotic is skipped since it is clear from definition. 
The baseline parameters are reported in Table \ref{Tab:baseline}.
\begin{table}[H]
\centering
\begin{tabular}{ccccccccccc}
\toprule 
$\kappa$ & $s$ & $\mu$ & $\gamma$ & $\gamma_R$ & $\tau$ & $\theta$ & $\varrho$ & $\alpha$ & $\beta$ & $n$ \\ \hline
2        & 0.5 & 1     & 1        & 0.03       & 0.05   & 0.2      & 0.25      & 0.5      & 4 & 5,6,\dots,55      \\ \bottomrule
\end{tabular}
\caption{Baseline setting in the numerical analysis of the homogeneous case.}
\label{Tab:baseline}
\end{table}

\begin{figure}[!ht]	
	\begin{subfigure}[t]{.49\linewidth}
		\centering\includegraphics[scale=0.415]{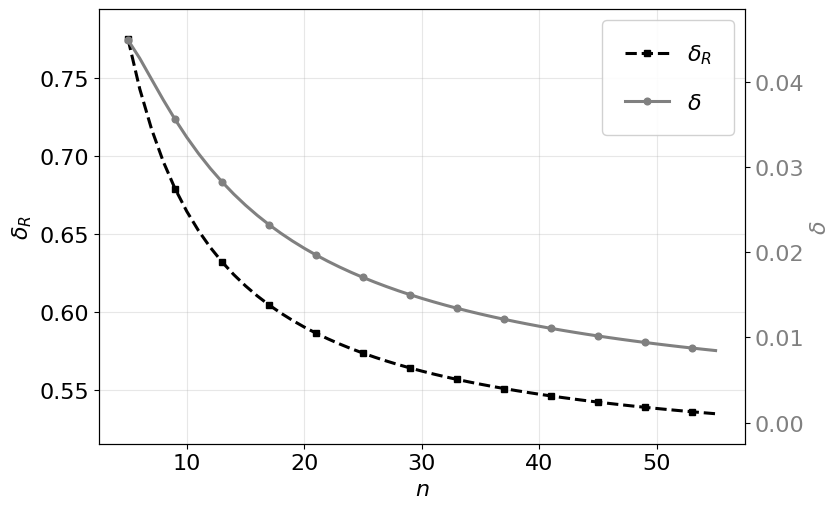}
		\caption{$\alpha=0.5$.}
		\label{fig:deltaRdelta_vs_n_a05}
	\end{subfigure}%
        \begin{subfigure}[t]{.49\linewidth}
		\centering\includegraphics[scale=0.415]{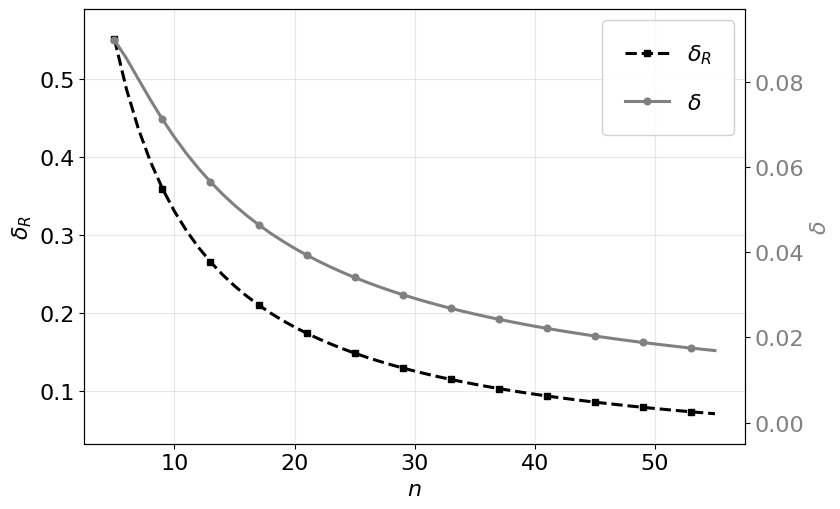}
		\caption{$\alpha=0.999$.}
		\label{fig:eltaRdelta_vs_n_a0999}
	\end{subfigure}%
    \caption{Impact of the pool size $n$ on the agents' bargaining power.}
	\label{fig:deltaRdelta_vs_n}
\end{figure}

%Fig.~\ref{fig:deltaRdelta_vs_n} plots the market power of the P2P reinsurer and the representative peer against $n$ for $\alpha=0.5$ and $\alpha=0.999$. In both cases, market power decreases at a decreasing rate, implying that both agents' market power stabilizes for the large pool size. 
%Juxtaposing Fig.~\ref{fig:deltaRdelta_vs_n_a05} and Fig.~\ref{fig:eltaRdelta_vs_n_a0999}, the P2P reinsurer retains relatively more power for $\alpha=0.5$ as the P2P reinsurer's asymptotic market power is higher in this case.

% \begin{itemize}
%     \item In both cases, market power decreases at a decreasing rate, but for $\alpha=0.5$, the P2P reinsurer retains more power when the pool size expands.
% \end{itemize}

\begin{figure}[!ht]	
	\begin{subfigure}[t]{.49\linewidth}
		\centering\includegraphics[scale=0.415]{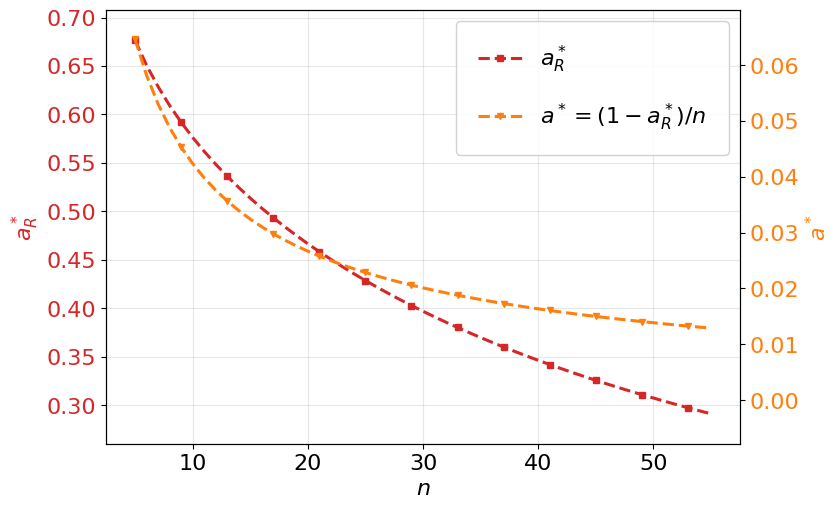}
		\caption{Impact on risk sharing $a_R^*$ and $a^*$.}
		\label{fig:aRa_vs_n}
	\end{subfigure}%
        \begin{subfigure}[t]{.49\linewidth}
		\centering\includegraphics[scale=0.405]{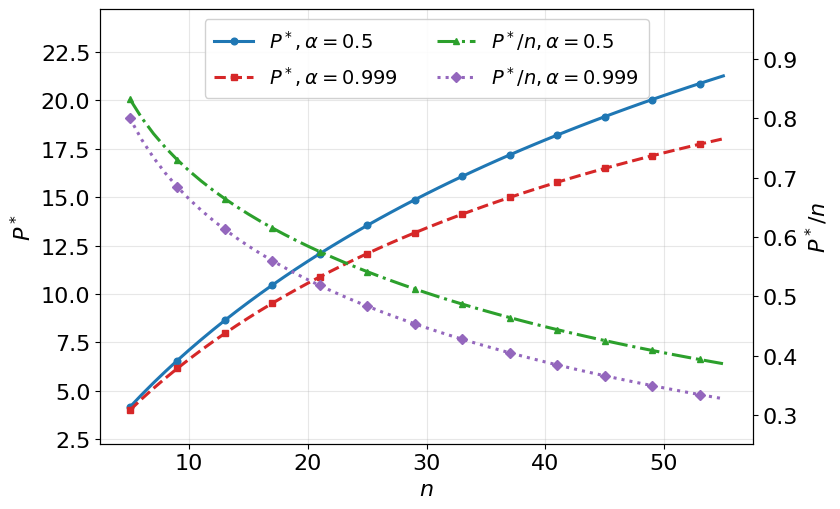}
		\caption{Impact on the premium $P$.}
		\label{fig:P_vs_n.png}
	\end{subfigure}%
    \caption{Impact of the pool size $n$ on the P2P insurance contract.} 
	\label{fig:PaRa_vs_n}
\end{figure}

Fig.~\ref{fig:PaRa_vs_n} shows the contractual response to changes in pool size. The left panel, Fig.~\ref{fig:aRa_vs_n}, plots $a_R^*$ and $a^*$ against $n$; both decrease as $n$ increases. Recall that the risk-sharing strategy is independent of the bargaining power (see Proposition \ref{Prop:ExpHomo} and particularly \eqref{Eq:ExpaREquation}),\footnote{So, we do not label the figure of $\alpha$ in that subplot.} and hence the variations are solely driven by the expansion of the pool size and the enhanced diversification effect. This is economically important: larger pools rely less on external reinsurance, but each peer also bears a smaller proportion of the aggregate retained loss, implying that the pool substitutes internal mutualization for external reinsurance. 

The right panel (Fig.~\ref{fig:P_vs_n.png}) plots the aggregate premium $P^*$ and peers' individual side payments $P^*/n$ against $n$ for $\alpha=0.5$ and $\alpha=0.999$. We observe that $P^*$ increases with $n$, whereas $P^*/n$ exhibits the opposite trend. The reduction in individual side payments as $n$ increases reflects the substitution of internal risk sharing for external reinsurance.
However, the increase in the number of peers outweighs the reduction in individual side payments, leading to an overall increase in the aggregate premium.
% Interestingly, although each peer pays less and $a_R^*$ declines when the pool size expands, more peers still increase the P2P reinsurer's premium income. \KTHNcomment{Not very surprising indeed, since $P^*$ is likely of order $n$ as more peers are paying (despite smaller individual payment)}
Another notable finding is that $P^*$ is lower when $\alpha=0.999$.
As the P2P reinsurer's bargaining power decreases, the bargaining outcome becomes less favorable to her. As a result, despite the same reinsurance contract (see Fig.~\ref{fig:aRa_vs_n}), the premium level is lower than in the $\alpha=0.5$ case.

% \begin{itemize}
%     \item $\uparrow n\Rightarrow \downarrow a_R^*$ and $\downarrow a^*$: Note first that the change of risk sharing is independent of the market power (see Proposition \ref{Prop:ExpHomo} and particularly \eqref{Eq:ExpaREquation}),\footnote{Thus, we do not label the figure of $\alpha$ in that subplot.} so it is mainly due to the expansion of pool size and the enlarging effect of diversification. This is economically important: larger pools rely less on external reinsurance, but each peer also bears a smaller proportion of the aggregate retained loss. The pool substitutes internal mutualization for external reinsurance. 
%     \item $\uparrow n\Rightarrow \uparrow P^*$: Although $a_R^*$ decreases with the pool size $n$, more peers still cause the P2P reinsurer's premium income increases. As the P2P reinsurer possesses less power when $\alpha=0.999$, her premium is affected and thus is less compared to the $\alpha=0.5$ case. 
% \end{itemize}

\begin{figure}[!ht]	
	\begin{subfigure}[t]{.49\linewidth}
		\centering\includegraphics[scale=0.415]{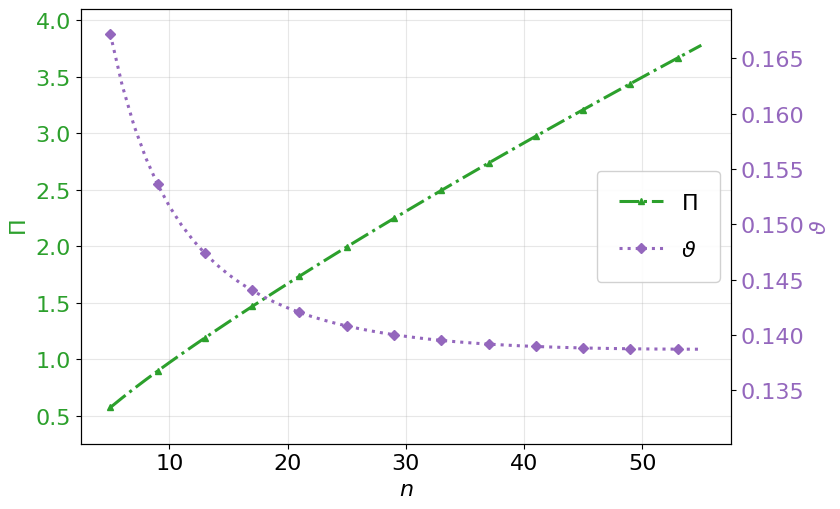}
		\caption{$\alpha=0.5$.}
		\label{fig:Pivartheta_vs_n_a05}
	\end{subfigure}%
        \begin{subfigure}[t]{.49\linewidth}
		\centering\includegraphics[scale=0.415]{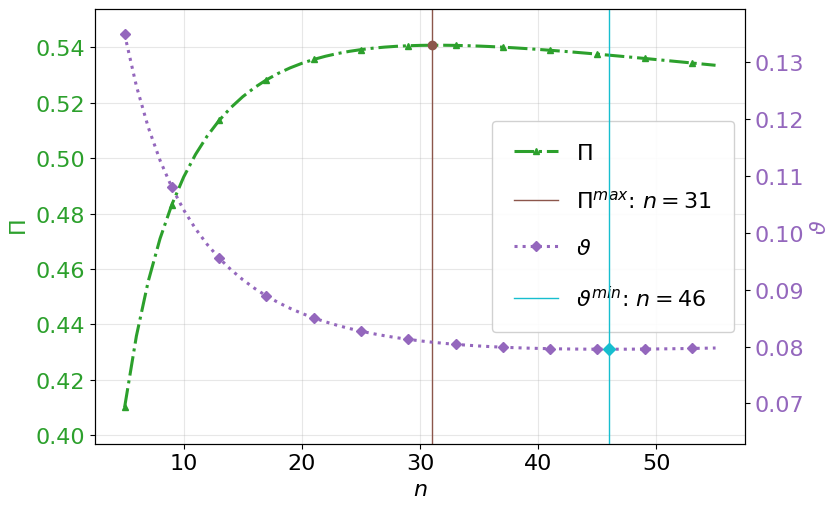}
		\caption{$\alpha=0.999$.}
		\label{fig:Pivartheta_vs_n_a0999}
	\end{subfigure}%
    \caption{Impact of the pool size $n$ on agents' welfare.}
	\label{fig:Pivartheta_vs_n}
\end{figure}

Lastly, we investigate the impact of the pool size on the welfare of all agents. Fig.~\ref{fig:Pivartheta_vs_n} plots all welfare metrics against $n$. For $\alpha=0.5$, Fig.~\ref{fig:Pivartheta_vs_n_a05} reveals that the expansion of the pool leads to an increase in $\Pi$ and a decrease in $\vartheta$, indicating welfare improvements for all agents. 
However, this is not always the case when $\alpha=0.999$; see Fig.~\ref{fig:Pivartheta_vs_n_a0999}, where $\Pi$ displays a hump shape and reaches its maximum (denoted by $\Pi^{max}$) at $n=31$, and $\vartheta$ attains its minimum (denoted by $\vartheta^{min}$) at $n=46$. 

At first sight, one might expect welfare to improve monotonically with the pool size: decreasing $a_R^*$ and increasing $P^*$ should improve the P2P reinsurer's welfare, while decreasing $a^*$ and $P^*/n$ should lower each peer's risk bearing and payment burden. However, the expanding pool also increases the total loss $S$, which explains the non-monotonicity observed when $\alpha=0.999$. In this scenario, the premium level per peer is relatively low due to the P2P reinsurer's weaker bargaining power. Once the pool becomes sufficiently large, the utility gain from further reductions in side payments is limited. At that point, the increase in aggregate losses $S$ outweighs the remaining gains from lower premiums per capita, causing the peers' welfare to decline. A similar explanation applies to the P2P reinsurer's case. Although a larger pool generates higher aggregate premium income, the increase in premium revenue is eventually insufficient to compensate for the marginal covered aggregate losses. Consequently, the P2P reinsurer's welfare also declines beyond a certain pool size.

\section{Concluding remarks and future directions}\label{Sec:Conclusion}

This paper develops a Nash bargaining framework for P2P insurance contracting between a risk-averse P2P reinsurer and a group of risk-averse peers, in which peers’ disagreement points are induced by outside centralized insurance options. We further impose a price-fairness condition that requires all peers’ expected contributions to be determined by a common loading applied to their individual expected losses.

We establish the applicability of this framework by adapting the asymmetric Nash-bargaining characterization in \cite{kalai1977nonsymmetric} to the present multi-agent P2P insurance setting. In addition, we characterize the Nash bargaining solution and derive sufficient ex post stability conditions that rule out viable subgroup deviations, both with and without price fairness.

Our numerical analysis shows that price fairness reduces dispersion in risk-sharing proportions and certainty-equivalent loadings among peers. At the same time, when peers differ in risk aversion, the common-loading rule may redistribute welfare in certainty-equivalent terms. Further analysis under homogeneity reveals the possible non-monotonicity of all agents' welfare with respect to pool size. This non-monotonicity is closely related to how bargaining power evolves with pool size, highlighting the role of bargaining power in determining whether pool expansion is welfare-improving.

% \takwacomment{I realize it's too long so I will put them to the introduction (contribution part) and have a more succinct version here.}

Nash bargaining in P2P insurance deserves further investigation. First, it would be interesting to consider a risk-neutral reinsurer with beliefs that differ from those of the peers, while allowing for stop-loss indemnity; see, e.g., \cite{boonen2025optimal} for such a risk-sharing architecture. 
% ; see, e.g.,  \cite{chi2019optimality,chi2024optimal}. 
In addition, the numerical findings on pool size call for a more general theoretical study of pool expansion under Nash bargaining; see, e.g., \cite{anthropelos2026expansion,blier2026designing}. In particular, future work could identify conditions under which enlarging the pool benefits all agents. We leave these for future work.

% \textbf{Future research}
% \begin{itemize}
%     \item Risk-neutral reinsurer with a heterogeneous belief and stop-loss reinsurance; see \cite{chi2019optimality,chi2024optimal}.
%     \item Kalai–Smorodinsky bargaining solution \cite{kalai1975other}
%     \item Pool expansion under ANB \cite{anthropelos2026expansion}
% \end{itemize}
\begin{spacing}{1.0}
\section*{Acknowledgments}	

Kenneth Ng acknowledges support from the start-up fund at The Ohio State University and the CKER research fund of the Society of Actuaries (Project title: Pricing and Staking of Decentralized Insurance).
 
\bibliographystyle{apalike}
%\setstretch{0.85}
\bibliography{sample}
%\onehalfspacing
%\singlespacing
%\setstretch{0.9}
\end{spacing}
\appendix
\section*{Appendix}
\renewcommand{\thesubsection}{\Alph{subsection}}
%\setstretch{1.2}
\section{Proof of Proposition \ref{pp:V:continuous}}\label{App:pp:V:continuous}
%\Timcomment{Do we show differentiability?}

    For the first two statements, we only show the finiteness, continuity, concavity, and continuous differentiability of $V_R$. The remaining statements can be shown in a similar fashion.
    %We first consider $V_R$.
    For any $(\bm{a},P)\in\Delta^n\times{\mathbb{R}}$, by the monotonicity of $U_R$, we have {\small 
        \begin{equation*}
      U_R\left(w_R + { P} -(1+\tau)S \right)    \leq   U_R\left(w_R + P - (1+\tau)\left(1-\sum_{i=1}^n a_i\right)S \right) \leq U_R\left(w_R +P \right) < \infty\quad\mathbb{P}\text{-a.s.}.
        \end{equation*}}%
    Since $\mathbb{E}[|U_R(w_R + {P}-(1+\tau)S)|]<\infty$ by Assumption \ref{Ass:EU}, we have $|V_R(\bm{a},P)|<\infty$. 
    
    To prove continuity, fix any $(\bm{a}_0,P_0)\in \Delta^n \times {\mathbb{R}}$ and choose $\underline{P},\overline{P}\in \mathbb{R}$ such that $\underline{P}<P_0<\overline{P}$. For any sequence $(\bm{a}_k,P_k)_{k=1}^\infty \subset \Delta^n\times {\mathbb{R}}$ with $\bm{a}_k=(a_{k,i})_{i=1}^n$ and $(\bm{a}_k,P_k)\to (\bm{a}_0,P_0)$ as $k\to\infty$, we have $\underline{P}<P_k<\overline{P}$ for all sufficiently large $k$, and hence
     {\small   \begin{equation*}
            \left|U_R\left(w_R + P_k - (1+\tau)\left(1-\sum_{i=1}^n a_{k,i}\right)S \right) \right| \leq \max\left\{\left|U_R(w_R +  \overline{P}) \right|, \left|U_R(w_R + { \underline{P} } -(1+\tau)S) \right|  \right\}\quad\mathbb{P}\text{-a.s.}.
        \end{equation*}}%
    By the continuity of $U_R$ and Assumption \ref{Ass:EU}, we have, by Dominated Convergence, 
     {\small    \begin{align*}
            V_R(\bm{a}_k,P_k) &= \mathbb{E}\left[U_R\left(w_R + P_k - (1+\tau)\left(1-\sum_{i=1}^n a_{k,i}\right)S \right) \right] 
          %  \\            &\to \mathbb{E}\left[U_R\left(w_R + P_0 - (1+\tau)\left(1-\sum_{i=1}^n a_{0,i}\right)S \right) \right] =
            \to V_R(\bm{a}_0,P_0),
        \end{align*}}%
    as $k\to \infty$. Since $\underline{P},\overline{P}\in\mathbb{R}$ can be taken arbitrarily, we conclude that $V_R$ is continuous.

    Next, for every $\omega\in\Omega$, $W_R(\bm{a},P)$ is affine in $(\bm{a},P)$. Since $U_R$ is concave, the mapping $(\bm{a},P)\mapsto U_R(W_R(\bm{a},P))$ is concave pointwise in $\omega$. Taking expectations preserves concavity. Hence $V_R$ is concave.

    Finally, we show that the partial derivative of $V_R$ with respect to $P$ exists. To this end, fix $(\bm{a}_0,P_0)\in \Delta^n\times \mathbb{R}$ and $e>0$. By the mean value theorem, for any $h\in\mathbb{R}$ such that $|h|<e$, there exists a random variable $P^h$, taking values in $[P_0-e,P_0+e]$, such that, almost surely, 
    % \takwacomment{$P^h$  is random? I changed $\delta$ to e as $\delta$ is used.}
    % \Timcomment{I guess proof is fine, but $P^h$ is indeed not ``random''.} \KTHNcomment{it is random because the MVT is applied depending on the value of S} \takwacomment{I added the lines for the almost sure sense.}
      {\small   \begin{align*}
          Q_R(\bm{a}_0,P_0,h) &:=  \frac{U_R\left(w_R + P_0 + h - (1+\tau)\left(1-\sum_{i=1}^n a_{0,i}\right)S \right) -  U_R\left(w_R + P_0  - (1+\tau)\left(1-\sum_{i=1}^n a_{0,i}\right)S \right)}{h}\\
           &= U_R'\left(w_R + P^h - (1+\tau)\left(1-\sum_{i=1}^n a_{0,i}\right)S \right) \in \left[0, U_R'(w_R + P_0-e - (1+\tau)S) \right],
        \end{align*}}%
    which is uniformly integrable with respect to $h$ by Assumption \ref{Ass:EU}. Therefore, by Dominated Convergence, 
       {\small  \begin{align*}
            \frac{\partial V_R(\bm{a}_0,P_0)}{\partial P} = \lim_{h\to 0}    \mathbb{E}\left[    Q_R(\bm{a}_0,P_0,h)\right]
            %&= \mathbb{E}\left[  \lim_{h\to 0}  Q_R(\bm{a}_0,P_0,h)\right] 
            = \mathbb{E}\left[U_R' \left(w_R + P_0 - (1+\tau)\left(1-\sum_{i=1}^n a_{0,i}\right)S \right) \right] < \infty,
        \end{align*}}%
    which establishes the claim. Likewise, we can establish the existence of $\frac{\partial V_R}{\partial a_i}$, and the continuity of all partial derivatives, by appropriate integrability conditions in Assumption \ref{Ass:EU}. 
    Therefore, $V_R$ is continuously differentiable.

To prove the third statement, note  that for fixed $\bm a\in\Delta^n$, the mapping
$P\mapsto V_i(\bm a,P)$ is continuous and strictly decreasing. Fix $P_0\in\mathbb{R}$, for any $P\leq P_0$, define $\hat{Z}(P):=U_i(W_i(\bm{a},P))-U_i(W_i(\bm{a},P_0))$, which is strictly decreasing in $P$ and nonnegative $\mathbb{P}$-a.s.. By the monotone convergence theorem, we have $\lim_{P\rightarrow-\infty}\mathbb{E}[\hat{Z}(P)]=\lim_{x\rightarrow+\infty}U_i(x)-V_i(\bm{a},P_0)$ and thus $\lim_{P\rightarrow-\infty}V_i(\bm{a},P)=\lim_{x\rightarrow+\infty}U_i(x)>d_i$. Following a similar argument, we also have $\lim_{P\rightarrow+\infty}V_i(\bm{a},P)=\lim_{x\rightarrow-\infty}U_i(x)<d_i$. \qed
    % For $V_i$, $i\in\mathcal{N}$, note that for any $(\bm{a},P)\in\Delta^n\times {\mathbb{R}}$, by the non-decreasing property of $U_i$,  
    %     \begin{equation*}
    %     U_i\left(w_i-S- \left(\frac{P}{\mathbb{E}[S]}+1\right)\mu_i \right)   \leq   U_i\left( w_i - a_i(S-\mathbb{E}[S]) - \left(\frac{P}{\mathbb{E}[S]} + \sum_{j=1}^n a_j \right)\mu_i \right) \leq U_i\left(w_i + \mathbb{E}[S] \right).
    %     \end{equation*}
    % By Assumption \ref{Ass:EU}, we infer that $|V_i(\bm{a},P)|<\infty$. The continuity of $V_i$ can be proven using a similar argument as above. Henceforth, we omit the proof. 

\section{Proof of Lemma \ref{Lem:ConcaveH}}\label{App:ConcaveH}
Under Assumption \ref{Ass:EU},
% and \ref{P_max_assumption}, 
Proposition \ref{pp:V:continuous} implies that the functions $V_j$, $j\in\mathcal{I}$, are finite, continuous, and concave, while the definitions of $P_{\min}$ and $P_{\max}$ are well-posed. To prove that $H(\bm{a})$ is concave, it suffices to show that $P_{\max}(\bm{a})$ is concave and $P_{\min}(\bm{a})$ is convex. First, we show that $P_{\max}^i(\bm{a})$ is concave for each $i\in\mathcal{N}$. 
% Define
% \begin{equation*}
%     f_i(\bm{a},P):=\mathbb{E}\left[U_i\left( w_i - a_i(S-\mathbb{E}[S]) - \left(\frac{P}{\mathbb{E}[S]} + \sum_{j=1}^n a_j \right)\mu_i \right)\right].
% \end{equation*}
By Proposition \ref{pp:V:continuous}, since $V_i$ is concave in $(\bm{a},P)$, the set $    B_i:=\{(\bm{a},P)\in\Delta^n\times\mathbb{R}:V_i(\bm{a},P)\geq d_i\}
$
% \[
%     B_i:=\{(\bm{a},P)\in\Delta^n\times\mathbb{R}:V_i(\bm{a},P)\geq d_i\}
% \]
is convex. As $V_i(\bm{a},\cdot)$ is strictly decreasing in $P$, we can write $
    B_i=\{(\bm{a},P)\in\Delta^n\times\mathbb{R}:P\leq P_{\max}^i(\bm{a})\},$
% \[
%     B_i=\{(\bm{a},P)\in\Delta^n\times\mathbb{R}:P\leq P_{\max}^i(\bm{a})\},
% \]
which is the hypograph of $P_{\max}^i$. Hence, $P_{\max}^i(\bm{a})$ is concave in $\bm{a}$. Since $P_{\max}(\bm{a}) = \min_{i\in\mathcal{N}}P^i_{\max}(\bm{a})$, we conclude that $P_{\max}(\bm{a})$ is also concave.

Next, we prove that $P_{\min}(\bm{a})$ is convex. By similar arguments, 
% we define
% \begin{equation*}
%     g(\bm{a},P):=\mathbb{E}\left[U_i\left( w_i - a_i(S-\mathbb{E}[S]) - \left(\frac{P}{\mathbb{E}[S]} + \sum_{j=1}^n a_j \right)\mu_i \right)\right].
% \end{equation*}
by Proposition \ref{pp:V:continuous}, since $V_R$ is concave in $(\bm{a},P)$, the set $  B_R:=\{(\bm{a},P)\in\Delta^n\times\mathbb{R}:V_R(\bm{a},P)\geq d_R\}$
% \begin{equation*}
%     B_R:=\{(\bm{a},P)\in\Delta\times[0,+\infty):V_R(\bm{a},P)\geq d_R\}
% \end{equation*}
is convex. As $V_R(\bm{a},\cdot)$ is strictly increasing, we have $B_R=\{(\bm{a},P)\in\Delta^n\times\mathbb{R}:P \geq  P_{\min}(\bm{a})\}$, which is the epigraph of $P_{\min}$. Hence, $P_{\min}(\cdot)$ is convex. \qed

\section{Proof of Lemma \ref{Lem:ccF}}\label{App:ccF}
Under Assumption \ref{Ass:EU}, 
% and \ref{P_max_assumption}, 
Lemma \ref{Lem:ConcaveH} applies and Proposition \ref{pp:V:continuous} ensures the continuity of all expected-utility functions.

\noindent\textbf{Convexity}: For any $(\bm{a}^{(1)},P^{(1)}),(\bm{a}^{(2)},P^{(2)}) \in \mathcal{F}$ and $\lambda\in[0,1]$, set $        \overline{\bm{a}}:=\lambda\bm{a}^{(1)}+(1-\lambda)\bm{a}^{(2)}$ and $\overline{P}:=\lambda P^{(1)}+(1-\lambda)P^{(2)}$.
    % \begin{equation*}
    %     \overline{\bm{a}}:=\lambda\bm{a}^{(1)}+(1-\lambda)\bm{a}^{(2)},\quad \overline{P}:=\lambda P^{(1)}+(1-\lambda)P^{(2)}.
    % \end{equation*}
    Since $\Delta^n$ is convex, $\overline{\bm{a}}\in \Delta^n$. By the convexity of $P_{\min}$ and the concavity of $P_{\max}$ established in the proof of Lemma \ref{Lem:ConcaveH}, we have
    {\small\begin{align*}
        &P_{\min}(\overline{\bm{a}})\leq\lambda P_{\min}(\bm{a}^{(1)})+(1-\lambda)P_{\min}(\bm{a}^{(2)})\leq\lambda P^{(1)}+(1-\lambda)P^{(2)}=\overline{P},\\
        &P_{\max}(\overline{\bm{a}})\geq\lambda P_{\max}(\bm{a}^{(1)})+(1-\lambda)P_{\max}(\bm{a}^{(2)})\geq \lambda P^{(1)}+(1-\lambda)P^{(2)}=\overline{P},
    \end{align*}}%
    and thus $P_{\min}(\overline{\bm a}) \le \overline P \le P_{\max}(\overline{\bm a})$. Hence, $\mathcal{F}$ is convex.

\vspace{2mm}
    
\noindent\textbf{Compactness}: Note that $\Delta^n\subset\mathbb{R}^n$ is closed and bounded, and thus its compactness follows. To proceed, we express $\mathcal{F}$ as
{\small\begin{equation*}
    \mathcal{F}=(\Delta^n\times[0,+\infty))\cap V_R^{-1}([d_R,+\infty))\cap \bigcap_{i=1}^nV_i^{-1}([d_i,+\infty)).
\end{equation*}}%
Since $V_R$ and $V_i$ are continuous (by {Proposition \ref{pp:V:continuous}}), the preimages $V_R^{-1}([d_R,+\infty))$ and 
$V_i^{-1}([d_i,+\infty))$, $i\in\mathcal N$, are closed. Therefore, $\mathcal{F}$ is closed.

To complete the proof, it suffices to show that $\mathcal{F}$ is bounded. Note that $\bm{a}$ is automatically bounded as it belongs to $\Delta^n$, and it remains to bound $P$. Using the monotonicity of $U_i$, we have
{\small\begin{equation*}
U_i(w_i-(1+\theta)\mu_i)\leq\mathbb{E}[U_i(W_i(\bm{a},P))]\leq U_i\left(w_i+\mathbb{E}[S]-\frac{\mu_i}{\mathbb{E}[S]}P\right),
\end{equation*}}%
which leads to $    P\leq \hat{P}:= \min_{i\in\mathcal{N}}\frac{\mathbb{E}[S]}{\mu_i}(\mathbb{E}[S]+(1+\theta)\mu_i)
$.
% \begin{equation*}
%     P\leq \hat{P}:= \frac{\mathbb{E}[S]}{\mu_i}(\mathbb{E}[S]+(1+\theta)\mu_i).
% \end{equation*}
Hence, $\mathcal{F}\subset\Delta^n\times[0,\hat{P}]$,
%, where
% \begin{equation*}
% \hat{P}:=\min_{1\leq i\leq n}\frac{\mathbb{E}[S]}{\mu_i}(\mathbb{E}[S]+(1+\theta)\mu_i),
% \end{equation*}
and the boundedness of $\mathcal{F}$ follows. Since $\Delta^n\times[0,\hat{P}]$ is compact and $\mathcal{F}$ is closed, $\mathcal{F}$ is a compact set. \qed

\section{Proof of Lemma \ref{Lemma:SF}}\label{App:SF}
Under Assumptions \ref{Ass:EU} and
Proposition \ref{pp:V:continuous},
% \ref{P_max_assumption}, 
$P_{min}$, $P_{max}$ and $H$ are well-defined by the preceding results. If $\mathcal{K}\neq\emptyset$, then there exists $(\bm{a},P)\in\mathcal{K}$ such that $P_{\min}(\bm{a})<P<P_{\max}(\bm{a})$, and thus $H(\bm{a})>0$. Therefore, $\sup_{\bm{a}\in\Delta^n}H(\bm{a})>0$. Conversely, suppose that $\sup_{\bm{a}\in\Delta^n}H(\bm{a})>0$. Then there exists $\hat{\bm{a}}\in\Delta^n$ such that $H(\hat{\bm{a}})>0$; otherwise, $H(\bm{a})\leq 0$ for all $\bm{a}\in\Delta^n$, contradicting the assumption. Thus, $P_{\min}(\hat{\bm{a}})<P_{\max}(\hat{\bm{a}})$. Choose any $\hat{P}\in(P_{\min}(\hat{\bm{a}}),P_{\max}(\hat{\bm{a}}))$. Since $P_{\min}(\hat{\bm{a}})\geq 0$, we obtain $\hat{P}\geq 0$ and $(\hat{\bm{a}},\hat{P})\in\mathcal{K}$, whence $\mathcal{K}\neq\emptyset$. \qed

%  ``Only if": Suppose that there exists $(\bm{a},P)$ such that $\mathbb{E}\left[U_j\left(W_j(\bm{a},P)\right)\right]
% > d_j$ for all $j\in\mathcal{I}$.
% %     \begin{equation*}
% %          \mathbb{E}[U_R(W_R(\bm{a},P))]> d_R, \quad
% % \mathbb{E}\left[U_i\left(W_i(\bm{a},P)\right)\right]
% % > d_i,\quad\text{for all }i\in\mathcal{N}.
% %     \end{equation*}
%     The first inequality implies that $P>P_{\min}(\bm{a})$, and the remaining leads to $P<P_{\max}^i(\bm{a})$ for all $i\in\mathcal{N}$, and thus $P<P_{\max}(\bm{a})$. Therefore, $P_{\min}(\bm{a})<P_{\max}(\bm{a})$ and hence \eqref{Eq:supH} holds.

% \vspace{2mm}

% \noindent``If": Suppose that there exists $\bm{a}^*\in\Delta^n$ such that $P_{\min}(\bm{a}^*)<P_{\max}(\bm{a}^*)$. By picking $P^*\in(P_{\min}(\bm{a}^*),$ $P_{\max}(\bm{a}^*))$, we have $\mathbb{E}[U_R(W_R(\bm{a}^*,P^*))]> d_R$ and $\mathbb{E}\left[U_i\left(W_i(\bm{a}^*,P^*)\right)\right]
% > d_i$ for all $i\in\mathcal{N}$.
% % \begin{equation*}
% %          \mathbb{E}[U_R(W_R(\bm{a}^*,P^*))]> d_R\quad\text{and}\quad\mathbb{E}\left[U_i\left(W_i(\bm{a}^*,P^*)\right)\right]
% % > d_i\quad\text{for all }i\in\mathcal{N}.
% %     \end{equation*}
%     Therefore, $(\bm{a}^*,P^*)$ is strictly feasible.

\section{Proof of Lemma \ref{Lem:ccU}}\label{App:ccU}

By Assumption \ref{Ass:EU}, 
% and \ref{P_max_assumption}, 
Proposition \ref{pp:V:continuous} and Lemma \ref{Lem:ccF} apply throughout the following analysis.

\noindent\textbf{Convexity}: For any $u^{(1)},u^{(2)}\in\mathcal{U}$, there exists $(\bm{a}^{(1)},P^{(1)}),(\bm{a}^{(2)},P^{(2)}) \in \mathcal{F}$ such that $ d_j\leq u_j^{(k)}\leq V_j(\bm{a}^{(k)},P^{(k)})$ for all $j\in\mathcal{I}$ and $k=1,2$.
% \begin{equation*}
%     d_j\leq u_j^{(k)}\leq V_j(\bm{a}^{(k)},P^{(k)}),\quad \text{for all }j\in\mathcal{I}\text{ and } k=1,2.
% \end{equation*}
Adopting $(\overline{\bm{a}},\overline{P})\in\mathcal{F}$ defined in Appendix \ref{App:ccF}, as $V_j$ is concave, we have $V_j(\overline{\bm{a}},\overline{P})\geq \lambda V_j(\bm{a}^{(1)},P^{(1)})+(1-\lambda)V_j(\bm{a}^{(2)},P^{(2)})\geq\lambda u_j^{(1)}+(1-\lambda) u_j^{(2)}\geq d_j.$
% \begin{equation*}
% V_j(\overline{\bm{a}},\overline{P})\geq \lambda V_j(\bm{a}^{(1)},P^{(1)})+(1-\lambda)V_j(\bm{a}^{(2)},P^{(2)})\geq\lambda u_j^{(1)}+(1-\lambda) u_j^{(2)}\geq d_j.
% \end{equation*}
Thus, $\lambda u^{(1)}+(1-\lambda) u^{(2)}\in\mathcal{U}$.
%and $\mathcal{U}$ is convex.

\vspace{2mm}

\noindent\textbf{Compactness}: Due to the compactness of $\mathcal{F}$ and the continuity of each $V_j$, each $V_j$ attains a finite maximum $    \hat{V}_j:=\max_{(\bm{a},P)\in\mathcal{F}}V_j(\bm{a},P)<\infty$.
% \begin{equation*}
%     \hat{V}_j:=\max_{(\bm{a},P)\in\mathcal{F}}V_j(\bm{a},P)<\infty.
% \end{equation*}
Thus, $\mathcal{U}\subset \prod_{j\in\mathcal{I}}[d_j,\hat{V}_j]$ is bounded.

To show closedness, define a sequence $u^{(m)}\in\mathcal{U}$ such that $\lim_{m\rightarrow\infty}u^{(m)}=u$. For each $m$, select $(\bm{a}^{(m)},P^{(m)})\in\mathcal{F}$ such that $ d_j\leq u_j^{(m)}\leq V_j(\bm{a}^{(m)},P^{(m)})$ for all $j\in\mathcal{I}$.
% \begin{equation*}
%     d_j\leq u_j^{(m)}\leq V_j(\bm{a}^{(m)},P^{(m)}),\quad \text{for all }j\in\mathcal{I}.
% \end{equation*}
Since $\mathcal{F}$ is compact, there exists a convergent subsequence $(\bm{a}^{(m_i)},P^{(m_i)})\rightarrow(\bm{a},P)\in\mathcal{F}$. By the continuity of $V_j$, $\lim_{i\rightarrow\infty}V_j(\bm{a}^{(m_i)},P^{(m_i)})= V_j(\bm{a},P)$. Taking limits in $u_j^{(m_i)}\leq V_j(\bm{a}^{(m_i)},P^{(m_i)})$ leads to $u_j\leq V_j(\bm{a},P)$. Along with $u_j\geq d_j$, we obtain $u\in\mathcal{U}$ and thus $\mathcal{U}$ is closed. \qed

\section{Proof of Theorem \ref{Thm:properties}}\label{App:properties}

We prove that the bargaining outcome $f(\mathcal{V},\bm{d})$,   $(\mathcal{V},\bm{d})\in\mathcal{B}$, is well-defined and satisfies the four properties.

\noindent\textbf{Well-defined}: The Nash product is continuous on the compact set $\mathcal{V}$, and hence it attains a maximum. Since $\mathcal{V}$ contains a point strictly above $\bm{d}$, the maximal Nash product is positive. Therefore, every maximizer satisfies $\bm{u}>\bm{d}$. On the strictly positive region $\{\bm{u}:\bm{u}>\bm{d}\}$, maximizing the Nash product is equivalent to maximizing $\sum_{j\in\mathcal{I}}\delta_j\ln(u_j-d_j)$, which is strictly concave. Hence, the maximizer is unique, and \eqref{eq:general_bargaining_form} is well-defined. Given this, denote $\bm{u}^*$ as the optimal utility vector of the problem in \eqref{eq:general_bargaining_form}.
\vspace{2mm}

\noindent\textbf{Strict feasibility}: Since $(\mathcal{V},\bm{d})\in\mathcal{B}$, there exists $\bar{\bm{u}}\in\mathcal{V}$ such that $\bar{\bm{u}}>\bm{d}$. Hence, the maximum Nash product is strictly positive. If $\bm{u}^*=f(\mathcal{V},\bm{d})$ gives $u^*_j=d_j$ for some $j\in\mathcal{I}$, its Nash product would be zero, contradicting the positivity. Therefore, $\bm{u}^*>\bm{d}$.

%It follows directly from Lemma \ref{Lemma:SF}, since there exists a feasible contract with strictly positive Nash product; any solution to Problem \eqref{Prob:hetero} cannot have a binding IR constraint, so does the optimal utility vector of Problem \eqref{Prob:NB} due to their equivalence.

\vspace{2mm}

\noindent\textbf{Pareto optimality}: Assume that $\bm{u}^*$ is not Pareto-optimal, i.e., there exists $\bm{\Tilde{u}}\in\mathcal{V}$ such that $    \Tilde{u}_j\geq u_j^*$ for all $j\in\mathcal{I}$,
% \begin{equation*}
%     \Tilde{u}_j\geq u_j^*\quad\text{for all }j\in\mathcal{I}
% \end{equation*}
with at least one inequality being strict. Together with strict feasibility that $\bm{u}^*>\bm{d}$, this leads to $\prod_{j\in\mathcal{I}}(\Tilde{u}_j-d_j)^{\delta_j}>\prod_{j\in\mathcal{I}}(u_j^*-d_j)^{\delta_j}$,
% \begin{equation*}
% \prod_{j\in\mathcal{I}}(\Tilde{u}_j-d_j)^{\delta_j}>\prod_{j\in\mathcal{I}}(u_j^*-d_j)^{\delta_j},
% \end{equation*}
contradicting the optimality of $\bm{u}^*$.

\vspace{2mm}

%\noindent\textbf{Invariance to positive affine transformations}:
\noindent\textbf{IPAT}: Adopting the transformation \eqref{Eq:positiveaffine}, let $\hat{\mathcal{V}}=T(\mathcal{V})$ and $\hat{\bm{d}}=T(\bm{d})$. For any $\hat{\bm{u}}\in\hat{\mathcal{V}}$, there exists $\bm{u}\in\mathcal{V}$ such that $\hat{\bm{u}}=T(\bm{u})$. Moreover, $\hat{\bm{u}}-\hat{\bm{d}}=\bm{\alpha}(\bm{u}-\bm{d})$. Hence, for $\hat{\bm{u}}\in\hat{\mathcal{V}}$, $\prod_{j\in\mathcal{I}}(\hat{u}_j-\hat{d}_j)^{\delta_j}=\left(\prod_{j\in\mathcal{I}}\alpha_j^{\delta_j}\right)\prod_{j\in\mathcal{I}}(u_j-d_j)^{\delta_j}.$
% {\small\begin{equation*}
% \prod_{j\in\mathcal{I}}(\hat{u}_j-\hat{d}_j)^{\delta_j}=\left(\prod_{j\in\mathcal{I}}\alpha_j^{\delta_j}\right)\prod_{j\in\mathcal{I}}(u_j-d_j)^{\delta_j}.
% \end{equation*}}%
The factor $\prod_{j\in\mathcal{I}}\alpha_j^{\delta_j}$ is positive and independent of $\bm{u}$. Therefore, $\bm{u}^*$ maximizes the original Nash product over $\mathcal{V}$ if and only if $T(\bm{u}^*)$ maximizes the transformed Nash product over $\hat{\mathcal{V}}$. Thus, $f(T(\mathcal{V}),T(\bm{d}))=T(f(\mathcal{V},\bm{d}))$.
% % \begin{equation*}
% %     \hat{V}_j=\alpha_jV_j+\beta_j,\quad \hat{d}_j=\alpha_jd_j+\beta_j,\quad\text{for all }j\in\mathcal{I}, 
% % \end{equation*}
% % where $\alpha_j>0$, 
% the objective the problem in \eqref{eq:general_bargaining_form} then becomes $  \prod_{j\in\mathcal{I}}(\hat{u}_j-\hat{d}_j)^{\delta_j}=\prod_{j\in\mathcal{I}}\alpha_j^{\delta_j}\prod_{j\in\mathcal{I}}(u_j-d_j)^{\delta_j}$.
% % \begin{equation*}
% %     \prod_{j\in\mathcal{I}}(\hat{u}_j-\hat{d}_j)^{\delta_j}=\prod_{j\in\mathcal{I}}\alpha_j\prod_{j\in\mathcal{I}}(u_j-d_j)^{\delta_j}.
% % \end{equation*}
% As $\prod_{j\in\mathcal{I}}\alpha_j^{\delta_j}$ is a positive constant, it would not change the original maximizer.

\vspace{2mm}

%\noindent\textbf{Independence of irrelevant alternatives}:
\noindent\textbf{IIA}: Let $\mathcal{V}'\subset\mathcal{V}$ such that $\bm{u}^*\in\mathcal{V}'$. Then, along with the optimality of $\bm{u}^*$ in $\mathcal{V}$, we have in particular that $ \prod_{j\in\mathcal{I}}(u^*_j-d_j)^{\delta_j}\geq\prod_{j\in\mathcal{I}}(u_j-d_j)^{\delta_j}$ for all $\bm u\in\mathcal{V}'$. This implies that $f(\mathcal{V}',\bm{d})=f(\mathcal{V},\bm{d})$.
% \begin{equation*}
%     \prod_{j\in\mathcal{I}}(\bm{u}^*_j-d_j)^{\delta_j}\geq\prod_{j\in\mathcal{I}}(u_j-d_j)^{\delta_j}\quad \text{for all }u\in\mathcal{U}'.
% \end{equation*}

\vspace{2mm}
Next, we show that any bargaining outcome $f:\mathcal{B}\to\mathbb{R}^{n+1}$ satisfying the four properties can be represented as a Nash-bargaining solution. By IPAT, it suffices to consider the normalized problem with disagreement point $\bm{0}$. Indeed, the affine transformation $T(\bm{u})=\bm{u}-\bm{d}$ maps $(\mathcal{V},\bm{d})$ to $(\mathcal{X},\bm{0})$, where $    \mathcal{X}:=\mathcal{V}-\bm{d}:=\{\bm{u}-\bm{d}:\bm{u}\in\mathcal{V}\}
$,
% \begin{equation*}
%     \mathcal{X}:=\mathcal{U}-\bm{d}:=\{x=\bm{u}-\bm{d}:\bm{u}\in\mathcal{U}\},
% \end{equation*}
which is a convex and compact set. We thus aim to characterize $f(\mathcal{X})\in\mathcal{X}$, where we omit the disagreement-point argument $\bm{0}$ in $f$ to ease the presentation.

\vspace{2mm}
\noindent\textbf{Step 1: Construction of $\bm{\delta}$}\\
% Define the standard simplex $\Delta_{+}:=\{\bm{x}\in\mathbb{R}^{n+1}_+:\sum_{i=1}^{n+1}x_i\leq1\}$.
Let $\bm{\delta}:=f(\Delta^{n+1})$, where $\Delta^{n+1}:=\{\bm{x}\in\mathbb{R}^{n+1}_+:\sum_{i=1}^{n+1}x_i\leq 1\}$. %\Timcomment{$\Delta$ already defined, but with equality.} \takwacomment{You mean $\Delta$? I think $\Delta$ and $\Delta$ are the same, I will use $\Delta$ (with inequality) and $\Delta_+$ (with equality) from now on.}\Timcomment{Sounds good. I see it's changed} 
By strict feasibility, we have $\bm{\delta}>\bm{0}$. In addition, Pareto optimality implies that $\bm{\delta}$ lies on the Pareto frontier of $\Delta$. Thus, $\bm{\delta}\in \Delta^{n+1}_+$.

\vspace{2mm}

\noindent\textbf{Step 2: Showing }$\bm{f(\mathcal{E}(\bm{e}))=\argmax_{\bm{x}\in\mathcal{E}(\bm e)}\sum_{j\in\mathcal{I}}\delta_j\ln x_j}$\\
Take $\bm{e}\in\mathbb{R}^{n+1}_{++}$, define a set $\mathcal{E}(\bm{e}):=\{\bm{x}\in\mathbb{R}^{n+1}_+:\sum_{i=1}^{n+1}e_ix_i\leq1\}$ and a positive affine transformation $\hat{T}(\bm{x}):=(e_1x_1,\dots,e_{n+1}x_{n+1})$. Then, $\hat{T}(\mathcal{E}(\bm e))=\Delta^{n+1}$. By the invariance to positive affine transformations, we have $\hat{T}(f(\mathcal{E}(\bm{e})))=f(\Delta^{n+1})={\bm\delta}$. Hence, relabeling the index $n+1$ as $R$, we have $ f_j(\mathcal{E}(\bm{e}))=\frac{\delta_j}{e_j}$, $j\in \mathcal{I}$.
% \begin{equation*}
%     f_j(\mathcal{E}(\bm{e}))=\frac{\delta_j}{e_j}\quad\text{for all }j\in\mathcal{I},
% \end{equation*}

We show that $f(\mathcal{E}(\bm{e}))$ is the solution of the optimization problem $\max_{\bm{x}\in\mathcal{E}(\bm e)}\sum_{j\in\mathcal{I}}\delta_j\ln x_j$.
% \begin{equation*}
%     \max_{\bm{x}\in\mathcal{E}}\sum_{j\in\mathcal{I}}\delta_j\ln x_j.
% \end{equation*}
Indeed, consider the Lagrangian $\mathcal{L}:=\sum_{j\in\mathcal{I}}\delta_j\ln x_j+z(1-\sum_{j\in\mathcal{I}}e_jx_j)$, the first-order condition leads to $x_j=\delta_j/ze_j$ for all $j\in\mathcal{I}$. Using the fact that $\sum_{j\in\mathcal{I}}e_jx_j=1$ {and $\sum_{j\in\mathcal{I}}\delta_j=1$, we have} $z=1$ and thus $x_j=\delta_j/e_j$.
\vspace{2mm}

\noindent\textbf{Step 3: Showing }$\bm{f(\mathcal{X})=\argmax_{\bm{x}\in\mathcal{X}}\sum_{j\in\mathcal{I}}\delta_j\ln x_j}$ \\
Consider the optimization problem 
{\small\begin{equation}\label{Eq:auxNB}
    \max_{\bm{x}\in\mathcal{X}}\sum_{j\in\mathcal{I}}\delta_j\ln x_j
\end{equation}}%
and define $\bm{c}:=\argmax_{\bm{x}\in\mathcal{X}}\sum_{j\in\mathcal{I}}\delta_j\ln x_j>\bm{0}$. Note that $\bm{c}$ is unique because of the strict concavity of the objective function and the convexity of $\mathcal{X}$. Fix $\bm{x}\in\mathcal{X}$, and $t\in[0,1]$, the vector $\bm{x}(t):=\bm{c}+t(\bm{x}-\bm{c})\in\mathcal{X}$
% \begin{equation*}
%     \bm{x}(t):=\bm{c}+t(\bm{x}-\bm{c})\in\mathcal{X},\quad\text{for all } t\in[0,1],
% \end{equation*}
is feasible in Problem \eqref{Eq:auxNB} due to the convexity of $\mathcal{X}$. 

Define $\varphi(t):=\sum_{j\in\mathcal{I}}\delta_j\ln (c_j+t(x_j-c_j))$, the optimality of $\bm{c}$ implies that $\varphi'_+(0)\leq 0$. By direct calculation, we have $  \varphi'_+(0)=\sum_{j\in\mathcal{I}}\frac{\delta_j}{c_j}(x_j-c_j)\leq 0$ for all $x\in\mathcal{X}$,
% \begin{equation*}
%     \varphi'_+(0)=\sum_{j\in\mathcal{I}}\frac{\delta_j}{c_j}(x_j-c_j)\leq 0\quad\text{for all }\bm{x}\in\mathcal{X},
% \end{equation*}
and thus $  \sum_{j\in\mathcal{I}}\delta_jx_j/c_j\leq\sum_{j\in\mathcal{I}}\delta_j=1$.
% \begin{equation*}
%     \sum_{j\in\mathcal{I}}\frac{\delta_jx_j}{c_j}\leq\sum_{j\in\mathcal{I}}\delta_j=1
% \end{equation*}
Let $p_i:=\delta_i/c_i$ for all $i=1,\dots,n$ and $p_{n+1}:=\delta_R/c_R$, we have $\mathcal{X}\subset\mathcal{E}(\bm{p})$. According to Step 2, $f(\mathcal{E}(\bm{p}))=\bm{c}$, along with the independence of irrelevant alternatives, we obtain $f(\mathcal{X})=\bm{c}$, and the proof is complete. \qed

\section{Proof of Lemma \ref{Lem:POSai}}\label{App:POSai}

Assume, on the contrary, that $a_i^*=0$ and $a_k^*>0$ for two distinct peers $i,k\in\mathcal{N}$ and $i\neq k$. In that case, $W_i^*=w_i-(1-a_R^*+P^*/\mathbb{E}[S])\mu_i$ and $W_k^*=w_{{k}}-a_k^*(S-\mathbb{E}[S]){-}(1-a_R^*+P^*/\mathbb{E}[S])\mu_{{k}}$. 

Consider a contract $(\bm{a}^{\varepsilon},P^*)$ such that $a_i^{\varepsilon}=\varepsilon>0$ and $a_k^{\varepsilon}=a_k^*-\varepsilon\geq0$ with other variables, $P^*$, $a_R^*$ and $a_j^{\varepsilon}=a_j^*$ for all $j\in\mathcal{N}\backslash\{i,k\}$, remain unchanged from the optimal contract.
The new contract $(\bm{a}^{\varepsilon},P^*)$
%with unchanged $a_R^*$ 
clearly satisfies the market-clearing condition. We then show the strict feasibility of this contract and $\bm{a}^{\varepsilon}\in[0,1]^n$. Note that the IR of the P2P reinsurer remains strictly satisfied as $a_R^*$ and $P^*$ are not changed. 

Define $s_m:=V_m(\bm{a}^*,P^*)-d_m$ for all $m\in\mathcal{I}$.
%\KTHNcomment{$m\in\mathcal{N}$?} \takwacomment{I also use $s_R$ to restrict $\overline{\iota}$ to reach the strict feasibility below.}. 
As each $V_m$ is continuous at $(\bm{a}^*,P^*)$ due to Proposition \ref{pp:V:continuous}, for each $m\in\mathcal{N}$, there exists $r_m>0$ such that if $\|(\bm{a}^*,P^*)-(\bm{a},P)\|_{\infty}<r_m$, then $|V_m(\bm{a}^*,P^*)-V_m(\bm{a},P)|<s_m/2$. By choosing $0<\varepsilon<\min\{a_k^*/2,1/2,\min_{m\in\mathcal{N}}r_m\}$, we have $\bm{a}^{\varepsilon}\in[0,1]^n$, and $\|(\bm{a}^*,P^*)-(\bm{a}^{\varepsilon},P^*)\|_{\infty}=\varepsilon<r_m$ leads to $|V_m(\bm{a}^*,P^*)-V_m(\bm{a}^{\varepsilon},P^*)|<s_m/2$ for all $m\in\mathcal{N}$, and thus the strict feasibility follows.
% \takwacomment{Obviously, $\bm{a}^{\varepsilon}\in\Delta$ and $P_{\min}(\bm{a}^{\varepsilon})=P_{\min}(\bm{a}^*)$. However, $P_{\max}(\bm{a}^{\varepsilon})\leq P_{\max}(\bm{a}^*)$ so we need to choose sufficiently small $\varepsilon$ so that $P^*<P_{\max}(\bm{a}^{\varepsilon})$}
% \KTHNcomment{Also need to satisfies $P_{\min}$ $P_{\max}$ condition along the direction. }
% \takwacomment{I changed to new version to directly show the strict feasibility. }

Along the perturbation $\bm{a}^{\varepsilon}$,  define $\Phi(\varepsilon):=\mathcal{J}(\bm{a}^{\varepsilon},P^*)$ and consider its first derivative at $\varepsilon=0^+$, we obtain
{\small\begin{align*}
\Phi'(0^+)=&\ \frac{\delta_i}{V_i(\bm{a}^*,P^*)-d_i}\left.\frac{d}{d\varepsilon}V_i(\bm{a}^\varepsilon,P^*)\right|_{\varepsilon=0^+}+\frac{\delta_k}{V_k(\bm{a}^*,P^*)-d_k}\left.\frac{d}{d\varepsilon}V_k(\bm{a}^\varepsilon,P^*)\right|_{\varepsilon=0^+}\\
=&\ \frac{\delta_i}{V_i(\bm{a}^*,P^*)-d_i}(-U'_i(W_i^*)\mathbb{E}[S-\mathbb{E}[S]])+\frac{\delta_k}{V_k(\bm{a}^*,P^*)-d_k}\mathbb{E}[U'_k(W_k^*){(S-\mathbb{E}[S])} ]
\\
=&\ \frac{\delta_k}{V_k(\bm{a}^*,P^*)-d_k}
\operatorname{Cov}\!\left(S-\mathbb{E}[S],
U_k'\!\left(W_k^*\right)\right)\\
=&\ \frac{\delta_k}{V_k(\bm{a}^*,P^*)-d_k}
\operatorname{Cov}\!\left(S-\mathbb{E}[S],
U_k'\!\left(w_k-a_k^*(S-\mathbb{E}[S])
-\left(1-a_R^*+\frac{P^*}{\mathbb{E}[S]}\right)\mu_k\right)\right)>0,
\end{align*}}%
% \Timcomment{Shouldn't the denominator read $V_i-d_i$?} \takwacomment{I am not sure what you mean? As $E[S-E[S]]=0$ the related term disappear, there will be no $.../(V_i-d_i)$ term.}\Timcomment{What I meant is that I think we take a derivative of roughtly $\ln(V_i-d_i)$, and so the denominator in $\frac{\delta_i}{V_i(\bm{a}^{\varepsilon},P^*)}$ should be $V_i-d_i$ instead of $V_i$?} \takwacomment{I see it now, not sure how to make that mistake and I amend it now. }
where the covariance is positive since $U'_k(W_k^*)$ is strictly increasing in $S-\mathbb{E}[S]$ by the strict concavity of $U_k$, and due to the non-degeneracy of $S$.
% (Assumption \ref{Ass:non_degenerate_S}). 
This contradicts the optimality of $(\bm{a}^*,P^*)$. 
Since $\mathcal{K}$ is non-empty (see Lemma \ref{Lemma:SF}) and we are confined into the set $\bm{a}\in \Delta^n$, the fact that $a_i=0$  and $a_k>0$ for some $i\neq k\in \mathcal{N}$ is sub-optimal implies that either one of the two cases depicted in the statement could hold.

% \takwacomment{Note that Condition \eqref{Eq:supH} ensures $\bm{a}^*\in\Delta$ according to Lemma \ref{Lemma:SF} \KTHNcomment{Condition \eqref{Eq:supH} only ensures $\mathcal{K}$ is non-empty, it does not imply the solution of the global, unconstrained problem is such that $\bm{a}\in\Delta$. Instead, we can say: ``Since $\mathcal{K}$ is non-empty and we are confined into the set $\bm{a}\in \Delta$, the fact that $a_i=0$  and $a_k>0$ for some $i\neq k\in \mathcal{N}$ is sub-optimal implies that either one of the two cases depicted in the statement could hold.} %then ruling out $a_i=0$ for some $i\in\mathcal{N}$ results in either cases in the statement.
% } 
%\takwacomment{I changed it to strict concavity;otherwise, we don't have the strict sign in the inequality.}

It remains to show that, when $\tau=0$, the second alternative cannot have $a_R^*=0$. Suppose, to the contrary, that $a_R^*=0$. When it holds, we have $P_{\min}(\bm{a})=0$. Since $(\bm{a}^*,P^*)\in\mathcal{K}$, strict feasibility implies $P^*>P_{\min}(\bm{a})=0$. Hence, the constraint $P\geq 0$ is slack, and the first-order condition with respect to $P$ applies. For $\iota\in(0,1)$, we define the perturbed strategy $(\bm{a}^\iota,P^\iota)$ as follows. Let $P^\iota:=P^*>0$ remain unchanged; define  $  a_R^{\iota}=\iota\in(0,1)$, $a_i^{\iota}=(1-\iota)a_i^*\in(0,1)$ for $i\in\mathcal{N}$.
% \begin{equation*}
%     a_R^{\iota}=\iota\in(0,1),\quad 
%     a_i^{\iota}=(1-\iota)a_i^*\in(0,1)\quad\text{for all }i\in\mathcal{N}.
% \end{equation*}
Clearly, the market-clearing condition is satisfied. In addition, following the above continuity argument, one can show that there exists {$\overline{\iota}\leq  1\wedge\min_{m\in\mathcal{I}}r_m$}
%\takwacomment{$\overline{\iota}\in(1,\min_{m\in\mathcal{I}}r_m)$} 
%\KTHNcomment{$\iota\in(0,1)$ otherwise $a^\iota_i<0$; also note that we are using $(1-\iota)a$ instead of absolute change here.} \takwacomment{It's true because $\overline{\iota}$ is the upper bound of $\iota$. I added 0 as the lower bound below.} \KTHNcomment{How to derive the range of $\overline{\iota}$? And if $\overline{\iota}>1$, will this mean $\iota\in(0,1)$ in particular which is too wide?} 
%\takwacomment{Following the above logic, $||(\bm{a}^{\iota},P^*)-(\bm{a}^{\iota},P^*)||_{\infty}=\iota<r_m$ for all $m\in\mathcal{I}$, thus $||V(\bm{a}^{\iota},P^*)-V(\bm{a}^{\iota},P^*)||_{\infty}<s_m/2$ for all $m\in\mathcal{I}$ and the strict feasibility follows. (I made a typo that messed with $r_m$ and $s_m$.)}
such that $(\bm{a}^\iota,P^*)$ is strictly feasible whenever {$\iota \in (0, \overline{\iota})$. } 
%\KTHNcomment{I still don't see why $\overline{\iota} \in (1,....)$, as $\iota$ must always be $\leq 1$; the range of $\overline{\iota}$ looks weird. Did you mean $\overline{\iota}\in(0,...)?$  } \takwacomment{We can just say there exists $\iota\in(0,...)$ to avoid $\overline{\iota}$.}
%(I didn't check this although I believe it's true) 
%Also, we further restrict $\iota\in(1,\min_{m\in\mathcal{I}})$, following a similar argument as above, we can ensure the stricty feasibility.
 %\KTHNcomment{what does it mean, and how to ensure the perturbed contract is feasible?} \takwacomment{The proof strategy is to make up something feasible that improves the optimal contract so that contraction arises; $p$ is something I make up, and terms with $p$ will be canceled out in the last expression. I added some conditions for the feasibility.} \KTHNcomment{How to show $P_{\min}(\bm{a}^\iota)\leq  P^\iota \leq P_{\max}(\bm{a}^\iota)$ for all $\iota>0$ (or within some interval $(0,\varepsilon)$? I guess you will need to use some convexity/concavity of $P_{\min}$ and $P_{\max}$, which means one needs to consider convex combination of $\bm{a}$ and $h$ instead?; see my revised proof of the corollary below.}
Then, we have
{\small\begin{align*}
&W_R^{\iota}=w_R+P^*-\iota S,\quad W_i^{\iota}=w_i-a_i^{\iota}(S-\mathbb{E}[S])-\left(1-\iota+\frac{P^*}{\mathbb{E}[S]}\right)\mu_i,\quad\text{for all }i\in\mathcal{N},\\ 
&W_R^*=w_R+P^*,\quad W_i^*=w_i-a_i^*(S-\mathbb{E}[S])-\left(1+\frac{P^*}{\mathbb{E}[S]}\right)\mu_i,\quad\text{for all }i\in\mathcal{N}.
\end{align*}}%
Consider $\Psi(\iota):=\mathcal{J}(\bm{a}^{\iota},P^{\iota})$ and its first derivative at $\iota=0^+$, which is given by 
{\small\begin{equation*}
\Psi'(0^+)
=-\frac{\delta_RU'_R(W_R^*)\mathbb{E}[S]}{V_R(\bm{a}^*,P^*)-d_R}+\sum_{i=1}^n\frac{\delta_i\mathbb{E}[U_i'(W_i^*)(a_i^*(S-\mathbb{E}[S])+\mu_i)]}{V_i(\bm{a}^*,P^*)-d_i}.
\end{equation*}}%
{Given that $\bm{a}^*\in[0,1]^n$, we have $P_{\min}(\bm{a}^*)\geq 0$; see the discussion in Section \ref{sec:prelim}. Along with the fact that the Nash product is well-defined only if $P\geq P_{\min}(\bm{a}^*)$, the optimal premium $P^*$ must satisfy $P^*\geq P_{\min}(\bm{a}^*)\geq 0$. Hence, the constraint $P\geq 0$ is unbinding, and the FOC with respect to $P$ for  $(\bm{a}^*,P^*)$ reads}
%As $P^*>0$ and the constraint $P\geq0$ is not active, obtaining the FOC with respect to $P$ with the contract $(\bm{a}^*,P^*)$, we have \KTHNcomment{Not sure about this, since $P^*$ could have been obtained from the Lagrangian instead of the raw unconstrained FOC, where the former makes $P^*>0$, not the other way around (i.e., not because $P^*>0$, it solves the raw FOC). Although I think adding the Lagrange multipler still preserve the result by the slackness condition (please check the sign I might got it wrong)  }
%\takwacomment{I checked the maths and the sign, your version should be correct. My version is like Condition \eqref{Eq:supH} applies in this case as well, so $P^*>P_{\min}=0$; otherwise, $P^*=0$ can lead to arbitrary $\bm{a}^*$ such that $\sum_ia_i^*=1$. Somehow, this is trivial and maybe we should exclude this case at the begining.} \KTHNcomment{I guess given that $\bm{a}\in[0,1]^n$, $P_{\min}(\bm{a}) \geq 0$ by definition. As the Nash product is well-defined only if $P\in(P_{\min},P_{\max})$, the fact that $P\geq P_{\min}$ automatically guarantees $P\geq  0$ without having to introduce the Lagrange multiplier? If this argument is true, maybe we can reinforce this when we introduce $P_{\min}$.   }
%\takwacomment{The logic is basically like that, but we need $P^*>0$ to avoid introducing the multiplier. One simple way is just to stay the condition $P^*>0$ in Lemma \ref{Lem:POSai}.} \KTHNcomment{Why do we need $P^*>0$?} \takwacomment{Indeed, for this Lemma, we don't need it; for the following theorem, to reach a nicer FOC conditions on P, relaxing it is better.}
{\small\begin{equation*}
\frac{\delta_RU_R'(W_R^*)}{V_R(\bm{a}^*,P^*)-d_R}=\frac{1}{\mathbb{E}[S]}\sum_{i=1}^n\frac{\delta_i\mu_i\mathbb{E}[U_i'(W_i^*)]}{V_i(\bm{a}^*,P^*)-d_i},
\end{equation*}}%
%\KTHNcomment{where $\lambda_P\geq 0$ is a Lagrange multiplier that satisfies the complementary slackness condition $\lambda_PP^*=0$.} 
This leads to
{\small\begin{equation*}
    \Psi'(0^+)= 
    %\KTHNcomment{\lambda_P\mathbb{E}[S]}+
    \sum_{i=1}^n\frac{\delta_ia_i^*\mathbb{E}[U_i'(W_i^*)(S-\mathbb{E}[S])]}{V_i(\bm{a}^*,P^*)-d_i}=\sum_{i=1}^n\frac{\delta_ia_i^*\operatorname{Cov}(U_i'(W_i^*),S-\mathbb{E}[S])}{V_i(\bm{a}^*,P^*)-d_i}>0,
\end{equation*}}%
contradicting the optimality of $(\bm{a}^*,P^*)$. Thus, $a_R^*=0$ with $\bm{a}^*>0$ is sub-optimal. \qed

\section{Proof of Theorem \ref{Thm:EUsol}}\label{App:EUsol}

At boundary points of $\Delta^n$, all derivatives below are interpreted as feasible one-sided directional derivatives. These derivatives exist under Assumption \ref{Ass:EU}.

\vspace{2mm}
\noindent\textbf{Existence}: Note that the objective {function} in Problem \eqref{Prob:hetero} is continuous {in $(\bm{a},P)$} on a convex and compact set $\mathcal{F}$ (see Lemma \ref{Lem:ccF}). By the extreme value theorem, it attains a maximum on $\mathcal{F}$. Using the strict feasibility (see Lemma \ref{Lemma:SF}), the maximizer lies in {$\mathcal{K}$} and thus the welfare for each agent is improved.

\vspace{2mm}
\noindent\textbf{Uniqueness}: By Lemma \ref{Lem:ccF},  the feasible set of Problem \eqref{Prob:hetero} is convex.
% We first show that the Nash product is strictly concave in the strictly feasible region $\mathcal{K}$.  
% This 
Then, the uniqueness result follows if the log-transformation of the objective in Problem \eqref{Prob:hetero}, $\mathcal{J}(\bm{a},P)$ in \eqref{Eq:J},
is strictly concave in $\mathcal{K}$.
%{\color{red}TB: I think this statement is not true, but also not needed. We get uniqueness if the log-transformed objective is strictly concave...} \takwacomment{I see, the Nash product is not necessarily strictly concave. I changed to the version you suggested.}
    
To this end, we prove that along any two distinct strictly feasible
contracts, at least one peer's terminal wealth differs with positive
probability; the corresponding expected utility surplus is then strictly
concave along that segment.
    % \KTHNcomment{(It seems that you can show this to hold for all $i\in\mathcal{N}$ if I have not mistaken.)} \takwacomment{Yes, this is the minimal requirement as Tim requested (if I was not mistaken)} \KTHNcomment{So instead of saying at least one, we can say all?} \takwacomment{Yes, it seems "say all" is better when considering the log-transformed problem so that the objective is finite.}
    Take $(\bm{a},P),(\Tilde{\bm{a}},\Tilde{P})\in\mathcal{K}$, we show that {$\mathbb{P}(W_i(\bm{a},P) \neq W_i(\Tilde{\bm{a}},\Tilde{P}))>0$ for some $i\in\mathcal{N}$.} Assume the contrary, i.e., for all $i\in\mathcal{N}$, 
 {\small   \begin{equation*}
W_i(\bm{a},P)-W_i(\Tilde{\bm{a}},\Tilde{P})
=(\Tilde{a}_i-a_i)(S-\mathbb{E}[S])-\left(\frac{P-\Tilde{P}}{\mathbb{E}[S]}+\sum_{j=1}^n(a_j-\Tilde{a}_j)\right)\mu_i=0 \quad \mathbb{P}\text{-a.s.}.
\end{equation*}}%
Since $S$ is non-degenerate, this leads to $\Tilde{a}_i=a_i$ for all $i\in \mathcal{N}$. Substituting this into the displayed equality yields $P=\Tilde P$. Hence, $(\bm a,P)=(\Tilde{\bm a},\Tilde P)$, contradicting the assumption that the two contracts are distinct.

%Therefore, if two contracts are distinct, then at least one peer’s terminal wealth must differ with positive probability. 
%Hence, for that peer, the strict concavity of $U_i$ gives a strict Jensen inequality. That is enough to prove that at least one utility surplus term is strictly concave, and then $\mathcal{J}$ is strictly concave on $\mathcal{K}$.

    Let $0<\lambda<1$, {and $i\in \mathcal{N}$ be such that $\mathbb{P}(W_i(\bm{a},P) \neq W_i(\Tilde{\bm{a}},\Tilde{P}))>0$.}
    % \takwacomment{The subtle thing is that we did not define $\mathbb{P}$} \KTHNcomment{If we don't have $\mathbb{P}$ then where does $\mathbb{E}$ come from} \takwacomment{I mean shall we define the prob. space explicitly or like our last one, just implicitly assuming that there is only one law here?} \KTHNcomment{We can just define a probability space $(\Omega,\mathcal{G},\mathbb{P})$  with an expectation $\mathbb{E}$ at the very beginning ($\mathcal{F}$ was taken as the admissible set; or change it to $\mathcal{A}$ if needed)} \takwacomment{I checked Tim did this in his ASTIN paper as well, so we can some lines for that.} 
    Using the concavity of $U_j$, $j\in\mathcal{I}$, and Jensen's inequality, we obtain 
    % \Timcomment{For at least one agent, only?} \takwacomment{Yes, we only need one strictly concave U and the other can be just concave, then it's sufficient to assert that the log-transformation is strictly concave.}
  {\small  \begin{align*}
        &(\mathbb{E}[U_R(\lambda W_R(\bm{a},P)+(1-\lambda)W_R({\Tilde{\bm{a}}},\Tilde{P}))]-d_R) \geq  \lambda(\mathbb{E}[U_R(W_R(\bm{a},P))]-d_R)+(1-\lambda)(\mathbb{E}[U_R(W_R({\Tilde{\bm{a}}},\Tilde{P}))]-d_R),\\
        &\mathbb{E}\left[U_i\left(\lambda W_i(\bm{a},P)+(1-\lambda) W_i(\Tilde{\bm{a}},\Tilde{P})\right)\right]-d_i \geq \lambda\left(\mathbb{E}\left[U_i\left(W_i(\bm{a},P)\right)\right]-d_i\right)
        +(1-\lambda)\left(\mathbb{E}\left[U_i\left(W_i(\Tilde{\bm{a}},\Tilde{P})\right)\right]-d_i\right). 
    \end{align*}}%
    {Note that strict inequality holds for $i\in \mathcal{N}$ since $U_i$ is strictly concave and $\mathbb{P}(W_i(\bm{a},P) \neq W_i(\Tilde{\bm{a}},\Tilde{P}))>0$. }
   % where at least one inequality is strict.
    Along with the fact that $\ln(\cdot)$ is increasing and concave, we conclude that $\mathcal{J}$ is strictly concave on $\mathcal{K}$.

    Now, with $\mathcal{K}\neq\emptyset$, assume the contrary that there are two maximizers $(\bm{a},P),(\Tilde{\bm{a}},\Tilde{P})\in\mathcal{K}$ of Problem \eqref{Prob:hetero}, where $(\bm{a},P)\neq (\Tilde{\bm{a}},\Tilde{P})$. Note that {any maximizers of Problem \eqref{Prob:hetero}} must lie in $\mathcal{K}$;  otherwise any contract lying in { $\mathcal{F}\setminus\mathcal{K}$}
    % $\mathcal{F}\setminus\mathcal{F}$ 
    must lead to a zero Nash product, contradicting the existence of a strictly feasible point with positive product. By the convexity of $\mathcal{F}$ (Lemma \ref{Lem:ccF}), $\lambda(\bm{a},P)+(1-\lambda)(\Tilde{\bm{a}},\Tilde{P}) \in \mathcal{F}$.
    Moreover, since both $(\bm a,P)$ and $(\widetilde{\bm a},\widetilde P)$
belong to $\mathcal K$ and each $V_j$ is concave, all IR constraints remain
strict at their convex combination. Hence, $  \lambda(\bm{a},P)+(1-\lambda)(\widetilde{\bm{a}},\widetilde{P})
    \in \mathcal K .$
% \[
%     \lambda(\bm{a},P)+(1-\lambda)(\widetilde{\bm{a}},\widetilde{P})
%     \in \mathcal K .
% \]
 In addition, by the strict concavity of $\mathcal{J}$ {and the fact that  $(\bm{a},P)\neq (\Tilde{\bm{a}},\Tilde{P})$,} we have $        \mathcal{J}(\lambda(\bm{a},P)+(1-\lambda)(\Tilde{\bm{a}},\Tilde{P}))>\lambda\mathcal{J}(\bm{a},P)+(1-\lambda)\mathcal{J}(\Tilde{\bm{a}},\Tilde{P})=\mathcal{J}(\bm{a},P),
$
    % \begin{equation*}
    %     \mathcal{J}(\lambda(\bm{a},P)+(1-\lambda)(\Tilde{\bm{a}},\Tilde{P}))>\lambda\mathcal{J}(\bm{a},P)+(1-\lambda)\mathcal{J}(\Tilde{\bm{a}},\Tilde{P})=\mathcal{J}(\bm{a},P),
    % \end{equation*}
    contradicting the simultaneous optimality of $(\bm{a},P)$ and $(\Tilde{\bm{a}},\Tilde{P})$. 

\vspace{2mm}
\noindent\textbf{Sufficiency and necessity of optimality}: For the remaining statements, we first show the sufficiency of optimality for all cases, and tackle the necessity by cases. As $\mathcal{K}\neq\emptyset$ and $\mathcal{J}$ is only defined thereon 
% $P_{\min}(\bm{a}) < P < P_{\max}(\bm{a})$ 
with $P_{\min}(\bm{a})\geq 0$ for $\bm{a}\in \Delta^n$, we consider the Lagrangian that relaxes the inequality constraint $P\geq0$  and enforces only the constraints for $\bm{a}\in \Delta^n$:
{\small         \begin{equation}
        \label{eq:L2}
            \mathcal{L}(\bm{a},P,\bm{\lambda},\upsilon) := \mathcal{J}(\bm{a},P) + \sum_{i=1}^n \lambda_i a_i + \upsilon \left(1-\sum_{i=1}^na_i \right),
        \end{equation}}%
    where $\bm{\lambda} = (\lambda_i)_{i=1}^n$, and $\lambda_i,\upsilon\geq 0$ for all $i\in \mathcal{N}$.
Then, {any optimal contract} $(\bm{a}^*,P^*)$ {must} satisfy the first-order condition and the complementary slackness, i.e., 
{\small \begin{equation}
        \label{Eq:KKT}
        \begin{aligned}
                &\frac{\partial\mathcal{L}(\bm{a}^*,P^*,\bm{\lambda}^*,\upsilon^*)}{\partial a_i} = 0, \quad \lambda_i^*a_i^*=0,  \quad i\in \mathcal{N}; \quad \frac{\partial\mathcal{L}(\bm{a}^*,P^*,\bm{\lambda}^*,\upsilon^*)}{\partial P} = 0, \quad \upsilon^*\left(1-\sum_{i=1}^n a_i^* \right)=0,
        \end{aligned}
        \end{equation}}%
        where $\bm{\lambda}^*$ and $\upsilon^*$ are the corresponding optimal Lagrangian multipliers. Note that each case corresponds to a different value set of $\bm{\lambda}^*$ and $\upsilon^*$, shown in Table \ref{tab:kkt_regimes}.
        
        By the concavity of $\mathcal{J}$ and \eqref{Eq:KKT}, we obtain for any $(\bm{a},P)\in\mathcal{K}$,
       {\small  \begin{align*}
            \mathcal{J}(\bm{a},P)-
            \mathcal{J}(\bm{a}^*,P^*)&\leq\nabla\mathcal{J}(\bm{a}^*,P^*)\cdot((\bm{a},P)-(\bm{a}^*,P^*)) =\sum_{i=1}^n(\upsilon^*-\lambda_i^*)(a_i-a_i^*)=\upsilon^*\left(\sum_{i=1}^na_i-1\right)-\sum_{i=1}^n\lambda_i^*a_i\leq 0,
        \end{align*}}%
        and thus $(\bm{a}^*,P^*)$ maximizes $\mathcal{J}$ over $\mathcal{K}$. Recall that any contract in $\mathcal{F}\backslash\mathcal{K}$ has at least one binding IR constraint and therefore yields a zero Nash product, whereas $(\bm{a}^*,P^*)\in\mathcal{K}$ yields a strictly positive Nash product. Hence $(\bm{a}^*,P^*)$ is globally optimal over $\mathcal{F}$, and the sufficiency follows.

        \begin{table}[htbp]
\centering
\small
\begin{tabular}{lll}
\toprule
Regime & Constraint status & Lagrangian multipliers \\
\midrule
Partial reinsurance 
& 
$a_i^*>0$ for all $i\in\mathcal N$, and $\sum_{i=1}^n a_i^*<1$
&
$\lambda_i^*=0$ for all $i\in\mathcal N$, and $\upsilon^*=0$
\\[0.8em]

Full reinsurance
&
$a_i^*=0$ for all $i\in\mathcal N$, and $\sum_{i=1}^n a_i^*=0<1$
&
$\upsilon^*=0$ and $\lambda_i^*\ge 0$ 
\\[0.8em]

Zero reinsurance
&
$a_i^*>0$ for all $i\in \mathcal N$, and $\sum_{i=1}^n a_i^*=1$
&
$\lambda_i^*=0$ for all $i\in\mathcal N$ and $\upsilon^*\ge 0$
\\
\bottomrule
\end{tabular}
\caption{Active constraints and KKT multipliers in the three reinsurance regimes.}
\label{tab:kkt_regimes}
\end{table}

  %\KTHNcomment{In the sequel, we assume that the feasibility condition $P\in (P_{\min}(\bm{a}),P_{\max}(\bm{a}))$ is satisfied (so that $\mathcal{J}$ is well-defined), and introduce the Lagrangian that enforces the constraint $\bm{a}\in \Delta$:
    %     \begin{equation}
    %     \label{eq:L}
    %         \mathcal{L}(\bm{a},P,\bm{\lambda},\lambda) := \mathcal{J}(\bm{a},P) + \sum_{i=1}^n \lambda_i a_i + \lambda \left(1-\sum_{i=1}^na_i \right),
    %     \end{equation}
    % where $\bm{\lambda} = (\lambda_i)_{i=1}^n$, and $\lambda_i,\lambda\geq 0$ for all $i\in \mathcal{N}$. By the concavity of $\mathcal{J}$ and the Karush–Kuhn–Tucker (KKT) conditions, the optimal solution $(\bm{a}^*,P^*)$ and the Lagrange multiplier $(\bm{\lambda}^*,\lambda^*)$ must satisfy 
    %     \begin{equation}
    %     \label{eq:KKT}
    %     \begin{aligned}
    %             &\frac{\partial\mathcal{L}(\bm{a}^*,P^*,\bm{\lambda}^*,\lambda^*)}{\partial a_i} = 0, \ \lambda_i^*a_i^*=0,  \ i\in \mathcal{N}; \ \frac{\partial\mathcal{L}(\bm{a}^*,P^*,\bm{\lambda}^*,\lambda^*)}{\partial P} = 0, \ \lambda^*\left(1-\sum_{i=1}^n a_i^* \right)=0.
    %     \end{aligned}
    %     \end{equation}

    In what follows, we prove the necessity for each reinsurance regime.

    \vspace{2mm}
    \noindent\textbf{Partial reinsurance}: 
    % {Recall that the objective function is well-defined only when $P_{\min}(\bm{a})<P<P_{\max}(\bm{a})$. This condition is intrinsic to the domain of definition of the problem and is therefore implicitly enforced; it does not need to be imposed as an additional constraint in the Lagrangian formulation.}
    {By Lemma \ref{Lem:POSai}, if $0<\sum_{i=1}^na_i^*<1$, we must have $a^*_i>0$ for all $i\in \mathcal{N}$, and thus $\bm{a}^*\in \Delta^\circ := \{\bm{a}\in (0,1)^n:\sum_{i=1}^na_i<1\}$. We also have $P^*> P_{\min}(\bm{a}^*)> 0$; see Section \ref{sec:prelim}. This implies that $(\bm{a}^*,P^*)$ lies in the interior of the constraint set of contracts $\Delta^n\times [0,\infty)$, and thus all constraints are non-binding. By the complementary slackness condition in \eqref{Eq:KKT}, we have $\upsilon^*=0$ and $\lambda_i^*=0$ for all $i\in \mathcal{N}$.  Therefore, the KKT condition for the contract $(\bm{a}^*,P^*)$ is characterized by the following unconstrained FOC:}
%   For the interior solution, as no constraints are active, we have, by direction differentiations of $\mathcal{J}$ with respect to $P$ and $a_k$, $k\in\mathcal{N}$, we obtain  
 {\small    \begin{align*}
&\frac{\delta_R\mathbb{E}[S]\mathbb{E}[U_R'(W_R(\bm{a}^*,P^*))]}
{\mathbb{E}[U_R(W_R(\bm{a}^*,P^*))]-d_R}
=\sum_{i=1}^n
\frac{\delta_i\mu_i\mathbb{E}\!\left[U_i'(W_i(\bm{a}^*,P^*))\right]}
{\mathbb{E}[U_i(W_i(\bm{a}^*,P^*))]-d_i},\\
&\frac{\delta_R(1+\tau)\mathbb{E}[U_R'(W_R(\bm{a}^*,P^*))S]}
{\mathbb{E}[U_R(W_R(\bm{a}^*,P^*))]-d_R} =
\sum_{j=1}^n
\frac{\delta_j\mu_j\mathbb{E}\!\left[U_j'(W_j(\bm{a}^*,P^*))\right]}
{\mathbb{E}[U_j(W_j(\bm{a}^*,P^*))]-d_j} +
\frac{\delta_i\mathbb{E}\!\left[U_i'(W_i(\bm{a}^*,P^*))(S-\mathbb{E}[S])\right]}
{\mathbb{E}[U_i(W_i(\bm{a}^*,P^*))]-d_i}, \quad i\in \mathcal{N}.
\end{align*}}%
    This system is equivalent to \eqref{Eq:HeteroFOC}.

    \vspace{2mm}
    \noindent\textbf{Full reinsurance}: 
Suppose that $(\bm{0},P^*(\bm{0}))$ is globally optimal. {Since $\sum_{i=1}^na_i^* = 0$, by the complementary slackness condition in \eqref{Eq:KKT}, we have $\upsilon^*=0$. Consequently, the KKT condition is reduced to 
 {\small    \begin{equation*}
       0= \frac{\partial \mathcal{L}}{\partial P}(\bm{0},P^*(\bm{0}),\bm{\lambda}^*,0)  = \frac{\partial\mathcal{J}}{\partial P}(\bm{0},P^*(\bm{0})) \quad \text{and} \quad 0 = \frac{\partial\mathcal{J}}{\partial a_i}(\bm{0},P^*(\bm{0})) + \lambda_i^*, \quad i\in \mathcal{N}.
    \end{equation*}}%
The former is equivalent to  \eqref{eq:P*0:equation}; and the latter can be expressed as, for $i\in \mathcal{N}$,  
  {\small   \begin{align*}
        0 &=\lambda_i^*+  \frac{\delta_R(1+\tau)\mathbb{E}[U_R'(W_R(\bm{0},P^*(\bm{0})))S]}
{\mathbb{E}[U_R(W_R(\bm{0},P^*(\bm{0})))]-d_R}
-
\sum_{j=1}^n
\frac{\delta_j\mu_j U_j'(W_j(\bm{0},P^*(\bm{0}))) }
{U_j(W_j(\bm{0},P^*(\bm{0})))-d_j}  - \frac{\delta_i U'_i(W_i(\bm{0},P^*(\bm{0})))\mathbb{E}[S-\mathbb{E}[S]]}{U_i(W_i(\bm{0},P^*(\bm{0})))-d_i} \\
&= \lambda_i^* +  \frac{\delta_R(1+\tau) \mathbb{E}[U_R'(W_R(\bm{0},P^*(\bm{0})))S]}
{\mathbb{E}[U_R(W_R(\bm{0},P^*(\bm{0})))]-d_R}
-\sum_{j=1}^n
\frac{\delta_j\mu_j U_j'(W_j(\bm{0},P^*(\bm{0}))) }
{U_j(W_j(\bm{0},P^*(\bm{0})))-d_j}.
    \end{align*}}%
% where we have used \eqref{eq:P*0:equation} in the last line. 
Since $\lambda_i^*\geq 0$, we arrive at \eqref{Eq:a0NPD}. 
}

%%%%%%%%%%%% hideen pertubation argument
% Fix $k\in\mathcal{N}$ and consider a one-side perturbation on $\bm{0}$, i.e., $\bm{a}^{\varepsilon}$ such that $a_k^{\varepsilon}=\varepsilon$ and $a_i^{\varepsilon}=0$ for all $i\in\mathcal{N}\backslash\{k\}$. As $(\bm{0},P^*(\bm{0}))\in\mathcal{K}$, the same continuity argument as in the proof of Lemma \ref{Lem:POSai} (see Appendix \ref{App:POSai}) implies that $(\bm{a}^{\varepsilon},P^*(\bm{0}))\in\mathcal{K}$ for all sufficiently small $\varepsilon$. Then, the global optimality of $(\bm{0},P^*(\bm{0}))$ entails that 
% \begin{equation*}
% \mathcal{J}(\bm{a}^{\varepsilon},P^*(\bm{0}))\leq\mathcal{J}(\bm{0},P^*(\bm{0}))\Rightarrow\lim_{\varepsilon\downarrow0}\frac{\mathcal{J}(\bm{a}^{\varepsilon},P^*(\bm{0}))-\mathcal{J}(\bm{0},P^*(\bm{0}))}{\varepsilon}=\frac{\partial\mathcal{J}}{\partial a_k}(\bm{0},P^*(\bm{0}))\leq0
% \end{equation*}
% for any $k\in\mathcal{N}$ as $k$ is arbitrary, and this implies \eqref{Eq:a0NPD}. \\
%%%%%%%%%%%%%%%%%%%%%%%%

% \section{Proof of Corollary \ref{Cor:FR}}\label{App:FR}

% It follows directly from the statement that the full reinsurance contract is optimal iff every feasible directional derivative from $\bm{a}=\bm{0}$ is nonpositive due to the strict concavity of $\mathcal{J}$, and the latter is equivalent to \eqref{Eq:a0NPD}.

\vspace{2mm}
\noindent\textbf{Zero reinsurance}: Note that $\tau>0$ is necessary for the zero reinsurance; see Lemma \ref{Lem:POSai}. {By the same lemma, we also have $\bm{a}^*>0$, which implies the Lagrange multipliers $\bm{\lambda}$ in \eqref{eq:L2} satisfy $\lambda_i^*=0$  for all $i\in \mathcal{N}$. }
To proceed, with $\sum_{i=1}^na_i^*=1$, the KKT condition reads as, for  $i\in\mathcal{N}$,  
  {\small   \begin{align}
    \label{eq:ZeroReinCond0}
       0&\leq  \upsilon^* = \frac{\partial\mathcal{J}}{\partial a_i}\left(\bm{a}^*,P^*\right) = -\sum_{j=1}^n \frac{\delta_j\mu_j \mathbb{E}[U_j'(W_j(\bm{a}^*,P^*))]}{\mathbb{E}[U_j(W_j(\bm{a}^*,P^*))]-d_j} -\frac{\delta_i\mathbb{E}[U_i'(W_i(\bm{a}^*,P^*))(S-\mathbb{E}[S])]}{\mathbb{E}\left[U_i\left(W_i(\bm{a}^*,P^*) \right)\right]-d_i}  + \frac{\delta_R(1+\tau)\mathbb{E}[S]U_R'(w_R+P^*) }{U_R(w_R+P^*)-d_R},\nonumber\\
        0 &= \frac{\partial \mathcal{J}}{\partial P}\left(\bm{a}^*,P^*\right) = - \sum_{j=1}^n \frac{\delta_j\mu_j\mathbb{E}[U_j'(W_j(\bm{a}^*,P^*))]}{\mathbb{E}[S]\left(\mathbb{E}[U_j(W_j(\bm{a}^*,P^*))]-d_j\right)} + \frac{\delta_R U_R'(w_R+P^*)}{U_R(w_R+P^*)-d_R}.
    \end{align}}%
The second equation in \eqref{eq:ZeroReinCond0} yields the last equation of \eqref{Eq:ZeroReinCond}. Substituting the last equation of \eqref{Eq:ZeroReinCond} into the first $n$ expressions in \eqref{eq:ZeroReinCond0}, we have, for all $i\in \mathcal{N}$, 
   {\small  \begin{equation*}
     \widehat{\lambda}:=    \frac{\delta_R\tau\mathbb{E}[S]U_R'(w_R+P^*) }{U_R(w_R+P^*)-d_R} - \upsilon^* =  \frac{\delta_i\mathbb{E}[U_i'(W_i(\bm{a}^*,P^*))(S-\mathbb{E}[S]) ]}{\mathbb{E}\left[U_i\left(W_i(\bm{a}^*,P^*) \right)\right]-d_i} .
    \end{equation*}}%
By the assumption that $P^*<P_{\max}(\bm{a}^*)$, we have $\mathbb{E}\left[U_i\left(W_i(\bm{a}^*,P^*) \right)\right]-d_i>0$. In addition, with $\bm{a}^*>0$,  $W_i(\bm{a}^*,P^*)$ and $S-\mathbb{E}[S]$ are 
%(strictly) 
counter-monotonic.  Along with the strictly decreasing property of $U_i'(\cdot)$, the random variable $U_i'(W_i(\bm{a}^*,P^*))$ is a strictly increasing function of $S$. Since $S$ is non-degenerate and $a_i^*>0$, we have $\mathbb{E}[U_i'(W_i(\bm{a}^*,P^*))(S-\mathbb{E}[S])] >0.
$
% \[
% \mathbb{E}[U_i'(W_i(\bm{a}^*,P^*))(S-\mathbb{E}[S])] >0.
% \]
Hence, $\widehat{\lambda}>0$, which establishes the first $n$ equations in \eqref{Eq:ZeroReinCond}. Finally, as $\upsilon^*\geq 0$,   we arrive at \eqref{Eq:ZeroReinCond} since $        \frac{\delta_R\tau\mathbb{E}[S]U_R'(w_R+P^*) }{U_R(w_R+P^*)-d_R} = \widehat{\lambda} + \upsilon^* \geq \widehat{\lambda}. 
$ \qed

\section{Proof of Lemma \ref{lem:NoFair:strict:feasible}}\label{App:NoFair:strict:feasible}
Note that Assumption \ref{ass:mgf} is {necessary for} \eqref{Eq:IRCond},
and the strict IR conditions of the P2P reinsurer and all peers are, respectively, equivalent to the following: there exist $\bm{a}\in \Delta^{n}$ and $\bm{b}\in\mathbb{R}^n$ such that, with $a_R=1-\sum_{i=1}^n a_i$,
%\Timcomment{What is $\partial \Delta$?} \takwacomment{It seems replace function does not work, I changed it to $\Delta^{n+1}$.}
{\small \begin{equation}
\label{eq:IR:equiv:exp}
    C_R(a_R)<\sum_{i=1}^nb_i, \quad b_i<(1+\theta)\mu_i-C_i(a_i),\quad i=1,\dots,n. 
\end{equation}}%
%where 

{To show the existence of a triplet $(\bm{a},a_R,\bm{b})$} that satisfies \eqref{eq:IR:equiv:exp}, note that from Condition \eqref{Eq:IRCond}, there exists $(\bm{a},a_R)\in \Delta^{n}$ and $a_R=1-\sum_{i=1}^na_i$ such that $D:=(1+\theta)\mathbb{E}[S]-\sum_{i=1}^nC_i(a_i)-C_R(a_R)>0
$.
% \begin{equation*}
%     D:=(1+\theta)\mathbb{E}[S]-\sum_{i=1}^nC_i(a_i)-C_R(a_R)>0.
% \end{equation*}
Let $\epsilon:=D/2n>0$, and set $    b_i:=(1+\theta)\mu_i-C_i(a_i)-\epsilon<(1+\theta)\mu_i-C_i(a_i),
$
% \begin{equation*}
%     b_i:=(1+\theta)\mu_i-C_i(a_i)-\epsilon<(1+\theta)\mu_i-C_i(a_i),
% \end{equation*}
which ensures that each peer's welfare strictly improves. In addition, $    \sum_{i=1}^nb_i=(1+\theta)\mathbb{E}[S]-\sum_{i=1}^nC_i(a_i)-\frac{D}{2}=C_R(a_R)+\frac{D}{2}>C_R(a_R),
$
% \begin{equation*}
%     \sum_{i=1}^nb_i=(1+\theta)\mathbb{E}[S]-\sum_{i=1}^nC_i(a_i)-\frac{D}{2}=C_R(a_R)+\frac{D}{2}>C_R(a_R),
% \end{equation*}
and thus the P2P reinsurer's IR condition is also strictly satisfied. \qed

\section{Proof of Proposition \ref{Prop:NoFairExpSol}}\label{App:NoFairExpSol}
Assumption \ref{ass:mgf} ensures that $M_S$, $m_S$ and $m_S'$ are well-defined below. We first show that $m_S(\cdot)$ is strictly increasing in $(-\infty,\bar{t})$. Indeed, for any $t\in \mathbb{R}$, by the Cauchy-Schwarz inequality and the non-degeneracy of $S$,
% (see Assumption \ref{Ass:non_degenerate_S}), 
 {\small    \begin{equation}\label{Eq:dm_S}
        m_S'(t) = \frac{M_S(t)M_S''(t)-[M_S'(t)]^2}{M_S^2(t)} = \frac{\mathbb{E}[e^{tS}]\mathbb{E}[S^2e^{tS}] -\mathbb{E}^2[Se^{tS}]  }{M_S^2(t)} > 0. 
    \end{equation}}%

\noindent\textbf{Unique solution of the system \eqref{Eq:OptNFaRa}}: 
We begin by showing that the system \eqref{Eq:OptNFaRa} admits a solution {$(\bm{a}^*,a_R^*)\in(0,1)^n\times [0,1)$} if \eqref{eq:Cond:exp:no:fairness} holds. {To this end, define an axillary function $\bm{\varphi}:[0,\overline{\gamma}]\to \mathbb{R}$ by:
{\small\begin{equation*}
    \bm{\varphi}(y):=m_S(y)-(1+\tau)m_S\left(\gamma_R(1+\tau)\left(1-\frac{y}{\overline{\gamma}}
    \right)\right). 
\end{equation*}}
We shall show that $\bm{\varphi}$ admits a unique root $y^*\in(0,\overline{\gamma}]$; recall that $\overline{\gamma}$ was defined in Assumption \ref{ass:mgf}.

Using \eqref{Eq:dm_S}, $\bm{\varphi}'(y)=m_S'(y)+\gamma_R(1+\tau)^2m_S'(\gamma_R(1+\tau)(1-y/\overline{\gamma}))/\overline{\gamma}>0$, i.e., $\bm{\varphi}$ is strictly increasing, ensuring the uniqueness of root, if any. To show that $\bm{\varphi}$ admits a root in $[0,\overline{\gamma}]$, note that $\bm{\varphi}(0)=m_S(0)-(1+\tau)m_S(\gamma_R(1+\tau))<0$ and $\bm{\varphi}\left(\overline{\gamma}\right)=m_S\left(\overline{\gamma}\right)-(1+\tau)m_S(0)=m_S\left(\overline{\gamma}\right)-(1+\tau)\mathbb{E}[S]\geq0,$
% \begin{align*}
%     \bm{\varphi}(0)&=m_S(0)-(1+\tau)m_S(\gamma_R(1+\tau))<0, \     \bm{\varphi}\left(\overline{\gamma}\right)=m_S\left(\overline{\gamma}\right)-(1+\tau)m_S(0)=m_S\left(\overline{\gamma}\right)-(1+\tau)\mathbb{E}[S]\geq0,
% \end{align*}
where the first and second inequality follow from \eqref{Eq:dm_S} and \eqref{eq:Cond:exp:no:fairness}, respectively. 
%and the second inequality from \eqref{eq:Cond:exp:no:fairness}. 
Since $\bm{\varphi}$ is continuous, by the intermediate value theorem, $\bm{\varphi}$ admits a root $y^*\in(0,\overline{\gamma}]$. A solution of the system \eqref{Eq:OptNFaRa} is then given by 
% $a_R^*=1-y^*\sum_{i=1}^n\gamma_i^{-1}$ 
$a_R^*=1-y^*/\overline{\gamma}$},
and $a_i^*=y^*/\gamma_i$, $i\in\mathcal{N}$ with $(\bm{a}^*,a_R^*)\in (0,1)^n\times [0,1)$; indeed, $y^*\in\left(0,\overline{\gamma}\right]$ leads to $a_R^*\in\left[0,1\right)$; since $y^*>0$, $a_i^*>0$, $i\in\mathcal{N}$, and, since $\overline{\gamma}<\gamma_i$, $i\in\mathcal{N}$ by the definition of $\overline{\gamma}$, we have $a_i^*<1$ for all $i\in\mathcal{N}$. Indeed, any solution of \eqref{Eq:OptNFaRa} is characterized by $\bm{\varphi}$. If $(\bm{a}^*,a_R^*)$ is a solution to \eqref{Eq:OptNFaRa}, then using \eqref{Eq:dm_S} leads to $\gamma_ia_i^*=\gamma_ja_j^*$, $a_i^*$ must satisfy $\bm{\varphi}(\gamma_ia_i^*)=0$. Hence, $a_i^*$ is uniquely given by $a_i^*=y^*/\gamma_i$.

\noindent\textbf{Optimal contract $(\bm{a}^*,a_R^*,\bm{b}^*,P)$}:
% \KTHNcomment{For notational convenience, we replace $P$ and $a_R$ by $\sum_{i=1}^nb_i$ and $1-\sum_{i=1}^na_i$ in the objective function in Problem \eqref{Prob:auxNF}, respectively. Consider the unconstrained optimization problem $\max_{(\bm{a},\bm{b})\in \mathbb{R}^n\times \mathbb{R}^n } \widetilde{\mathcal{J}}(\bm{a},\bm{b})$, where 
%     \begin{equation*}
%         \widetilde{\mathcal{J}}(\bm{a},\bm{b}) := \delta_R\kn\left(1-\exp\left(\gamma_R\left(C_R(1-\sum_{i=1}^na_i \right) \right) \right)
%     \end{equation*}
% }
The objective in \eqref{Prob:auxNF} is concave on 
%the strictly feasible set 
$\widetilde{\mathcal{K}}$. 
% As each normalized utility surplus is a positive constant minus the expectation of an exponential of an affine function, it is concave in the decision variables.
Indeed, each normalized utility surplus is concave in the decision variables as it is a positive constant minus the expectation of an exponential of an affine function. 
Since $\log(\cdot)$ is increasing and concave, the logarithm of each positive surplus is concave. The feasible constraints are affine. Hence, the KKT conditions are sufficient for global optimality on $\widetilde{\mathcal{K}}$. As contracts in $\widetilde{\mathcal{F}}\backslash\widetilde{\mathcal{K}}$ yield a zero Nash product while $\widetilde{\mathcal{K}}\neq\emptyset$, a maximizer of the Nash product must lie in $\widetilde{\mathcal{K}}$.

Consider the Lagrangian {that enforces the market-clearing condition, $\sum_{i=1}^nb_i=P$, and $(\bm{a},a_R)\in[0,1]^{n+1}$}:
{\small \begin{align}\label{Eq:NFLag}
    \mathcal{L}(\bm{a},a_R,\bm{b},P)
    :=&\ \delta_R\ln\left(1-e^{\gamma_R(C_R(a_R)-P)}\right)+\sum_{i=1}^n\delta_i\ln\left(1-e^{\gamma_i(C_i(a_i)+b_i-(1+\theta)\mu_i)}\right) \nonumber \\
    &+\lambda_1\left(1-a_R-\sum_{i=1}^na_i\right)+\lambda_2\left(P-\sum_{i=1}^nb_i\right) {+ \sum_{i=1}^n\nu_i a_i + \nu_R a_R} ,
\end{align}}%
where $\lambda_1, \lambda_2 \in\mathbb{R}$ {and $\nu_R,\nu_i\geq 0$, $i\in\mathcal{N}$,} are Lagrangian multipliers. At this stage, we do not incorporate the constraint $P\geq 0$; rather, we show that the maximizer of \eqref{Eq:NFLag} automatically satisfies this constraint under Condition \eqref{eq:Cond:exp:no:fairness}. In addition, we show that $a_R^*=1$ is impossible under \eqref{Eq:IRCond}. 
%\takwacomment{and $a_R^*=1$ leads to a contraction and thus cannot be the case}.
%\takwacomment{and they are not active}.  }

Let 
{\small \begin{equation*}
    \Lambda_R:=   \frac{\delta_R\gamma_Re^{\gamma_R(C_R(a_R)-P)}}{1-e^{\gamma_R(C_R(a_R)-P)}},\quad {\Lambda_i}:= \frac{\delta_i\gamma_ie^{{\gamma_i}(b_i+C_i(a_i)-(1+\theta)\mu_i)}}{1-e^{\gamma_i(b_i+C_i(a_i)-(1+\theta)\mu_i)}},\quad i\in\mathcal{N}.
\end{equation*}}%
By noting  $C_R'(a_R)=(1+\tau)m_S(\gamma_R(1+\tau)a_R)$ and $C_i'(a_i)=m_S(\gamma_ia_i),\,i\in\mathcal{N}$, 
% \begin{equation*}
% C_R'(a_R)=(1+\tau)m_S(\gamma_R(1+\tau)a_R),\quad C_i'(a_i)=m_S(\gamma_ia_i),\quad i=1,\dots,n.
% \end{equation*}
the FOCs of \eqref{Eq:NFLag} read as
{\small \begin{align*}
    &\frac{\partial\mathcal{L}}{\partial a_R}= {-}\Lambda_RC_R'(a_R)-\lambda_1+ {\nu_R} =0,\quad \frac{\partial\mathcal{L}}{\partial a_i}={-}\Lambda_iC_i'(a_i)-\lambda_1 + {\nu_i} =0,\quad i\in\mathcal{N},\\
    &\frac{\partial\mathcal{L}}{\partial P}=\Lambda_R+\lambda_2=0,\quad \frac{\partial\mathcal{L}}{\partial {b}_i}={-}\Lambda_i-\lambda_2=0,\quad i\in\mathcal{N}.
\end{align*}}%
These imply $\Lambda_R C'_R(a_R) - {\nu_R} = \Lambda_iC'_i(a_i) -{\nu_i} $ and $\Lambda_R= \Lambda_i$ for all $i\in \mathcal{N}$, {leading to the system 
   {\small  \begin{equation}
    \label{eq:FOC:lagrange:NoFairness}
        \Lambda_R\left[(1+\tau)m_S(\gamma_R(1+\tau)a_R) - m_S(\gamma_ia_i) \right] = \nu_R-\nu_i, \ i\in \mathcal{N}.  
    \end{equation}}%
Since \eqref{Eq:OptNFaRa} admits a unique solution $(\bm{a}^*,a_R^*)\in {(0,1)^n\times[0,1)} $ by  \eqref{eq:Cond:exp:no:fairness}, we conclude that $(\bm{a}^*,a_R^*)$ is optimal with the Lagrange multipliers satisfying $\nu_R^*=\nu_i^*$ and $\nu_i^*a_i^*=\nu_R^*a_R^*=0$ for $i\in \mathcal{N}$. 
} {In particular}, when $\tau=0$, \eqref{eq:Cond:exp:no:fairness} holds since $m_S\left(\overline{\gamma}\right)\geq m_S(0)=\mathbb{E}[S],
$
% $m_S\left(1/{\sum_{i=1}^n\gamma_i^{-1}}\right)\geq m_S(0)=\mathbb{E}[S],
% $
% \[
% m_S\left(\frac{1}{\sum_{i=1}^n\gamma_i^{-1}}\right)\geq m_S(0)=\mathbb{E}[S],
% \]
and \eqref{Eq:tau0aRa} follows by solving \eqref{Eq:OptNFaRa} in this case.

% \KTHNcomment{[Indeed, there can be other characterization if $a_R^*=0$ that may not need to solve \eqref{Eq:OptNFaRa}, may need to check on this like the case with price fairness. The condition will likely be: $(1+\tau)m_S(0) \geq m_S(\gamma_ia_i^*)$ and there exists $\lambda >0$ such that $m_S(\gamma_ia_i^*)=\lambda$ for all $i$; please check. ]}
% \takwacomment{In this case, $a_i^*=(1/\gamma_i)/\sum_j\gamma_j^{-1}$ and we must have $(1+\tau)m_S(0) = m_S(\gamma_ia_i^*)$ due to Condition \eqref{eq:Cond:exp:no:fairness}, then we can integrate this case with the system  \eqref{Eq:OptNFaRa}.}\KTHNcomment{So equality in \eqref{Eq:OptNFaRa} needs to hold for that case, which is consistent with the case $G(0)=0$.  }\takwacomment{You mean \eqref{eq:Cond:exp:no:fairness}? I added some lines to remark this case; please uncomment the added lines if you see fit.}
% \begin{equation}\label{Eq:LambdaRi}
%     {\Lambda_R C'_R(a_R) = \Lambda_iC'_i(a_i) \quad\text{and}\quad }  \Lambda_R= \Lambda_i\quad \text{for all }i\in\mathcal{N},
% \end{equation}
% $\Lambda_R=\Lambda_i$ for all $i=1,\dots,n$, 
% leading to the system \eqref{Eq:OptNFaRa}, which, from the above proof, admits a unique solution $(\bm{a}^*,a_R^*)\in {(0,1)^n\times[0,1)} $ by Condition \eqref{eq:Cond:exp:no:fairness}.

We now establish \eqref{Eq:Optbi}. Set $q(P):=e^{-\gamma_R(C_R(a_R^*)-P)}$, the equations $\Lambda_i=\Lambda_R$, $i\in\mathcal{N}$, can be simplified as 
%it follows that \KTHNcomment{Can show more steps for the below? How to ensure there is a unique $q^*\in(0,1)$? } \takwacomment{I recalculate all of $b_i$ and it turns out they are wrong from the definition of $\Lambda_R$ $\Lambda_i$ but now they are fixed.}
{\small \begin{align*}
  %  &\frac{\chi_i}{\takwacomment{q(P)}-1}=\frac{1}{e^{-\gamma_i(C_i(a_i^*)+b_i-(1+\theta)\mu_i)}-1},\quad i=1,\dots,n,\\
    %\Rightarrow&
    e^{-\gamma_i(C_i(a_i^*)+b_i-(1+\theta)\mu_i)}=1+\frac{q(P)-1}{\chi_i}, \quad i\in\mathcal{N},
\end{align*}}%
which is equivalent to \eqref{Eq:Optbi} with $q^*=q(P)$. Summing \eqref{Eq:Optbi} over $i\in \mathcal{N}$ yields 
   {\small \begin{equation*}
         \sum_{i=1}^n b_i^*(P) =    (1+\theta)\mathbb{E}[S]-\sum_{i=1}^n\left(C_i(a_i^*)+\frac{1}{\gamma_i}\ln\left(1+\frac{q(P)-1}{\chi_i}\right)\right). 
    \end{equation*}}%

{As $P =C_R(a_R^*)+ \frac{\ln q(P)}{\gamma_R}$, we can solve $(\bm{b}^*,P^*)$ with $\sum_{i=1}^n b_i^*=P^*$ if and only if $\widetilde{G}:[1,\infty) \to \mathbb{R}$,
%defined by
  {\small   \begin{equation*}
        \widetilde{G}(x):= (1+\theta)\mathbb{E}[S]-\sum_{i=1}^n\left(C_i(a_i^*)+\frac{1}{\gamma_i}\ln\left(1+\frac{{x}-1}{\chi_i}\right)\right) -  C_R(a_R^*) - \frac{\ln x}{\gamma_R}
    \end{equation*}}%
admits a root. Since $\widetilde{G}$ is continuous, strictly decreasing in $x\in(1,+\infty)$}, with $\widetilde{G}(1)>0$ by Condition \eqref{Eq:IRCond}, and $\widetilde{G}(x)\to -\infty$ as $x\to+\infty$,
%and has different signs when $\takwacomment{x}=1$ and $\takwacomment{x}\rightarrow+\infty$ due to Condition \eqref{Eq:IRCond},
the claim follows and $\bm{b}^*$ is given by \eqref{Eq:Optbi}. {Note that} the constraint $P^*\geq0$ {is unbinding. Indeed,} 
%, $a_R\leq1$ and $a_i\geq0$, $i\in\mathcal{N}$, are not active.
by \eqref{Eq:Optbi}--\eqref{Eq:qstar} and the fact that $q^* > 1$, we have $P^*=\sum_{i=1}^n b_i^* = C_R(a_R^*) + \frac{\ln q^*}{\gamma_R}>0$.
% \takwacomment{In addition, the system \eqref{eq:Cond:exp:no:fairness} explicitly ruled out the following cases:
%     \begin{enumerate}
%         \item $a_i=0$ and $a_j>0$ for $i\neq j$ and $i,j\in\mathcal{N}$. This is because $m_S(\gamma_ia_i) = m_S(\gamma_ja_j) \rightarrow \mathbb{E}[S] = m_S(\gamma_j a_j) > m_S(0) = \mathbb{E}[S]$;
%         \item $a_R=1$: this is because $a_i=0$ for all $i\in\mathcal{N}$, and $\mathbb{E}[S] = m_S(0) = (1+\tau)m_S(\gamma_R(1+\tau)a_R) > (1+\tau)m_S(0)=(1+\tau)\mathbb{E}[S]$. 
%     \end{enumerate}
% So the only possibility is $a_i>0$ for all $i\in\mathcal{N}$ and $a_R<1$.}
% \takwacomment{In addition, 
% similar to Lemma \ref{Lem:POSai}, nonzero $a_i$ can be shown by considering the perturbation: Let $(\bm{a}^*,\bm{b}^*)$ with $a_j^*=0$ and $a_i^*>0$ for all $j\in\mathcal{N}\backslash\{j\}$. For $j\neq k$ and $j,k\in\mathcal{N}$, we consider, for sufficiently small $\epsilon$,}
% \begin{equation*}
% a_j^{\epsilon}=\epsilon>0,\quad b_j^{\epsilon}=b_j^*-\epsilon\mathbb{E}[S],\quad a_k^{\epsilon}=a_k^*-\epsilon>0,\quad b_k^{\epsilon}=b_k^*+\epsilon\mathbb{E}[S].
% \end{equation*}
%where $j\neq k$.

{Next, we verify that the IR constraints for all agents are all strictly satisfied, and it suffices to show that $(\bm{a}^*,\bm{b}^*)$ satisfies \eqref{eq:IR:equiv:exp}. Note that the IR constraint for the P2P reinsurer holds since $P^* = C_R(a_R^*) + \frac{\ln q^*}{\gamma_R}> C_R(a_R^*)$. By \eqref{Eq:Optbi} and the fact that $q^*>1$, we have $b_i^* < (1+\theta)\mu_i-C_i(a_i^*)$ for $i\in \mathcal{N}$, and \eqref{eq:IR:equiv:exp}  follows.   }

Then, we show that $a_R^*\neq 1$ under Condition \eqref{Eq:IRCond}. Assume the contrary that $(\bm{a}^*,a_R^*)=(\bm{0},1)$. By the complementary slackness condition, we must have $\nu_R^*=0$. Hence, \eqref{eq:FOC:lagrange:NoFairness} is reduced to $\Lambda_R[(1+\tau)m_S(\gamma_R(1+\tau)) - m_S(0)] = -\nu_i^*\leq 0$.  However, $(1+\tau)m_S(\gamma_R(1+\tau))>m_S(0)$, and by the definition of $\Lambda_R$ and $P^*=C_R(a_R^*)+\frac{\ln q^*}{\gamma_R}$, we have $\Lambda_R = \frac{\delta_R\gamma_R  }{q^*-1} >0$ since $q^*>1$. This leads to a contradiction and thus $a_R^*\neq 1$.

We remark that for $a_R^*=0$, we must have $(1+\tau)m_S(0) \geq m_S(\gamma_ia_i^*)$, and there exists $\iota >0$ such that $m_S(\gamma_ia_i^*)=\iota$ for all $i\in\mathcal{N}$. 
In this case, 
% \takwacomment{$a_i^*=\gamma_i^{-1}/\overline{\gamma}$} 
$a_i^*=\overline{\gamma}/\gamma_i$
% $a_i^*=(1/\gamma_i)/\sum_{j\in\mathcal{N}}\gamma_j^{-1}$ 
and we must have $(1+\tau)m_S(0) = m_S(\gamma_ia_i^*)$ due to Condition \eqref{eq:Cond:exp:no:fairness}, which is consistent with $\bm{\varphi}(\overline{\gamma})=0$ and the system \eqref{Eq:OptNFaRa}. \qed
%\KTHNcomment{This requires showing $\Lambda_R \geq 0$ first; maybe we can discuss $P$ and $\bm{b}^*$ first before we talk about the system for $\bm{a}$?} \takwacomment{OK, I put this as a final remark.}

\section{Proof of Proposition \ref{Prop:OptesSR}}\label{App:OptesSR}
Theorem \ref{Thm:EUsol} establishes the unique existence of $(\bm{a}^*,P^*)$.
Using \eqref{Eq:PricingFairness} and $\mathbb{E}[S]=n\mu$, Problem \eqref{Prob:hetero} reads:
{\small \begin{align}\label{Prob:auxhomo}
    \begin{split}
        &\max_{a_R\in[0,1],\,\bm{a}\in[0,1]^n,\,P\geq 0} {\delta_R \ln\left(\mathbb{E}[U_R(w_R+P-(1+\tau)a_RS)] -d_R  \right)   } \\
        &\quad \quad +\delta\sum_{i=1}^n\ln\left(\mathbb{E}\left[U\left(w-a_i(S- n\mu)-\frac{P}{n}-(1-a_R)\mu)\right]-d\right)\right) \quad 
        s.t.\quad  a_R+\sum_{i=1}^na_i=1.
    \end{split}
\end{align}}%
For fixed $(a_R,P)$, define $ g(a):=
    \ln\left(
    \mathbb{E}\left[
    U\left(w-a(S-n\mu)-\frac{P}{n}-(1-a_R)\mu\right)
    \right]-d
    \right).$
% {\small \[
%     g(a):=
%     \ln\left(
%     \mathbb{E}\left[
%     U\left(w-a(S-n\mu)-\frac{P}{n}-(1-a_R)\mu\right)
%     \right]-d
%     \right).
% \]}%
On the strictly feasible region, the term inside the logarithm is positive.
Moreover, it is concave in $a$ because $U$ is concave and the argument is
affine in $a$. Since $\ln(\cdot)$ is increasing and concave, $g$ is concave.
By applying Jensen's inequality,
{\small \begin{align*}
    &\ln\left(\mathbb{E}\left[U\left(w-\frac{(1-a_R)(S-n\mu)+P}{n}-(1-a_R)\mu\right)\right]-d\right)
    =\ln\left(\mathbb{E}\left[U\left(w-\frac{(S-n\mu){\sum_{i=1}^n} a_i + P}{n} -(1-a_R)\mu\right)\right]-d\right)\\
    &\geq  \frac{1}{n}\sum_{i=1}^n\ln\left(\mathbb{E}\left[U\left(w-a_i(S-n\mu)-\frac{P}{n}-(1-a_R)\mu\right)\right]-d\right). 
\end{align*}}%
Therefore, for each fixed $(a_R,P)$, the peer part of the objective is maximized by the equal-splitting choice $a_i=(1-a_R)/n$, $i\in \mathcal{N}$.
%  $a_1=a_2=\dots=a_n=(1-a_R)/n
% $.
% \begin{equation*}
% a_1=a_2=\dots=a_n=\frac{1-a_R}{n}.
% \end{equation*}
Therefore, it suffices to consider the equal-splitting rule \eqref{Eq:esSharing} under Assumption \ref{Ass:homo}. Since $\mu_i=\mu$ and $a_i^*=(1-a_R^*)/n$ for all $i\in\mathcal{N}$, the
price-fairness condition gives $b_i^*= (1+\rho(P^*,a_R^*))\mu-a_i^*\mathbb E[S]$, $i\in\mathcal{N}$.
% \[
%     b_i^*= (1+\rho(P^*,a_R^*))\mu-a_i^*\mathbb E[S],
%     \qquad i=1,\dots,n.
% \]
The right-hand side is independent of $i$, and as
$\sum_{i=1}^n b_i^*=P^*$, we obtain $b_i^*=P^*/n$ for all $i\in\mathcal{N}$. \qed
%The statement on $\bm{b}$ follows from {the equal-splitting rule of $\bm{a}$ and the price-fairness condition \eqref{Eq:PricingFairness}.

%direct calculation.

\section{Proof of Proposition \ref{Prop:ExpHomo}}\label{App:ExpHomo}

The objective of Problem \eqref{Prob:Homo1} is equivalent to $ \mathcal{J}(a_R,P):= (1-n\delta)\ln\left(1-e^{-\gamma_RP}M_S(\gamma_R(1+\tau)a_R)\right)+n\delta\ln\left(1-e^{\gamma\left(\frac{P}{n}-(1+\theta)\mu\right)}M_S\left(\frac{\gamma(1-a_R)}{n}\right)\right).$
% {\small \begin{equation*}
%      \mathcal{J}(a_R,P):= (1-n\delta)\ln\left(1-e^{-\gamma_RP}M_S(\gamma_R(1+\tau)a_R)\right)+n\delta\ln\left(1-e^{\gamma\left(\frac{P}{n}-(1+\theta)\mu\right)}M_S\left(\frac{\gamma(1-a_R)}{n}\right)\right).
% \end{equation*}}%
The FOC of $\mathcal{J}$ with respect to $a_R$ and $P$ leads to \eqref{Eq:ExpaREquation} and \eqref{Eq:HomoPFOC}. Hence, the proof is complete once we show the unique existence of a solution $(a_R^*,P^*)$ to \eqref{Eq:ExpaREquation}--\eqref{Eq:HomoPFOC} that satisfies $a_R^*\in[0,1)$ and $P^*\in( P_{\min}(a_R^*),P_{\max}(a_R^*))$.

Note that \eqref{Eq:ExpaREquation} is well-defined by the first condition
in the statement and is equivalent to the equation $\overline{H}'(a_R)=0$.
By 
% Assumption \ref{Ass:non_degenerate_S} and 
the positivity of
$m_S'$ on its domain; see \eqref{Eq:dm_S}, 
$\overline{H}'(a_R)$ is strictly decreasing. Also, $\overline{H}'(0)
= m_S(\gamma/n)-(1+\tau)n\mu \geq 0$ and $\overline{H}'(1)
= n\mu-(1+\tau)m_S(\gamma_R(1+\tau))<0.$
% \[
% \overline{H}'(0)
% = m_S(\gamma/n)-(1+\tau)n\mu \geq 0,
% \qquad
% \overline{H}'(1)
% = n\mu-(1+\tau)m_S(\gamma_R(1+\tau))<0.
% \]
Thus, by continuity and strict monotonicity, there exists a unique
$a_R^*\in[0,1)$ solving \eqref{Eq:ExpaREquation}, which is  the
unique maximizer of $\overline{H}$. Finally, define
\begin{equation*}
    Q(P)
:=\frac{(1-n\delta)\gamma_RM_S(\gamma_R(1+\tau)a_R^*)}
{e^{\gamma_RP}-M_S(\gamma_R(1+\tau)a_R^*)}  
-\frac{\delta\gamma M_S\left(\frac{\gamma(1-a_R^*)}{n}\right)}
{e^{\gamma\left((1+\theta)\mu-\frac{P}{n}\right)}
-M_S\left(\frac{\gamma(1-a_R^*)}{n}\right)} .
\end{equation*}
% \[
% \begin{aligned}
% Q(P)
% :={}&
% \frac{(1-n\delta)\gamma_RM_S(\gamma_R(1+\tau)a_R^*)}
% {e^{\gamma_RP}-M_S(\gamma_R(1+\tau)a_R^*)}  
% -\frac{\delta\gamma M_S\left(\frac{\gamma(1-a_R^*)}{n}\right)}
% {e^{\gamma\left((1+\theta)\mu-\frac{P}{n}\right)}
% -M_S\left(\frac{\gamma(1-a_R^*)}{n}\right)} .
% \end{aligned}
% \]
Then $Q(P)\to+\infty$ as
$P\downarrow P_{\min}(a_R^*)$, while $Q(P)\to-\infty$ as
$P\uparrow P_{\max}(a_R^*)$. Since $Q$ is continuous and strictly
decreasing on
$(P_{\min}(a_R^*),P_{\max}(a_R^*))$, the proof is complete. \qed

% Finally, define 
% \begin{equation*}
%     C(a_R):=m_S\left(\frac{\gamma(1-a_R)}{n}\right)-(1+\tau)m_S(\gamma_R(1+\tau)a_R),
% \end{equation*}
% we have
% \begin{equation*}
%     C(0)=m_S\left(\frac{\gamma}{n}\right)-(1+\tau)\mathbb{E}[S]>0\quad\text{and}\quad C(1)=\mathbb{E}[S]-(1+\tau)m_S(\gamma_R(1+\tau))<0,
% \end{equation*}
% where the second inequality follows from the fact that $\mathbb{E}[S]=m_S(0)$ and $m_S$ is strictly increasing. 
% So, $a_R^*\in(0,1)$ follows from the intermediate value theorem. 

\section{Proof of Lemma \ref{Lem:BoundPsiM}}\label{App:BoundPsiM}
The stated conditions ensure that $\hat{u}_j^{\mathcal{N}}$,
$j\in\mathcal{I}$, are well-defined, enabling the subsequent analysis.
As $U_R$ is strictly concave, applying Jensen's inequality on \eqref{Eq:PminM} yields {\small$$\hat{u}_R^{\mathcal{N}}\leq U_R\left(w_R+\hat{P}_{\min}^{\mathcal{M}}\left(\bm{a}^{\mathcal{M}}\right)-(1+\tau)\left(1-\sum_{j\in \mathcal{M}}a_j^{\mathcal{M}}\right)\mathbb{E}[S_{\mathcal{M}}]\right).$$}
% {\small \begin{equation*}
%     \hat{u}_R^{\mathcal{N}}\leq U_R\left(w_R+\hat{P}_{\min}^{\mathcal{M}}\left(\bm{a}^{\mathcal{M}}\right)-(1+\tau)\left(1-\sum_{j\in \mathcal{M}}a_j^{\mathcal{M}}\right)\mathbb{E}[S_{\mathcal{M}}]\right).
% \end{equation*}}%
Since $U_R$ is strictly increasing, the above inequality and \eqref{Eq:Pi} imply $\hat{P}_{\min}^{\mathcal{M}}(\bm{a}^{\mathcal{M}})\geq \Pi+(1+\tau)(1-\sum_{j\in \mathcal{M}}a_j^{\mathcal{M}})\mathbb{E}[S_{\mathcal{M}}].$
% {\small \begin{equation}\label{Eq:PminMINE}
% \hat{P}_{\min}^{\mathcal{M}}(\bm{a}^{\mathcal{M}})\geq \Pi+(1+\tau)\left(1-\sum_{j\in \mathcal{M}}a_j^{\mathcal{M}}\right)\mathbb{E}[S_{\mathcal{M}}].
% \end{equation}}%

Using the above bound, the wealth of Peer $i$ in subgroup $\mathcal{M}$ at the premium level $\hat{P}_{\min}^{\mathcal{M}}(\bm{a}^{\mathcal{M}})$ satisfies
{\small \begin{align*}
W_i^{\mathcal{M}}(\bm{a}^{\mathcal{M}},\hat{P}_{\min}^{\mathcal{M}}(\bm{a}^{\mathcal{M}}))=&\ w_i-a_i^{\mathcal{M}}(S_{\mathcal{M}}-\mathbb{E}[S_{\mathcal{M}}])-\frac{\mu_i}{\mathbb{E}[S_{\mathcal{M}}]}\hat{P}_{\min}^{\mathcal{M}}(\bm{a}^{\mathcal{M}})-\sum_{j\in \mathcal{M}}a_j^{\mathcal{M}}\mu_i\\
% \end{equation*}}%
% Using \eqref{Eq:PminMINE}, we have
% {\small \begin{align*}
%W_i^{\mathcal{M}}(\bm{a}^{\mathcal{M}},\hat{P}_{\min}^{\mathcal{M}}(\bm{a}^{\mathcal{M}}))
\leq \ &w_i-a_i^{\mathcal{M}}(S_{\mathcal{M}}-\mathbb{E}[S_{\mathcal{M}}])-\frac{\mu_i}{\mathbb{E}[S_{\mathcal{M}}]}\Pi-\left[(1+\tau)\left(1-\sum_{j\in \mathcal{M}}a_j^{\mathcal{M}}\right)+\sum_{j\in \mathcal{M}}a_j^{\mathcal{M}}\right]\mu_i\\
\leq&\ w_i-a_i^{\mathcal{M}}(S_{\mathcal{M}}-\mathbb{E}[S_{\mathcal{M}}])-\mu_i-\frac{\mu_i}{{\mathbb{E}[S_\mathcal{M}]}}\Pi,
%\leq&\ w_i-a_i^{\mathcal{M}}(S_{\mathcal{M}}-\mathbb{E}[S_{\mathcal{M}}])-\mu_i-\frac{\mu_i}{\mathbb{E}[S]-\min_{j\in \mathcal{N}\backslash\{i\}}\mu_j}\Pi,
\end{align*}}%
where the last inequality follows from $(1+\tau)(1-\sum_{j\in \mathcal{M}}a_j^{\mathcal{M}})+\sum_{j\in \mathcal{M}}a_j^{\mathcal{M}}\geq1$.
%, and \KTHNcomment{(why don't we take a potentially tighter bound) $\mathbb{E}[S_{\mathcal{M}}] = \mathbb{E}[S] - \sum_{j\not\in\mathcal{M}}\mu_j \leq \mathbb{E}[S] - |\mathcal{N}\setminus \mathcal{M}|\min_{j\not\in \mathcal{M} }\mu_j$?}
%$\mathbb{E}[S_{\mathcal{M}}]\leq\max_{i\in \mathcal{M},\mathcal{M}\subsetneq\mathcal{N}}\mathbb{E}[S]=\mathbb{E}[S]-\min_{j\in \mathcal{N}\backslash\{i\}}\mu_j$. 
%\KTHNcomment{Is this bound correct? For example, consider $\mathcal{M}=\{i\}$, then $\mathbb{E}[S_\mathcal{M}] = \mu_i$ }
As $U_i$ is strictly increasing,
%we obtain
{\small \begin{equation*}
\hat{\Psi}_i^{\mathcal{M}}(\bm{a}^{\mathcal{M}})\leq\mathbb{E}\left[U_i\left(w_i-a_i^{\mathcal{M}}(S_{\mathcal{M}}-\mathbb{E}[S_{\mathcal{M}}])-\mu_i-\frac{\mu_i}{{\mathbb{E}[S_\mathcal{M}]}}\Pi\right)\right]\leq   U_i\left(w_i-\mu_i-\frac{\mu_i}{{\mathbb{E}[S_\mathcal{M}]}}\Pi\right),
\end{equation*}}%
where we have used Jensen's inequality.  
%and noting that $\mathbb{E}[S_{\mathcal{M}}-\mathbb{E}[S_{\mathcal{M}}]]=0$
%we have $\hat{\Psi}_i^{\mathcal{M}}(\bm{a}^{\mathcal{M}})\leq U_i\left(w_i-\mu_i-\frac{\mu_i}{{\mathbb{E}[S_\mathcal{M}]}}\Pi\right).$
% {\small \begin{equation*}
% \hat{\Psi}_i^{\mathcal{M}}(\bm{a}^{\mathcal{M}})\leq U_i\left(w_i-\mu_i-\frac{\mu_i}{{\mathbb{E}[S_\mathcal{M}]}}\Pi\right).
% \end{equation*}}%
% \begin{equation*}
% \hat{\Psi}_i^{\mathcal{M}}(\bm{a}^{\mathcal{M}})\leq\mathbb{E}\left[U_i\left(w_i-a_i^{\mathcal{M}}(S_{\mathcal{M}}-\mathbb{E}[S_{\mathcal{M}}])-\mu_i-\frac{\mu_i}{\mathbb{E}[S]-\min_{j\in \mathcal{N}\backslash\{i\}}\mu_j}\Pi\right)\right].
% \end{equation*}
% Adopting Jensen's inequality again and note that $\mathbb{E}[S_{\mathcal{M}}-\mathbb{E}[S_{\mathcal{M}}]]=0$, we have
% \begin{equation*}
% \hat{\Psi}_i^{\mathcal{M}}(\bm{a}^{\mathcal{M}})\leq \mathbb{E}\left[U_i\left(w_i-\mu_i-\frac{\mu_i}{\mathbb{E}[S]-\min_{j\in \mathcal{N}\backslash\{i\}}\mu_j}\Pi\right)\right].
% \end{equation*}
% where the strict sign follows from $U_i$ is strictly concave and $S_{\mathcal{M}}$ is non-degenerate (see Assumption \ref{Ass:SM}). 
% {\color{red}TB: Do we have that $a_i^Mneq 0$?} \takwacomment{As $\bm{a}^{\mathcal{M}}\in\Delta^{|\mathcal{M}|}$, it does not preclude $a_i^{\mathcal{M}}=0$. I see point that that does not validate the strict inequality. Now, I replace the strict sign and move it to Condition \ref{Eq:varthetaUpperBound}.}
Taking the supremum over $\bm{a}^{\mathcal{M}} \in \Delta^{|\mathcal{M}|}$ yields \eqref{Eq:BoundPsiM}. \qed

\section{Proof of Theorem \ref{Thm:core}}\label{App:core}
{For any $i\in \mathcal{N}$,  $\max_{\mathcal{M}\subsetneq \mathcal{N}, i\in\mathcal{M} }\mathbb{E}[S_\mathcal{M}] \leq  \mathbb{E}[S]- \min_{j \in \mathcal{N}\backslash \{i\} } \mu_j$.} 
%\takwacomment{The left hand side is to maximize over all subgroups but the right hand side is to minimize within a particular subgroup $M$?}
%\takwacomment{This seems strange; $\mu_1=1,\mu_2=2,\mu_3=3$, for peer $3$, the largest should be 5 instead of 3 from your formula.} 
By Lemma \ref{Lem:BoundPsiM},  \eqref{Eq:vartheta} and \eqref{Eq:varthetaUpperBound}, we have, for any   $\mathcal{M}\subsetneq\mathcal{N}$ and any peer $i\in \mathcal{M}$, $    \sup_{\bm{a}^{\mathcal{M}}\in\Delta^{|\mathcal{M}|}}\hat{\Psi}_i^{\mathcal{M}}(\bm{a}^{\mathcal{M}}) < U_i\left( w_i - (1+\vartheta_i)\mu_i\right) = \hat{u}_i^{\mathcal{N}}$ 
% \begin{equation*}
%     \sup_{\bm{a}^{\mathcal{M}}\in\Delta^{|\mathcal{M}|}}\hat{\Psi}_i^{\mathcal{M}}(\bm{a}^{\mathcal{M}}) < U_i\left( w_i - (1+\vartheta_i)\mu_i\right) = \hat{u}_i^{\mathcal{N}}
% \end{equation*}
. If there is no feasible contract $(\bm{a}^{\mathcal{M}},P^{\mathcal{M}})$ that weakly improves upon the P2P reinsurer's welfare in $\mathcal{M}$ (i.e., $\bm{a}^{\mathcal{M}} \in \Delta^{|\mathcal{M}|}$ and $P^{\mathcal{M}}\geq \hat{P}_{\min}^{\mathcal{M}}(\bm{a}^{\mathcal{M}})$ so that $\mathbb{E}[U_R(W^{\mathcal{M}}_R(\bm{a}^{\mathcal{M}},P^{\mathcal{M}}))] \geq \hat{u}^{\mathcal{N}}_R$), then $\mathcal{M}$ is not incentivized to be formed.
%then peers in $\mathcal{M}$ have no incentive to form a subgroup.

Assume that there exists a feasible contract $(\bm{a}^{\mathcal{M}},P^{\mathcal{M}})$ {with $\bm{a}^{\mathcal{M}}\in\Delta^{|\mathcal{M}|}$ and $P^{\mathcal{M}}\geq \hat{P}_{\min}^{\mathcal{M}}(\bm{a}^{\mathcal{M}})$,}
so that $ \mathbb{E}[U_i(W_i^{\mathcal{M}}(\bm{a}^{\mathcal{M}},P^{\mathcal{M}}))]
    \leq \hat{\Psi}_i^{\mathcal{M}}(\bm{a}^{\mathcal{M}})\leq \sup_{\bm{a}^{\mathcal{M}}\in\Delta^{|\mathcal{M}|}}\hat{\Psi}_i^{\mathcal{M}}(\bm{a}^{\mathcal{M}})<\hat{u}_i^{\mathcal{N}},$
% {\small \begin{equation*}
%     \mathbb{E}[U_i(W_i^{\mathcal{M}}(\bm{a}^{\mathcal{M}},P^{\mathcal{M}}))]
%     \leq \hat{\Psi}_i^{\mathcal{M}}(\bm{a}^{\mathcal{M}})\leq \sup_{\bm{a}^{\mathcal{M}}\in\Delta^{|\mathcal{M}|}}\hat{\Psi}_i^{\mathcal{M}}(\bm{a}^{\mathcal{M}})<\hat{u}_i^{\mathcal{N}}, 
% \end{equation*}}%
i.e., all peers in the subgroup $\mathcal{M}$ are strictly worse off than the welfare achieved under the grand coalition $\mathcal{N}$. This proves that no viable subgroup formation exists under Condition \eqref{Eq:varthetaUpperBound}.  The last statement follows from \eqref{Eq:smallvartheta}  {and the fact that $\mathbb{E}[S]- \min_{j \in\mathcal{N}\backslash \{i\} } \mu_j \leq \mathbb{E}[S]-\min_{j\in\mathcal{N}}\mu_j$} for all $i\in\mathcal{N}$. 
%noting that the right-hand side of \eqref{Eq:thetaCond2} is equal to $\min_{i\in \mathcal{N}}(\Pi/(\mathbb{E}[S]-\min_{j\in \mathcal{N}\backslash\{i\}}\mu_j))$, and the proof is complete.

% \section{Proof of Corollary \ref{Cor:thetaCond}}\label{App:thetaCond}

% From the strict improvement of Peer $i$'s welfare, we have
% \begin{equation*}
%     U_i^{-1}(\hat{u}_i^{\mathcal{N}})>w_i-(1+\theta)\mu_i.
% \end{equation*}
% Then, \eqref{Eq:varthetaUpperBound} follows if 
% \begin{equation*}
%     w_i-(1+\theta)\mu_i>w_i-\mu_i-\frac{\mu_i}{\mathbb{E}[S]}(U_R^{-1}(\hat{u}_R^{\mathcal{N}})-w_R)
% \end{equation*}
% is checked, which is equivalent to \eqref{Eq:thetaCond2} and thus the result follows. 

\section{Proof of Proposition \ref{Prop:core2}}\label{App:core2}
Proposition \ref{Prop:NoFairExpSol} applies, and hence the unique optimal
contract $(\bm{a}^*,a_R^*,\bm{b}^*,P^*)$ exists. According to Definition \ref{Def:subgroupformationNG}, a subgroup $\mathcal{M}\subsetneq\mathcal{N}$ can be viable only if  
{\small \begin{equation}\label{Eq:PMlowerbound}
    \sum_{i\in \mathcal{M}}b_i^{\mathcal{M}}\geq \Tilde{\Pi}+\frac{1}{\gamma_R}\ln\mathbb{E}\left[e^{\gamma_R(1+\tau)(1-\sum_{i\in \mathcal{M}}a_i^{\mathcal{M}})S_{\mathcal{M}}}\right], \quad   b_i^{\mathcal{M}}\leq (1+\Tilde{\vartheta}_i)\mu_i-\frac{1}{\gamma_i}\ln\mathbb{E}\left[e^{\gamma_ia_i^{\mathcal{M}}S_{\mathcal{M}}}\right], \ i\in \mathcal{M}.
\end{equation}}%
% and for all $i\in \mathcal{M}$,
% {\small \begin{equation}\label{Eq:biM}
%     b_i^{\mathcal{M}}\leq (1+\Tilde{\vartheta}_i)\mu_i-\frac{1}{\gamma_i}\ln\mathbb{E}\left[e^{\gamma_ia_i^{\mathcal{M}}S_{\mathcal{M}}}\right].
% \end{equation}}%
From \eqref{Eq:PMlowerbound}, we have the following, which gives an upper bound on the aggregate premium for subgroup $\mathcal{M}$ that keeps all peers in $\mathcal{M}$ weakly better off than in the grand coalition $\mathcal{N}$:
{\small \begin{equation}\label{Eq:PMupperbound}
    \sum_{i\in \mathcal{M}}b_i^{\mathcal{M}}\leq \sum_{i\in \mathcal{M}}\left((1+\Tilde{\vartheta}_i)\mu_i-\frac{1}{\gamma_i}\ln\mathbb{E}[e^{\gamma_ia_i^{\mathcal{M}}S_{\mathcal{M}}}]\right). 
\end{equation}}%
%which signifies 

Combining \eqref{Eq:PMlowerbound} and \eqref{Eq:PMupperbound}, $\mathcal{M}$ can be formed only if
{\small \begin{equation*}
    \frac{1}{\gamma_R}\ln\mathbb{E}[e^{\gamma_R(1+\tau)(1-\sum_{i\in \mathcal{M}}a_i^{\mathcal{M}})S_{\mathcal{M}}}]+\sum_{i\in \mathcal{M}}\frac{1}{\gamma_i}\ln\mathbb{E}[e^{\gamma_ia_i^{\mathcal{M}}S_{\mathcal{M}}}]\leq\sum_{i\in \mathcal{M}}\left((1+\Tilde{\vartheta}_i)\mu_i\right)-\Tilde{\Pi},
\end{equation*}}%
and applying Jensen's inequality leads to $    \sum_{i\in \mathcal{M}}((1+\Tilde{\vartheta}_i)\mu_i)-\Tilde{\Pi}\geq (1+\tau)(1-\sum_{i\in \mathcal{M}}a_i^{\mathcal{M}})\mathbb{E}[S_{\mathcal{M}}]+\sum_{i\in \mathcal{M}}a_i^{\mathcal{M}}\mathbb{E}[S_{\mathcal{M}}]$ $\geq \mathbb{E}[S_{\mathcal{M}}]. 
$
% {\small \begin{equation*}
%     \sum_{i\in \mathcal{M}}\left((1+\Tilde{\vartheta}_i)\mu_i\right)-\Tilde{\Pi}\geq (1+\tau)\left(1-\sum_{i\in \mathcal{M}}a_i^{\mathcal{M}}\right)\mathbb{E}[S_{\mathcal{M}}]+\sum_{i\in \mathcal{M}}a_i^{\mathcal{M}}\mathbb{E}[S_{\mathcal{M}}]\geq \mathbb{E}[S_{\mathcal{M}}]. 
% \end{equation*}}%
This implies that if $\sum_{i\in \mathcal{M}}\Tilde{\vartheta}_i\mu_i<\Tilde{\Pi}$, then
% \begin{equation}\label{Eq:PsiBoundM}
%     \sum_{i\in \mathcal{M}}\Tilde{\vartheta}_i\mu_i<\Tilde{\Pi},
% \end{equation}
 $\mathcal{M}$ is not incentivized to be formed. The final statement holds by noting that $\Tilde{\vartheta}_i < \theta$ and $\max_{\emptyset\neq \mathcal{M}\subsetneq \mathcal{N}} \sum_{i\in\mathcal{M}}\Tilde{\vartheta}_i\mu_i \leq \theta \mathbb{E}[S]$. \qed 
 
 %Then, note that the R.H.S of \eqref{Eq:PsiBound} is $\max_{\emptyset\neq \mathcal{M}\subsetneq\mathcal{N}}\sum_{i\in \mathcal{M}}\Tilde{\vartheta}_i\mu_i$ and thus imposing Condition \eqref{Eq:PsiBound}
% ensure that \eqref{Eq:PsiBoundM} is true for any possible subgroup $\mathcal{M}\subsetneq\mathcal{N}$, and the proof is complete.

\end{document}